\documentclass[onefignum,onetabnum]{siamart190516}

\usepackage{tikz}
\usepackage{pgfplots}
\usepackage{adjustbox}
\usepackage[utf8]{inputenc}
\pgfplotsset{compat=newest}
\usepgfplotslibrary{groupplots}
\usepgfplotslibrary{dateplot}
\usepackage{subcaption}

\usepackage{amssymb}
\usepackage{amsmath}
\usepackage{mathtools}
\usepackage{cleveref}
\usepackage{cite}

\newcommand{\imax}{i^{\rm{max}}}

\newcommand{\kmax}{k^{\rm{max}}}
\newcommand{\veps}{\varepsilon}
\newcommand{\p}{\partial}

\newcommand{\R}{\mathbb{R}}

\newcommand{\cP}{\mathcal{P}}
\newcommand{\bbT}{\mathbb{T}}

\newcommand{\bi}{\mathbf{i}}
\newcommand{\be}{\mathbf{e}}

\newtheorem{example}{Example}%[section]{Example}

\crefname{equation}{}{}
\crefname{theorem}{Theorem}{Theorems}
\crefname{proposition}{Proposition}{Propositions}
\crefname{problem}{Problem}{Problems}
\crefname{lemma}{Lemma}{Lemmas}
\crefname{definition}{Definition}{Definitions}
\crefname{corollary}{Corollary}{Corollaries}
\crefname{table}{Table}{Tables}
\crefname{figure}{Figure}{Figures}
\crefname{remark}{Remark}{Remarks}
\crefname{section}{Section}{Sections}
\crefname{subsection}{Section}{Sections}

\ifpdf
\DeclareGraphicsExtensions{.eps,.pdf,.png,.jpg}
\else
\DeclareGraphicsExtensions{.eps}
\fi

\newsiamremark{remark}{Remark}
\newsiamremark{hypothesis}{Hypothesis}
\crefname{hypothesis}{Hypothesis}{Hypotheses}
\newsiamthm{claim}{Claim}

\headers{Asynchronous data flow}{R.C. Barnard, C.D. Hauck, K. Huang}

\title{
	A mathematical model of asynchronous data flow in parallel computers \thanks{. The United States Government retains and the publisher, by accepting the article for publication, acknowledges that the United States Government retains a non-exclusive, paid-up, irrevocable, world-wide license to publish or reproduce the published form of this manuscript, or allow others to do so, for the United States Government purposes. The Department of Energy will provide public access to these results of federally sponsored research in accordance with the DOE Public Access Plan (\texttt{http://energy.gov/downloads/doe-public-access-plan}).
		\funding{This research is sponsored by the Office of Advanced Scientific Computing Research; U.S. Department of Energy. The work was performed at the Oak Ridge National Laboratory, which is managed by UT-Battelle, LLC under Contract No. De-AC05-00OR22725.}}}

\author{Richard C. Barnard \thanks{Western Washington University, Bellingham, WA
		(\email{rick.barnard@wwu.edu}, \url{https://cse.wwu.edu/mathematics/barnarr3}).}
	\and Kai Huang \thanks{Michigan State University, East Lansing, MI
		(\email{huangk18@msu.edu}).}
	\and Cory Hauck\thanks{Oak Ridege National Laboratory, Oak Ridge, TN
		(\email{hauckc@ornl.gov}, \url{https://www.csm.ornl.gov/\string~hfd/}).}
}

\usepackage{amsopn}

\begin{document} 

\maketitle

% REQUIRED
\begin{abstract}
We present a simplified model of data flow on processors in a high performance computing framework involving computations necessitating inter-processor communications.  From this ordinary differential model, we take its asymptotic limit, resulting in a model which treats the computer as a continuum of processors and data flow as an Eulerian fluid governed by a conservation law.
We derive a Hamilton-Jacobi equation associated with this conservation law for which the existence and uniqueness of solutions can be proven.  
We then present the results of numerical experiments for both discrete and continuum models; these show a qualitative agreement between the two and the effect of variations in the computing environment's processing capabilities on the progress of the modeled computation.
\end{abstract}

% REQUIRED

% REQUIRED

\section{Introduction}
\label{sec:intro}
It has been well-established that current and future generations of extreme scale computers have achieved and, for the foreseeable future, are expected to achieve increases in performance via greater levels of parallelism at multiple levels --- e.g., within the processors as well as increasing the number of processors and nodes ---as opposed to increases in clock speeds, which are expected to remain relatively flat. Additionally, extremely concurrent codes, involving dynamic parallelism and greater degrees of asynchronous parallel executions, are increasingly needed to leverage this large scale parallellism \cite{SteWhi09,DonHitBel14}. 

As machine improvements depend on increasingly complex architectures and as additional constraints on system development and planning (such as power consumption \cite{DonHitBel14}) arise, a need for predictive, quantitative models of computational performance will grow greater.   
Previously developed modeling tools such as LogP \cite{CulKarPat93,CulKarPat96} result in easily evaluated models which can prove difficult to extend and modify.  
Alternatively, PRAM models have been used as abstractions of codes; these however have scalability issues due to the complexity of simulating them \cite{PiePuc97}.  
Other modern tools \cite{NunFerFil12,Kun13} are similarly still limited to fine-grain simulations of at most a few dozen nodes, again due to their computational complexity during simulations.

Core counts are now in the hundreds of thousands and millions on machines in the TOP500 list of supercomputers; node counts consistently are in the thousands \cite{TOP500}.  
Such numbers mean that fine-grained simulation tools (such as those listed above) are incapable of describing large-scale phenomena; essentially the simulation tools begin to require computational resources beyond those of the systems they are simulating.
Alternative approaches have been proposed to address these issues: miniapp codes can mimic key features of the performance of exascale codes with a much smaller codebase \cite{DosBarDoe14}.
Aspen, a framework for performance modeling \cite{SpaVet12,VetMer15}, uses a domain specific language which encodes both abstracted features of machines hardware and specific software applications to provide coarse-grained simulations.
However, these suffer from the need to develop specialized simulation codes which can be problem dependent, resulting in possibly labor-intensive tools. 
A workflow modeling apporach, Pegasus, has been developed to model workflows using a graph-theoretic perspective to detect and manage anamolies in the computing environment \cite{DeeVahJuv15}.

We propose developing a macroscopic model of extreme scale computers which views such computing environments in a continuum framework.
Such a model has several potential benefits: in addition to being computationally tractable, it will open up the possibility of using the theoretical tools of partial differential equations to understand and control the performance of high-performance computing systems.
Specifically, our goal is to derive a fluid-limit model of data flow --- which can be described by a partial differential equation --- from a simplified deterministic model of data processing and flow in an extreme scale computer with interprocessor communications and asynchronous executions.   
Fluid models, beyond their obvious utility in physical systems, have been used to model flows in networks, such as vehicular traffic flows \cite{AwRas00}, supply chains \cite{ArmMarRin03}, and gas networks \cite{BanHerKla06,BroGasHer11}.
In particular, as discussed in \cite{KlaWeg99a} and \cite{ArmMarRin03}, such fluid models lie at the end of a hierarchy of models which begin with microscopic or discrete models. 
That is, similar to the derivation of physical fluid laws from many body physics, one may derive continuum-level flow equations from discrete-level models of the dynamics of agent interactions.  
With such a model, standard numerical simulation tools and analytical methods may be brought to bear for studying large-scale phenomena in extreme-scale computing.

We begin in \cref{sec:discrete} with a microscopic model of a network of processors performing a multi-stage computational task which necessitates inter-processor communications.  
In \cref{sec:continuum}, we derive a formal asymptotic limit of this agent-based model as the scale of the system increases, resulting in an Eulerian fluid flow model.  
Along with the resulting nonlinear conservation law, we present a related Hamilton-Jacobi equation and establish the existence of solutions in \cref{sec:continuum}.  
In \cref{sec:numerics}, we present the results of numerical experiments to show agreement between the microscopic and fluid models and then illustrate the behavior of solutions under heterogeneous computing layouts.

\section{The discrete model}
\label{sec:discrete}

In this section, we introduce the microscopic model, which is based on a highly simplified, deterministic, semi-discrete ordinary differential equation (ODE).   We imagine the computer as a network of processors $\{\cP_i\}_{i=1}^{\imax}$ that are arranged in a one-dimensional, periodic lattice. 
The computer is assigned a computational job involving a sequence of $\kmax$ tasks which are identical in the sense that each one takes the same effort to complete.
This computational job is divided by distributing data amongst processors.
%The computational job is divided amongst processors and an initial allocation of the data to be processed is given to each processor; the completion of task $k$ by $\cP_i$ is modeled as moving data within that processor from stage $k$ to stage $k+1$.  However, this task completion is not completely parallel across processors. In order to complete task $k$, a processor must receive some information from its neighbors; the relevant information is associated with those neighbors' data at stage $k$.  The communication and flow of data between processors is not directly modeled here.  Rather, this is modeled by either slowing or stopping the processing of data from stage $k$ to stage $k+1$ on $\cP_i$ if $\cP_{i-1}$ and $\cP_{i+1}$ have not processed data to stage $k$. A schematic of the setup is given in Figure \ref{fig:network_fig}.
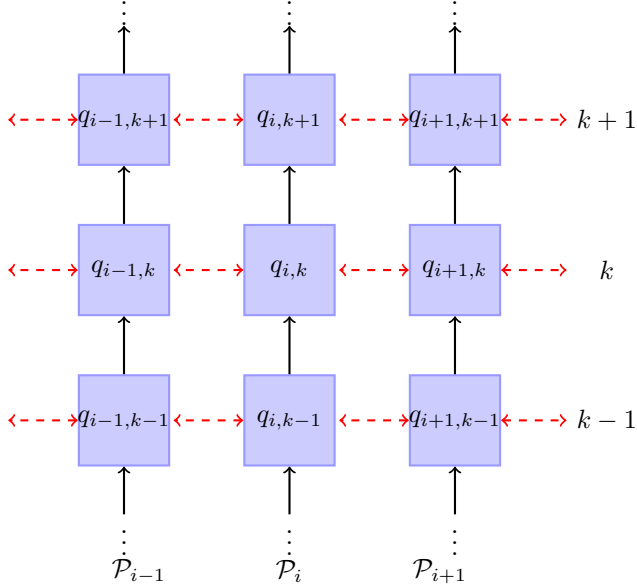
\begin{figure}
	\centering
	\begin{tikzpicture}
	[inner sep =6mm,core/.style ={draw,draw=blue!40, 	fill=blue!20,thick}]
	\foreach \x in {-2.2,0,2.2} {
		\node  (nextstep) at (\x,2)   [core] {};
		\node  (thisstep) at (\x,0)   [core]  {};
		\node  (oldstep)  at  (\x,-2) [core] {};
		\node  (before) at (\x,-3.52) {$\vdots$};
		\node  (after) at (\x,3.52) {$\vdots$};
		\draw [->,thick] (\x,-3.25) -- (oldstep.south);
		\draw [->,thick] (oldstep.north) -- (thisstep.south);
		\draw [->,thick] (thisstep.north) -- (nextstep.south)	;
		\draw [->,thick] (nextstep.north) -- (\x,3.25);
		\draw [<->, thick,red,dashed]  (oldstep.west) -- ++(-.925,0);
		\draw [<->, thick,red,dashed]  (thisstep.west) -- ++(-.925,0);
		\draw [<->, thick,red,dashed]  (nextstep.west) -- ++(-.925,0);			
	}
	\node at (-2.2,-2) {$q_{i-1,k-1}$};
	\node at (-2.2,0) {$q_{i-1,k}$};
	\node at (-2.2,2) {$q_{i-1,k+1}$};
	\node at (0,-2) {$q_{i,k-1}$};
	\node at (0,0) {$q_{i,k}$};
	\node at (0,2) {$q_{i,k+1}$};
	\node at (2.2,-2) {$q_{i+1,k-1}$};
	\node at (2.2,0)  {$q_{i+1,k}$};
	\node at (2.2,2) {$q_{i+1,k+1}$};
	\draw [<->,thick,red,dashed] (oldstep.east) -- ++(.87,0);
	\draw [<->,thick,red,dashed] (thisstep.east) -- ++(.87,0);		
	\draw [<->,thick,red,dashed] (nextstep.east) -- ++(.87,0);
	\node at (4.22,-2)		{$k-1$};
	\node at (4.22,0)		{$k$};
	\node at (4.22,2)		{$k+1$};
	\node at (-2,-4)		{$\cP_{i-1}$};
	\node at (0,-4)			{$\cP_{i}$};
	\node at (2,-4)			{$\cP_{i+1}$};	
	\end{tikzpicture}
	\caption{Schematic of network of processors.  Dashed lines denote inter-processor communications}
	\label{fig:network_fig}
\end{figure}
We denote by $q_{i,k}(t)$ the amount of data in $\cP_i$ that sits in stage $k$ at time $t$.  

\subsection{Conservation law}
The dynamics of $q_{i,k}$ are given by a conservation law of the form
\begin{equation}
\label{eq:ODE}
\dot q_{i,k}(t) = F_{i,k-1}(t) - F_{i,k}(t),
\quad k=1,\dots,\kmax, \quad i = 1,\dots \imax,
\end{equation}
where $F_{i,k}$ ($i = 1,\dots ,\imax$, $k=1,\dots,\kmax-1$) is the rate of data moving in processor $i$ from stage $k$ to $k+1$, referred to as the {\em throughput}. At the first stage $k=1,$ $F_{i,0}$ ($i = 1,\dots ,\imax$) is the rate of data being loaded into processor $i$ to be processed, referred to as the {\em inflow}, and at the final stage, $F_{i,\kmax}$ ($i = 1,\dots, \imax$) is the rate of data completing the final stage of the job, referred to as the {\em outflow}. 
Equation \eqref{eq:ODE} implies that the data in each processor is neither created or destroyed, only moved in and out of the processor or in between stages; that is,
\begin{equation}
\label{eq:ODE_cons}
\frac{d}{d t} \left( \sum_{k=1}^{\kmax} {q}_{i,k} \right) = F_{i,0} - F_{i,\kmax}.
\end{equation} 
A key aspect of the model is that it does not separately track data that moves between processors; instead the effects of communication delays will be incorporated directly into the definition of the throughputs.  

A fundamental quantity of interest in the discrete model is $Q_{i,k}(t)$, which is defined as the amount of data at time $t$ that has gone through the first $k-1$ stages of $\cP_i$.  For each $t \geq0$,
\begin{equation} \label{eq:Q}
Q_{i,k}(t) = \Big(\sum_{j=k}^{\kmax}q_{i,j}(t)\Big)+\int_0^tF_{i,\kmax}(s)ds.
\end{equation}

\subsection{Processor throttling}
We now turn to specifying the form of $F_{i,k}$. 
In the absence of throttling, each processor $\cP_i$ moves data between stages at a rate $a_i \geq 0$, which we refer to as the {\em maximum throughput}. For the purposes of the current paper, we assume that $a_i$ is given.  In practice, it must be determined from experiments, fine-scale models, or a combination of both.  It may also depend on $k$, although for simplicity, we assume here that it does not.
Throttling is said to occur whenever $F_{i,k}(t) < a_i$; this happens for one of two reasons. 
\begin{enumerate}
	
	\item \textbf{Self-throttling:} Given an amount of data $q_{i,k}$ to be processed at stage $k$ in processor $\cP_{i}$, we define the self-throttling function 
	\begin{equation} \label{eq:v1}
	v_1(q_{i,k};q_*) = \max \left \{0, \min \left \{1,\frac{q_{i,k}}{q_*} \right\} \right\}.
	\end{equation}
	%	\begin{equation} \label{eq:v1}
	%		v_1(q_{i,k};q_*) = 
	%		\begin{cases}
	%		1, & q_{i,k} > q_*, \\
	%		\frac{q}{q_*}, & 0\leq q_{i,k} \leq q_*,\\
	%		0, & q_{i,k}<0
	%		\end{cases}
	%	\end{equation}
	Clearly if no data is available to be processed, then $F_{i,k}=0.$ Furthermore, we assume that if the amount of data to be processed drops below a certain threshold $q_*>0$, then $\cP_i$ cannot maintain the throughput $a_i$ and the throughput is reduced.

	\item \textbf{Neighbor throttling:} As the computational task is not entirely parallel across processors,  $\cP_i$ requires sufficient information from its neighbors to perform task $k$ at full throughput. The neighbor throttling function $v_2$ models this dependence. It gives the amount of available data on $\cP_i$ at stage $k$
	\begin{equation} \label{eq:v2}
	v_2(q_{i,k},\Delta_{i+1,k},\Delta_{i-1,k};\beta) = \min\left\{q_{i,k},\frac{1}{\beta}\max\{\Delta_{i+1,k},0\},\frac{1}{\beta}\max\{\Delta_{i-1,k},0\}\right\}.
	\end{equation}
	Here $\Delta_{i\pm1,k}$ denotes the data on the right/left neighbor which is available to be used by $\cP_i$ to process $q_{i,k}.$ 
	The parameter $\beta\in(0,1]$ allows for the possibility that computations do not rely in a one-to-one fashion upon the availability of data from neighbors. 
	If $\Delta_{i\pm 1,k}=0$ the processing of data stops due to the absence of a necessary component of the computational task and so $F_{i,k}=0.$
	Alternatively, if both $\Delta_{i+1,k}$ and $\Delta_{i-1,k}$ exceed $\beta q_{i,k},$ then $\cP_{i}$ has sufficient data from its neighbors to process $q_{i,k}$ and no throttling occurs.   
	
	The data from the left/right neighbor which is available for processing at stage $k$ is given by
	\begin{equation}
	\Delta_{i\pm1,k} = Q_{i\pm 1,k}-Q_{i,k+1}=Q_{i\pm 1,k}-(Q_{i,k}-q_{i,k}).
	\end{equation}
	The data on each neighbor must have completed the same stage for it to be available; additionally, this data is not reused on $\cP_i$ for the same stage.
	This means that the data available to be used from the neighbors can be written as above and so the amount of data available to be processed on $\cP_{i}$ at stage $k$ is given by
	\begin{equation}
	v_2\big(q_{i,k},Q_{i+1,k}-Q_{i,k}+q_{i,k},Q_{i-1,k}-Q_{i,k}+q_{i,k};\beta\big).
	\end{equation}
	
\end{enumerate}

\begin{figure}
	\centering
	\begin{subfigure}[b]{.4\linewidth}
	\includegraphics[width=\textwidth]{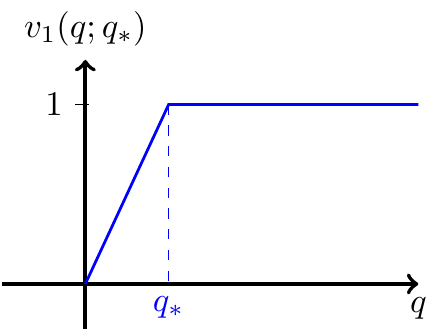}
		\vskip8pt
		\caption{$v_1(q;q_*)$ for a given $q_*$}
	\end{subfigure}
	%	\hfill
	\begin{subfigure}[b]{.50\linewidth}
		\vskip10pt
			\includegraphics[width=\textwidth]{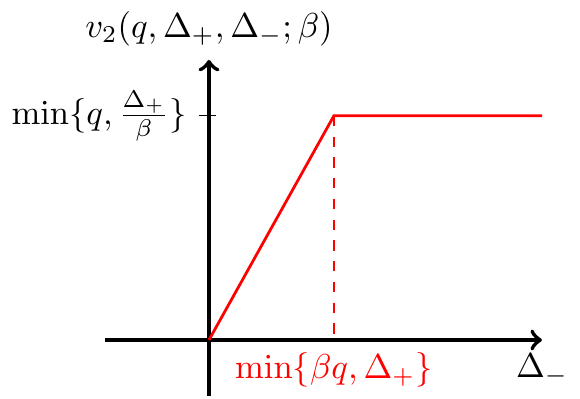}
		\caption{$v_2(q,\Delta_+,\Delta_-;\beta)$ for a given $q$ and $\Delta_+\geq0$.}
	\end{subfigure}
	\caption{Throttling functions $v_1$ and $v_2$, defined in \eqref{eq:v1} and \eqref{eq:v2}, respectively} 
	\label{fig:v_fcn_scheme}
\end{figure}

The throughput $F_{i,k}$ is a composition of the throttling functions $v_1$ and $v_2$:
\begin{equation}
\label{eq:flux}
F_{i,k} 
= a_i v_1\Big(v_2\big(q_{i,k},Q_{i+1,k}-Q_{i,k}+q_{i,k},Q_{i-1,k}-Q_{i,k}+q_{i,k}; \beta\big);q_*\Big).
\end{equation} 
At first glance, this definition of $F_{i,k}$ appears circular since it depends on $Q_{i,k}$, which in turn depends on $F_{i,\kmax}$.  However, as a consequence of the conservation law \eqref{eq:ODE_cons}, 
\begin{equation}
\label{eq:unwrap_discrete}
\int_0^tF_{i,\kmax}(s)ds = \int_0^t F_{i,0}(s)ds + \sum_{j=1}^{\kmax}q_{i,j}(0) -  \sum_{j=1}^{\kmax}q_{i,j}(t)
\end{equation} 
Thus to complete the model, we need only prescribe initial data $q_{i,k}(0)$ and the inflow $F_{i,0}$. 
To prescribe the inflow, we specify $q_{i,0}$ and then let $F_{i,0}$ be evaluated according to \eqref{eq:flux}.  

\begin{proposition}
	The system \eqref{eq:ODE} with (i) throughput $F_{i,k}$ defined in \eqref{eq:flux} for $i=1,\dots,\imax$ and $k=1,\dots,\kmax$; (ii) prescribed initial data $q_{i,k}(0)$ for $i=1,\dots,\imax$ and $k = 1,\dots,\kmax$; and (iii) prescribed inflow data $F_{i,0}$ for $i=1,\dots,\imax$ and $ t \geq 0$ has a unique solution for all $t \geq 0$.  Moreover, if $q_{i,k}(0)\geq 0$ for all $i=1,\dots,\imax$ and $k = 1,\dots,\kmax,$ then $q_{i,k}(t)\geq 0$ for all $t\geq 0$ and $i=1,\dots,\imax$ and $k=1,\dots,\kmax$.  
\end{proposition}
\begin{proof}
	Since $F_{i,k}$ is globally Lipschitz in its arguments for every $i,k$, standard ODE theory (see, for example Theorem III.VI of \cite{Wal98})  implies the existence of a unique solution. Moreover, it is clear from \eqref{eq:flux} that $0 \leq F_{i,k} \leq a_i v_1(q_{i,k};q_*)$.  Hence according to \eqref{eq:ODE},
	\begin{equation}
	\label{eq:diff_comparison}
	\dot{q}_{i,k}(t) \geq -a_i v_1(q_{i,k};q_*).
	\end{equation}
	%	In particular, if $q_{i,k}(t_0)<q_*$ for some time $t_0$, then from the definition fo $v_1$ in 
	%	\eqref{eq:v1},
	%	\begin{equation}
	%	\dot{q}_{i,k}(t_0)\geq - \frac{a_i}{q_*} q_{i,k}
	%	\end{equation}
	Standard comparison results for ordinary differential equations imply then that $q(t) \geq 0$. (See, for example, Lemma 1.2 of \cite{Tes12} and for comparison, use the zero function, which satisfies \eqref{eq:diff_comparison} as an equality.)  
\end{proof}
\section{The continuum model}
\label{sec:continuum}
In this section, we derive a continuous model that is formally accurate in the limit as $\imax$ and $\kmax$ tend to infinity.
We assume, in taking this limit, that the job performed by the computer is fixed -- that is, the total amount of work does not change. 
For given $\imax,\kmax$, we define the quantities:
\begin{equation} \label{eq:scale_params}
\delta := (\kmax)^{-1},\quad \veps := (\imax)^{-1}, \quad \eta:=\frac{\veps}{\delta}.
\end{equation}
Here, $\delta$ is the fraction of the work done in each stage and $\veps$ is the average amount of data in a processor.  Finally, $\eta$ is simply the ratio of $\kmax$ and $\imax$ which will be of use in the following analysis.

\subsection{Formal Derivation}\label{sec:derive}

To derive a continuum model, we first express the ODE \eqref{eq:ODE} in terms of the following $O(1)$ quantities: 
\begin{align}\label{eq:scaling}
r_* := \frac{q_*} {\veps \delta} \,,
\quad 
r_{i,k} := \frac{q_{i,k}}{\veps \delta}  \,,
\quad
R_{i,k} 
%:= \frac{1}{\veps} \int_{x_{i-1/2}}^{x_{i+1/2}} P(x,z_{k}) dx 
:= \frac{1}{\veps} Q_{i,k} \,,
\quad 
D^\pm_{i,k} := \pm \frac{R_{i \pm 1,k} - R_{i,k}}{\veps}  \,,
\quad
\alpha_i := \frac{a_i}{\veps} \,.
\end{align} 
In terms of these rescaled quantities, the neighbor throttling function as can be written as
\begin{align}
v_2\big(q_{i,k},&Q_{i+1,k}-Q_{i,k}+q_{i,k},Q_{i-1,k}-Q_{i,k}+q_{i,k};\beta\big)  \\
& = \min\Big\{\veps\delta r_{i,k},\frac{1}{\beta}\max\{ \veps R_{i+1,k}-\veps R_{i,k}+\veps\delta r_{i,k},0\},\frac{1}{\beta}\max\{ \veps R_{i-1,k}-\veps R_{i,k}+\veps\delta r_{i,k},0\}   \Big\}
\nonumber \\
& = \veps\min\Big\{\delta r_{i,k},\frac{1}{\beta}\max\{  R_{i+1,k}- R_{i,k}+\delta r_{i,k},0\},\frac{1}{\beta}\max\{  R_{i-1,k}- R_{i,k}+\delta r_{i,k},0\}   \Big\}
\nonumber \\
& = \veps\delta\min\Big\{ r_{i,k},\frac{1}{\beta}\max\{  \eta D^+_{i,k}+ r_{i,k},0\},\frac{1}{\beta}\max\{  -\eta D^-_{i,k}+ r_{i,k},0\}  \Big\}. 
\nonumber 
\end{align}
Therefore,
\begin{multline}
\label{eq:v1v2_rescale}
v_1\bigg({v}_2(q_{i,k},{v}_2(q,Q_{i+1,k}-Q_{i,k}+q_{i,k},Q_{i-1,k}-Q_{i,k}+q_{i,k}));q_*\bigg)\\
= \min\left\{1,\frac{ \min\Big\{ r_{i,k},\frac{1}{\beta}\max\{  \eta D^+_{i,k}+ r_{i,k},0\},\frac{1}{\beta}\max\{  -\eta D^-_{i,k}+ r_{i,k},0\}  \Big\} }{r_*}  \right\}.
\end{multline}
With \eqref{eq:v1v2_rescale} in mind, we define the rescaled throttling functions
\begin{subequations}
	\label{eq:w1andw2}
	\begin{align}
	w_1(r,r_*)&=\max \left \{0, \min \left \{1,\frac{r}{r_*} \right\} \right\}. \label{eq:w1}
	\\
	w_2\left(r,D^-,D^+;\eta,\beta \right)&=\min \left\{r,\frac{1}{\beta}\max\{\eta D^+ +r,0\}, \frac{1}{\beta}\max\{\eta D^- +r,0\} \right\} \label{eq:w2}
	\end{align}
\end{subequations}
and the composite function
\begin{equation}
\label{eq:w}
w\left(r,D^-,D^+;r_*,\alpha,\eta,\beta \right):=\alpha w_1\left(w_2(r,D^-,D^+;\eta,\beta);r_*\right).
\end{equation}
The dynamics in \eqref{eq:flux} can now be re-expressed in terms of the $O(1)$ quantities in \eqref{eq:scaling}, thereby obtaining a evolution formula for $r_{i,k}$:
\begin{equation}
\dot{r}_{i,k}(t) = \frac{f_{i,k-1}(t) - f_{i,k}(t)}{\delta}\label{eq:ODE_scaled}
\end{equation}
for $i=1,\dots,\imax$ and $k=1,\dots,\kmax$, where
\begin{equation}\label{eq:flux_scaled}
f_{i,k}(t) = 
w\left( r_{i,k}(t),-D_{i,k}^-(t),D_{i,k}^+(t);r_*,\alpha_i,\eta,\beta\right) 
\end{equation}
for $i=1,\dots,\imax, \, k=0,\dots,\kmax$, and $r_{i,0}$ is prescribed for $i=1,\dots,\imax$

The next step is to interpret \eqref{eq:ODE_scaled} as a conservative finite-difference formula for a sufficiently smooth function $\rho = \rho(x,y,t)$, defined on $[0,1) \times [0,1] \times [0,\infty)$, such that
\begin{equation}\label{eq:ODEtogrid}
\rho(x_i,z_k,t)  = r_{i,k}(t),
\end{equation}
on grid points
\begin{equation}\label{eq:meshpts}
x_i = \left(i-0.5\right)\veps \quad \text{and} \quad z_k = k \delta, 
\end{equation}
for $i=1,\dots,\imax$ and $k=0,\dots,\kmax$.    We also let $\alpha = \alpha(x)$ be a continuous function such that $\alpha(x_i)=\alpha_i$.  

Let $\psi = \psi(x,z,t)$ be a smooth test function with compact support on $[0,1) \times[0,1]\times [0,\infty)$ and set $\psi_{i,k}(t) = \psi(x_i,z_k,t)$. From \eqref{eq:ODE_scaled}, 
\begin{align} 
\label{eq:discrete_cons}
\delta\sum_{k=1}^{\kmax}\psi_{i,k}(t)\dot{r}_{i,k}(t)
& = \sum_{k=0}^{\kmax-1}[\psi_{i,k+1}(t) - \psi_{i,k}(t)]f_{i,k}(t) \\
& \qquad + \psi_{i,0}(t)f_{i,0}(t) - \psi_{i,\kmax}(t)f_{i,\kmax}(t).
\end{align}
Let the function $\phi = \phi(x,z,t)$ interpolate the fluxes on the grid:
\begin{equation}
\phi (x_i,z_k,t) = f_{i,k}(t),
\end{equation}
for $i=1,\dots,\imax, k=1,\dots,\kmax,t\geq 0$. Then \eqref{eq:discrete_cons} can be interpreted formally as the weak formulation (with respect to $z$) of a conservation law for $\rho$ with flux $\phi$:
\begin{align} 
\int_0^1 \psi(x,\xi,t)\p_t\rho(x,\xi,t) d\xi 
& = \int_0^1\p_z\psi(x,\xi,t)\phi (x,\xi,t)d\xi \nonumber  \\
& \quad +\psi(x,0,t)\phi (x,0,t)- \psi(x,1,t)\phi (x,1,t)+O(\delta). \label{eq:wk_form}
\end{align}

To derive a closed model from \eqref{eq:wk_form}, we approximate $\phi $ in terms of $\rho$.   Such an approximation depends on $D_{i,k}^\pm$ via the formula for $f_{i,k}$ in \eqref{eq:flux_scaled}.  From the definition of $Q_{i,k}$  in \eqref{eq:Q} and the scalings in \eqref{eq:scaling}, it follows that
\begin{equation}
R_{i,k} (t)
= \delta \sum_{j=k}^{\kmax} r_{i,j}(t) + \int_0^t f_{i,\kmax}(s)ds 
\end{equation}
and, moreover, that for any finite $t>0$,
\begin{align}
\label{eq:D_approx}
\pm D_{i,k}^\pm(t) 
&=   \frac{\delta}{\veps} \sum_{j=k}^{\kmax} [r_{i \pm 1,j} - r_{i,j}(t)] 
+ \frac{1}{\veps} \int_0^t \left[ f_{i\pm 1,\kmax}(s) - f_{i,\kmax}(s)\right] ds \\
&= \delta \sum_{j=k}^{\kmax} \left[ \pm \p_x \rho(x_i,z_j,t) 
+ \frac{\veps}{2} \p_x^2 \rho(x_i,z_j,t) + O(\veps^2)\right] \\
& \qquad \qquad \qquad + \int_0^t \left[\pm \p_x \phi(x_i,1,s)  
+ \frac{\veps}{2} \p_x^2 \phi(x_i,1,t) + O(\veps^2)\right]  ds\nonumber \\
&= \pm\p_x P(x_i,z_k,t)+\frac{\veps}{2}\p_x^2 P(x_i,z_k,t)+O(\veps^2)+O(\delta)
\label{eq:D_to_P}	
\end{align}
where $P$ is given by 
\begin{align}\label{eq:P_def_prelim}
P(x,z,t) = \int_z^1\rho(x,\xi,t)d\xi&+\int_0^t\phi(x,1,s)ds
\end{align}
Motivated by the above calculation, we approximate $\phi$ by one of two flux functions:

\begin{subequations}
	\label{eq:Phi_def}
	\begin{align}\label{eq:Phi0}
	\Phi^{(0)} \left( \rho,\p_x  P;r_*,\alpha,\eta,\beta \right)
	&= w \left( \rho , - \p_x  P , \p_x  P ;r_*,\alpha,\eta,\beta\right)
	\\
	\label{eq:Phi1} 
	\Phi^{(1)} \left(\rho,\p_x  P, \p^2_x P ;r_*,\alpha,\eta,\beta \right)
	&= w \left(\rho , - \p_x  P + \frac{\veps}{2} \p_x^2  P, \p_x  P + \frac{\veps}{2} \p_x^2  P;r_*,\alpha,\eta,\beta\right).
	\end{align}
\end{subequations}
Using \eqref{eq:D_to_P} and the Lipschitz continuity of $w$ with respect to $D^{\pm}$, we conclude that  
\begin{align}
w&(r_{i,k}(t),D_{i,k}^-(t),D_{i,k}^+(t);r_*,\alpha_i,\eta,\beta)  
\\ 
&= \Phi^{(1)}(\rho(x_i,z_k,t),\p_x {P}(x_i,z_k,t), \p^2_x {P}(x_i,z_k,t) ;r_*,a,\eta,\beta) + O(\veps^2) +  O(\delta). \nonumber \\
&= \Phi^{(0)}(\rho(x_i,z_k,t),\p_x {P}(x_i,z_k,t);r_*,a,\eta,\beta) + O(\veps) +  O(\delta).
\nonumber
\end{align}
Thus for  $0 \ll \veps,\delta \ll 1$, with $\eta \in (0,\infty)$ fixed, \eqref{eq:wk_form} is formally consistent with the  continuum model

\begin{subequations}
	\label{eq:PDEmodel}
	\begin{align}
	&\p_t \rho + \p_z \Phi^{(\ell)}(\rho,\p_xP,\p_x^2P;r_*,a,\eta,\beta) = 0,
	&& \,(x,z,t)\in \bbT^1\times(0,1)\times(0,\infty), \label{eq:cons_law_cont}\\
	&\rho(x,0,t) = \rho_{\rm{bc}}(x,t),
	&& \, (x,t) \in \bbT^1\times(0,\infty)\label{eq:bndry},\\
	&\rho(x,z,0) = \rho_0(x,z), 
	&& \, (x,z) \in \bbT^1\times(0,1)  \label{eq:initial} 
	\end{align}
\end{subequations}
where
\begin{subequations}
	\begin{align}
	P(x,z,t) &= \int_z^1\rho(x,\xi,t)d\xi 
	+ \int_0^t \phi^{(\ell)}(x,1,s) ds, \label{eq:P_def} \\
	\phi^{(\ell)}(x,z,t) &= \Phi^{(\ell)} \left(\rho(x,z,t),\p_x  P(x,z,t), \p^2_x P(x,z,t) ;r_*,\alpha,\eta,\beta \right), \label{eq:phi_ell_def}
	\end{align}
\end{subequations}
and
$\Phi^{(\ell)}$, $\ell \in \{0,1\}$, is given in \eqref{eq:Phi_def}.  For the sake of compactness, we have slightly abused notation in \eqref{eq:cons_law_cont}, as the definition of $\Phi^{(0)}$ is independent of $\p_x^2P$.  Additionally, we have identified $[0,1)$ with the one-dimensional torus $\bbT^1$ in order to reflect the periodic layout of the processors. 

As in the discrete case, it may appear that the model in \eqref{eq:PDEmodel} is circular due to  the definition of $P$ in \eqref{eq:P_def}.  However, as with $F$ in \eqref{eq:unwrap_discrete}, $\Phi^{(\ell)}$ can be unwrapped, this time using the conservation law \eqref{eq:cons_law_cont}; that is
\begin{equation}
\int_0^t\phi^{(\ell)}(x,1,s)ds 
\label{eq:cont_unwrap}
= \int_0^t \phi(x,0,s)ds + \int_0^1 \rho_0(x,\xi)d\xi  -\int_0^1\rho(x,\xi,t) d\xi 
\end{equation}
Thus the continuum model is complete once initial condition $\rho_0$ and inflow condition $\phi_{\rm{bc}} := \phi(\cdot,0,\cdot)$ are specified.  In practice, $\rho_{\rm{bc}}$ is prescribed and then $\phi_{\rm{bc}}$ is evaluated using \eqref{eq:phi_ell_def} and \eqref{eq:Phi_def}.

We use the flux function $\Phi^{(0)}$ for all of the numerical simulations in Section \ref{sec:numerics}.  This function is a piecewise constant that can be expressed in the following form:
\begin{equation}\label{eq:w_vals}
\Phi^{(0)}(r,D;r_*,\alpha,\beta)=\begin{dcases}
\alpha&(r,D) \in \Omega_1,\\
\frac{\alpha r}{r_*}& (r,D) \in  \Omega_2,\\
\frac{\alpha (r+\eta D)}{\beta r_*}&(r,D) \in \Omega_3,\\
\frac{\alpha (r-\eta D)}{\beta r_*}&(r,D) \in \Omega_4,\\
\end{dcases}
\end{equation}
where the subdomains $\Omega_i$ are depicted in Figure \ref{fig:wD}.

\begin{figure}
	\centering
	\begin{tikzpicture}[scale = 2]
	\draw [very thick,->, opacity = .8]  (0.0,0.0) -- node [at end, right] {$r$} (0.0,2.25);
	\draw [very thick, ->, opacity = .8]	(-2.25,0.0) -- node [very near end, right, below] {$D$} (2.25,0.0);
	\draw [dashed, red]  (0.0,0.0) -- node [very near end, right=3pt] {$\eta D = r$} (2.,2.);
	\draw [dashed,red]  (0.0,0.0) -- node [very near end,left = 3pt] {$-\eta D = r$} (-2.,2.);
	\draw [dashed,red]  (-2.25,1) -- node [very near end ,below=6pt,right] {$r=r_*$} (2.25,1);
	\draw [dashed,red] (0.5,1) -- node [very near end,above  =12pt] {$\eta D=r-\beta r^*$} (1.5,2);
	\draw [dashed,red] (-0.5,1) -- node [very near end,above  =12pt] {$-\eta D = r-\beta r^*$} (-1.5,2);
	\draw [dashed,red] (0.5,1) -- (0,0);
	\draw [dashed,red] (-0.5,1) -- (0,0);
	\fill [blue,opacity = .2] (-0.5,1)--(0.5,1) -- (1.5,2.) -- (-1.5,2.)--cycle;
	\fill [purple!40,opacity = .2] (0,0) -- (-2,2) -- (-1.5,2) -- (-0.5,1) -- cycle;
	\fill [cyan!60,opacity = .2] (0,0) -- (2,2) -- (1.5,2) -- (0.5,1) -- cycle;;		
	\fill [green,opacity = .2] (0.0,0.0) -- (-.5,1) -- (.5,1) -- cycle;	
	\node at (0.3,1.5) {$\Omega_1$};
	\node at (-0.3,1.5) {$\Omega_1$};
	\node at (-0.15,0.7) {$\Omega_2$};
	\node at (0.15,0.7) {$\Omega_2$};
	\node at (-0.5,.7) {$\Omega_3$};
	\node at (-1.1,1.3) {$\Omega_3$};
	\node at (1.1,1.3) {$\Omega_4$};
	\node at (.5,0.7) {$\Omega_4$};
	\draw[red] (.3,.6) -- node [pos =1,right,red] {$\frac{\eta D}{1-\beta} = r$} (1.5,.15);
	\draw[red] (-.3,.6) -- node [pos =1,left,red] {$\frac{-\eta D}{1-\beta} = r$} (-1.5,.15);
	\end{tikzpicture}
	\caption{Flux $\Phi^{(0)}$ defined in \eqref{eq:w_vals}}
	\label{fig:wD}
\end{figure}
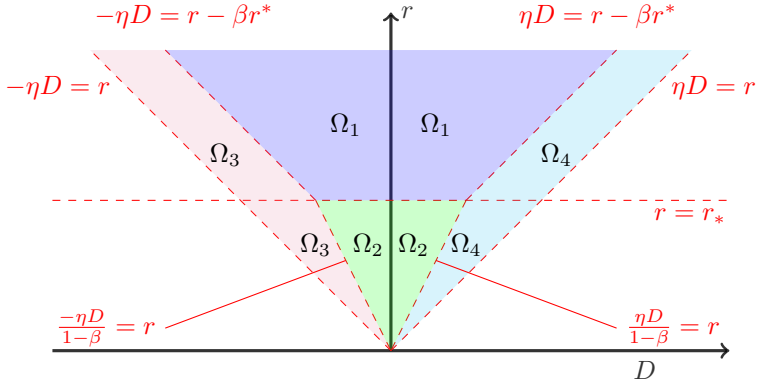

\subsection{A Hamilton-Jacobi formulation}\label{sec:HJB}
To our knowledge, there are no immediate conclusions available from the literature regarding the existence or uniqueness of solutions to \eqref{eq:PDEmodel}.  
However, we may consider instead a related Hamilton-Jacobi equation which in turn opens up the extensive theory of viscosity solutions. 
For background, we refer to \cite{CraIshLio92, ImbSil13, Bar13}.  
We are interested primarily in obtaining, for regular inputs $\alpha$, $\phi_{\rm{bc}}$, and  $\rho_0$, the existence and uniqueness of $P$.  
It is possible that for more general inputs such results are available in the extensive viscosity solution literature (e.g. \cite{Bar93,Coc03,JakKar02}).

Integrating \eqref{eq:cons_law_cont} with respect to $z$ gives 
\begin{equation}
\label{eq:P1}
\p_t\int_{z}^1\rho(x,\xi,t)d\xi + \Phi^{(\ell)}(x,1,t) - \Phi^{(\ell)}(x,z,t) =0
\end{equation}
Meanwhile, differentiating \eqref{eq:P_def} gives
\begin{equation}
\label{eq:P2}
\p_t P(x,z,t) = \p_t\int_{z}^1\rho(x,\xi,t)d\xi + \Phi^{(\ell)}(x,1,t)
\end{equation}
Combining \eqref{eq:P1} and \eqref{eq:P2} and using the fact that $\rho=-\p_zP$ gives a closed Hamilton-Jacobi equation for $P$ with  initial and boundary conditions that are derived by applying the definition of $P$ in \eqref{eq:P_def} to \eqref{eq:initial} and \eqref{eq:bndry}, respectively. The complete model is, for some $T>0$,
\begin{subequations}
	\label{eq:HJ_prelim}
	\begin{align}
	\p_tP-\Phi^{(\ell)}(-\p_zP,\p_x P, \p^2_x P ;r_*,\alpha,\eta,\beta) = 0, 
	&& \, (x,z,t) \in \bbT^1\times(0,1)\times(0,T),\\
	P(x,0,t) - \int_0^1 \rho_0(x,\xi)d\xi - \int_0^t \phi_{\rm{bc}}(x,s)ds=0,
	&& \, (x,t) \in \bbT^1\times(0,T), 	\label{eq:HJ_bndry} \\
	P(x,z,0) - \int_z^1\rho_0(x,\xi)d\xi=0,
	&&  (x,z) \in \bbT^1\times(0,1), \label{eq:HJ_ic}
	\end{align}
\end{subequations}
where \eqref{eq:HJ_bndry} is derived by integrating \eqref{eq:bndry} over $z \in (0,1)$ and applying \eqref{eq:cont_unwrap}.

\begin{theorem}\label{thm:HJ_exist}
	Assume that $\alpha$ and $\rho_0$ are (i) non-negative, (ii) uniformly Lipschitz in their arguments, and (iii) periodic in $x$ (that is, $\alpha(0)=\alpha(1)$ and   $\rho_{\rm{bc}}(0,t)=\rho_{\rm{bc}}(1,t)$ for all $t\in[0,T]$). Further, assume that there is an $M$ where {$\int_0^T \phi_{\rm{bc}}(x,s)ds \leq M$} for all $x\in \bbT^1$.  Then there exists a unique, continuous, viscosity solution (in the sense of \cite{CraIshLio92}) to \eqref{eq:HJ_prelim}.
\end{theorem}

\begin{proof}  
	%\cdh{ and that $\int_z^1\rho_0(x,\xi)d\xi$ is uniformly Lipschitz as a function of $(x,z)$}
	We show that \cref{thm:Gig_comp} applies by first modifying the domain in \eqref{eq:HJ_prelim}.  We extend $\alpha$, $\rho_0$, $\rho_{\rm{bc}}$, and $\phi_{\rm{bc}}$ as functions of $x$ from $\bbT^1$ to all of $\R$ by tiling; for simplicity, in the remainder of the proof we still refer to these extensions by the same name.  The assumption that $\alpha(0)=\alpha(1)$ means that the extended version of $\alpha$ is uniformly Lipschitz on $\R.$  We then consider \eqref{eq:HJ_prelim} defined on $\Omega := \R\times(0,1)\times(0,T)$, and to more closely align with the results in the appendix, let 
	\begin{equation}
	\label{eq:PhitoH}
	H^{(\ell)}(x,z,t, P,\nabla P, \nabla^2 P) = -\Phi^{(\ell)}(-\p_zP,\p_x P, \p^2_x P ;r_*,\alpha(x),\eta,\beta)
	\end{equation}
	for $\ell \in \{0,1\}$, where $\nabla = (\p_x, \p_z)$.
	By the hypothesis on $\alpha$, both $H^{(0)}$ and $H^{(1)}$ are uniformly Lipschitz on all of $\Omega \times \R \times \R^2 \times \mathcal{S}^2$, where $\mathcal{S}^n$ is the space of all $n \times n$ symmetric matrices.  Moreover, $H^{(\ell)}$ is nonnegative, bounded by $\alpha$, and independent of the argument $P$. This means that it immediately satisfies Hypotheses 1 and 3-8 of \cref{thm:Gig_comp}. Thus the only condition of \cref{thm:Gig_comp} left to be verified is Hypothesis 2, which is the degenerate ellipticity condition on $H^{(\ell)}$.  Verifying this condition can be done in a sequence of simple steps, starting with the definitions of $w_1$ and $w_2$.  
	\begin{alignat*}{3}
	&\text{$w_1$ is non-decreasing WRT $r$ and $w_2$ is non-decreasing WRT $D^+,D^{-} $} & \quad  &\text{(see \eqref{eq:w1andw2}) } \\
	& \qquad \Longrightarrow
	\text{$w$ is non-decreasing WRT $D^+,D^{-} $} &&\text{(see \eqref{eq:w}) } \\
	& \qquad \Longrightarrow
	\text{$\Phi^{(\ell)}$ is non-decreasing WRT $\p_x^2 P $} && \text{(see \eqref{eq:Phi_def}) } \\
	& \qquad \Longrightarrow
	\text{$H^{(\ell)}$ is degenerate elliptic}  && \text{(see \eqref{eq:PhitoH}) }
	\end{alignat*}

	Finally, to invoke \cref{thm:visc_solution}, we must establish the existence of subsolutions and supersolutions as defined in \eqref{eq:sub_soln} and \eqref{eq:super_soln}, respectively. This is done by the usual construction found in, for instance, \cite[Section 2.3.2.1]{ImbSil13}.  Let 
	\begin{equation}
	P^\pm(x,z,t):=\int_z^1\rho_0(x,\xi)d\xi \pm M\pm t \cdot \sup_x \alpha(x).
	\end{equation}  
	Clearly $\p_t P^\pm=\pm\sup\alpha(x)$ and since $|\Phi_\ell| \leq |\alpha(x)|$, it follows that \begin{equation}
	\pm [\p_t P^\pm - \Phi^{(\ell)}(-\p_z P^\pm,\p_x P^\pm, \p^2_x P^\pm ;r_*,\alpha,\eta,\beta)] \geq 0.
	\end{equation}
	Thus $P^\pm$ satisfy the interior conditions in \eqref{eq:super_soln_int} and \eqref{eq:sub_soln_int}, respectively. Next write \eqref{eq:HJ_bndry} and \eqref{eq:HJ_ic} in the form $h(t,x,P,\nabla P)=0$. Then it is straight-forward to verify that $\pm h(t,x,P^\pm,\nabla P^\pm) \geq 0$. Hence $P^\pm$ satifies the parabolic boundary conditions in \eqref{eq:super_soln_int} and \eqref{eq:sub_soln_int}, respectively.  Therefore $P^-$ is a subsolution and $P^+$ is a supersolution for \eqref{eq:HJ_prelim}.  This completes the proof. 
	%Let $u(x,z,t) = P(x,z,0)-||F_0||_\infty z-||\alpha||t$ and $v(x,z,t) = P(x,z,0)+||F_0||_\infty\phi+||\alpha||t.$ 
\end{proof}

%\cdh{stuff below:  huh?}
%We note that more general assumptions on the functions $a,F,\rho_0$ may give existence and uniqueness of solutions using variants of the proofs found in the viscosity solution literature as discussed in \cref{sec:appendix}.

\begin{remark}
	In general, results regarding the regularity of solutions to \eqref{eq:HJ_prelim} using $\Phi^{(0)}$ (no such results are known by the authors for $\Phi^{(1)}$) require additional smoothness of (and possibly convexification of) $\Phi^{(0)}$ as well as other technical conditions (see \cite{Rif08}, \cite{BiaTon12}, \cite{CanFra14}).  Therefore obtaining the existence of a $L^1$ function $\rho$ solving \eqref{eq:PDEmodel} (in some generalized sense) via the existence of $P$ solving \eqref{eq:HJ_prelim} is still an open problem. 
\end{remark}

\subsection{Higher Dimensional Models}\label{sec:highdim}
%\todo[inline]{To cleanup}
Both the discrete and continuum models above can be readily extended to systems of processors arranged in an $n$-dimensional periodic lattice.   Assuming that processors only communicate with their nearest neighbors (i.e., no diagonal communication), the $n$-dimensional analog of the system formed by \eqref{eq:ODE}, \eqref{eq:Q}, and \eqref{eq:flux} is: 
\begin{align}
\frac{d q_{\bi,k}}{dt}=F_{\bi,k-1}-F_{\bi,k} , \qquad 		Q_{\bi,k}(t) = \left(\sum_{j=k}^{\kmax}q_{\bi,j}(t)\right)+\int_0^tF_{\bi,\kmax}(s)ds.\\
F_{\bi,k} = a_{\bi} \, v_1 \left( 
\min_{1\leq d\leq n} \Big\{ 
v_2 \left( q_{\bi,k},Q_{\bi-\be_d,k}-Q_{\bi,k}+q_{\bi,k},Q_{\bi+\be_d,k}-Q_{\bi,k}+q_{\bi,k} ;\beta  \right) 
\Big\}  ;q_*
\right)
\end{align} 
where $\bi = (i_1, \dots, i_n)$ is a multi-index and $(\be_d)_i = \delta_{d,i}$.   As in the one-dimensional case, $v_2$ provides the amount of available data to process, after accounting for the throttling from neighbors over a given axis.  
The multidimensional discrete model then takes the minimum over all possible axes in order to determines what is available to be processed to the next stage.
As in the one-dimensional case, self-throttling is computing use $v_1$ based on the amount of data available for processing.

If $i^{\text{max}}_d$ denotes the number of processors along the $d$ direction, we let $\veps_d=(i^{\text{max}}_d)^{-1}$; the definition of $\delta$ is unchanged.  For notational convenience, we set $V:=\prod_{d=1}^n\veps_d.$ We define $w_{2,d}$ by replacing $\eta$ in the definition of $w_2$ with $\eta_d:=\veps_d/\delta.$  Then we define the quantities analogous to those in \eqref{eq:scaling}
\begin{align}
r_* := \frac{q_*} {\delta V} \,,
\quad 
\hat{\alpha}_{\bi} := \frac{\alpha_{\bi} }{V} \,,
\quad
r_{{\bi} ,k} := \frac{q_{{\bi} ,k}}{\delta V}  \,,
\quad
R_{{\bi} ,k} 
:= \frac{1}{V} Q_{{\bi} ,k} \,,
\quad 
D^{d,\pm}_{{\bi} ,k} := \pm \frac{R_{\bi \pm \be_d,k} - R_{i,k}}{\veps_d}  \,.
\end{align}

Continuing as in Section \ref{sec:derive}, we define 
\begin{equation}
x_{\bi}=\Big( \veps_1(i_1+0.5),\dots,\veps_n(i_n+0.5) \Big),\qquad z_k = k\delta
\end{equation}
and the smooth density function $\rho(x,z,t)$ defined on $[0,1)^n\times(0,1)\times[0,\infty)$ for which $\rho(x_{\bi},z_k,t)=r_{\bi,k}(t)$.  Arguments analogous to those used to obtain \eqref{eq:D_approx} and \eqref{eq:D_to_P} give us that 
\begin{equation}
\pm D^{d,\pm}_{\bi,k}\approx\pm\p_{x_d}P(x_{\bi},z_k,t) + \frac{\veps_d}{2}\p_{x_d}^2P(x_{\bi},z_k,t).
\end{equation}
We assume that all of the $\imax_d$ are of the same order, so that the order of accuracy of the above approximation is consistent across all dimensions.

As advection in the $z$-direction is unchanged, we have the continuum model
%Here, the definition of $P$ is unchanged from before, we simply note that it is now a function on $(0,1)^n\times(0,1)\times[0,\infty)$.  Thus, the multidimensional analogue to \eqref{eq:PDEmodel} is
%\begin{align}
%\p_t\rho+\p_z\Phi_k(\rho,\nabla_xP,\nabla^2_x P;r_*,a,\eta,\beta)&=0,\qquad (x,z,t)\in T^n\times[0,1]\times[0,\infty)\\
%\rho(x,0,t) &= F_0(x,t)\\
%\rho(x,z,0) &= \rho_0(x,z)\\
%\Phi^{(0)} = aw_1\bigg(\min_{1\leq d\leq n}\Big(\big[w_{2}&(\rho,\p_{x_d}P,-\p_{x_d}P;\eta,\beta)\big]\Big);r_*\bigg)\\
%\Phi^{(1)} = aw_1\bigg(\min_{1\leq d\leq n}\Big(\big[w_{2}&(\rho,\p_{x_d}P+\frac{\veps_d}{2}\p_{x_d}^2P,-\p_{x_d}P+\frac{\veps_d}{2}\p_{x_d}^2P;\eta,\beta)\big]\Big);r_*\bigg)
%\end{align}

\begin{subequations}
	\label{eq:PDEmodel_n}
	\begin{align}
	&\p_t \rho + \p_z \Phi^{(\ell)}(\rho,\nabla_xP,\nabla_x^2P;r_*,\alpha,\vec\eta,\beta) = 0,
	&& \,(x,z,t)\in \bbT^n\times(0,1)\times(0,\infty),\\
	&\rho(x,0,t) = \rho_{\rm{bc}}(x,t),
	&& \, (x,t) \in \bbT^n\times(0,\infty),\\
	&\rho(x,z,0) = \rho_0(x,z), 
	&& \, (x,z) \in \bbT^n\times(0,1)  ,
	\end{align}
\end{subequations}
where  $\mathbb{T}^n$ denotes the $n$-dimensional torus parameterized by $[0,1)^n$, $P$ and $\phi^{(\ell)}$ are defined as in \eqref{eq:P_def} and \eqref{eq:phi_ell_def}, respectively, and the form of $\Phi^{(\ell)}$, $\ell \in \{0,1\}$ is a slightly generalized version of \eqref{eq:Phi_def}.
%\begin{subequations}
%	\begin{align}
%	P(x,z,t) &= \int_z^1\rho(x,\xi,t)d\xi 
%	+ \int_0^t \Phi^{(\ell)}(x,1,s) ds,  \\
%	\phi^{(\ell)}(x,z,t) &= \Phi^{(\ell)} \left(\rho(x,z,t),\p_x  P(x,z,t), \p^2_x P(x,z,t) ;r_*,a,\eta,\beta \right),
%	\end{align}
%\end{subequations}
\begin{subequations}
	\begin{align}
	&\Phi^{(0)}(\rho,\nabla_xP,\nabla_x^2P;r_*,\alpha,\vec\eta,\beta)  = \alpha \, w_1\bigg(\min_{1\leq d\leq n}\Big(\big[w_{2}(\rho,\p_{x_d}P,-\p_{x_d}P;\eta_d,\beta)\big]\Big);r_*\bigg)\\
	&\Phi^{(1)}(\rho,\nabla_xP,\nabla_x^2P;r_*,\alpha,\vec\eta,\beta)  = \\
	& \qquad \qquad \alpha \, w_1\bigg(\min_{1\leq d\leq n}\Big(\big[w_{2}(\rho,\p_{x_d}P+\frac{\veps_d}{2}\p_{x_d}^2P,-\p_{x_d}P+\frac{\veps_d}{2}\p_{x_d}^2P;\eta_d,\beta)\big]\Big);r_*\bigg)
	\end{align}
\end{subequations}

For notational convenience we have defined $\Phi^{(\ell)}$ using the full tensor $\nabla_x^2P$; however, we note that the flux functions do not depend on mixed second derivatives. The procedure used in Section \ref{sec:HJB} to obtain a Hamilton-Jacobi equation for $P$ can be repeated here; the only changes are (i) the multi-dimensional version of $\Phi^{(\ell)}$ in \eqref{eq:Phi_def} and (ii) the domain of the $x$ variable.    Verifying that these newly-defined flux functions  satisfy the conditions of Theorem \ref{thm:Gig_comp}  is essentially the same as before.  Existence and uniqueness of viscosity solutions $P$ then follow.

%If we wish to further generalize this model, by including diagonal communication between processors, we must first identify the directions over which these additional communications will take place.  This results in additional directions over which we must evaluate $w_2$ to determine the amount of available data. For each of these directions, $\nu\in S^{n-1}$, we must determine a corresponding length scale $\veps_\nu$ and then evaluate $w_2$ using the $\nu$-directional derivative and the corresponding $\eta_\nu$ obtained by the same scaling approach as above.  
%The model will then need to determine the available data for processing by taking the minimum, again, over all directions $\nu$.  
%If we assume connections in all directions, we have, a flux function $\Phi^{(0)}$, for instance as
%\begin{equation}
%\Phi^{(0)}(\rho,\nabla_xP,\nabla_x^2P;r_*,\alpha,\vec\eta,\beta)  = \alpha \, w_1\bigg(\min_{\nu\in S^{n-1}}\Big(\big[w_{2}(\rho,\p_{\nu}P,-\p_{\nu}P;\eta_\nu,\beta)\big]\Big);r_*\bigg)
%\end{equation}

%T\input{numerics}
\section{Numerical Simulations}
\label{sec:numerics}

In this section, we perform numerical simulations of the one dimensional processor system in order to (i) test the ability of the macroscopic model to approximate the discrete model when $\veps$ and $\delta$ are small and (ii) explore how model parameters affect the model output.    All simulations are based on the flux $\Phi^{(0)}$, although results with $\Phi^{(1)}$ demonstrate similar characteristics. Problem data is specified in terms continuum model of continuum models quantities.  These quantities are translated back to discrete model quantities in order to implement ODE simulations.   

\subsection{ODE Implementation}
The explicit two-step Adams-Bashforth (Section III of \cite{HaiNorWan93}) is used to simulate the discrete model formed by {\eqref{eq:ODE}, \eqref{eq:Q}, and \eqref{eq:flux}.    Given $\eta$, values $\imax$ and $\kmax$ are chosen so that $\kmax / \imax = \eta$ (cf. \eqref{eq:scale_params}).   We then compute a solution to the discrete model as follows. Using \eqref{eq:scaling} and \eqref{eq:ODEtogrid}, we convert $r_*,a, \rho_0, \rho_{bc}$ to their discrete counterparts:
	\begin{equation}
	q_* = \veps \delta r_* , \qquad q_{i,k}(0) = \veps \delta \rho_0 (x_i,z_k) , \qquad q_{i,0}(t) = \rho_{\rm{bc}}(x_i,t), \quad a_i = \veps \alpha(x_i).
	\end{equation}
	This discrete model data is used to set the time step:
	\begin{equation}
	\Delta t = \frac{q_*}{2\big(\max_ia_i\big)\sqrt{\imax\kmax}}.
	\end{equation}	
	The outflow at $F_{i,\kmax}$ is tracked and accumulated over time in order to compute $Q_{i,k} $ from \eqref{eq:Q}.   At the final time $T$, the result of the explicit time stepping is converted back, via the formula in \eqref{eq:scaling}, i.e.,  $r_{i,k}(T) = (\veps \delta)^{-1} q_{i,k}(T)$.
	In order to compare this against solutions to the continuum model (see below), we use these point-wise values to generate a piecewise constant function $r$ over the cells $C_{i,k} = (x_i-.5\veps,x_i+.5\veps)\times (z_k,z_k+\delta)$:  
	\begin{equation}
	r(x,z) = \sum_{i,k} r_{i,k} \chi_{C_{i,k}}(x,z).
	\end{equation}

	\subsection{Hamilton-Jacobi Implementation}
	The Hamilton Jacobi equation \eqref{eq:HJ_prelim} is solved numerically using a fifth-order WENO interpolation in $x$ and $z$ and the optimal third-order SSP Runge-Kutta method for time integration. Details of these algorithms can be found in Sections 3.2 and 6, respectively, of \cite{shu2007high}.   Once a numerical solution for $P$ is computed, we again use WENO interpolation to approximate $\rho$ via the relation
	$\rho(x,z,t) = - \partial_zP(x,z,t)$.
	
	To condense the notation, let $\sigma = \p_x P$, $\tau = \p_z P$ and $\upsilon=\p_{xx} P$.  Then for fixed $r_*$, $\alpha$, $\eta$, $\beta$, and $\ell$, let $H(\sigma, \tau, \upsilon) =  -\Phi^{(\ell)}(-\tau, \sigma,\upsilon; r_*,a,\eta,\beta)$.
	The numerical solution for $P$ is computed on a grid $\{x_n, z_m \}$ where
	\begin{align}
	&x_n = n\Delta x, & &\ n=1,2,\dots,N,\quad & &\Delta x = {N} ^{-1},\\
	&z_m = m\Delta z, & &\ m=1,2,\dots,M,\quad & &\Delta z = {M}^{-1}.
	\end{align}
	The semi-discrete method for the grid function $P_{n,m}(t) \approx P(x_n,z_m,t)$ is 
	\begin{equation}
	\frac{d}{dt}P_{n,m}(t) = - \hat{H}(\sigma_{n,m}^-,\sigma_{n,m}^+,\tau_{n,m}^-,\tau_{n,m}^+;\upsilon_{n,m}),
	\end{equation}
	where the numerical approximations $\sigma_{n,m}^{\pm} \approx \sigma(x_n^ {\pm}, z_m)$ and $\tau_{n,m}^{\pm}  \approx \tau(x_n, z_m^ {\pm})$ are obtained via  WENO interpolation and  $\upsilon_{n,m}  \approx \upsilon(x_n, z_m)$ is computed by central difference.  The numerical flux function $\hat{H}$, based on the global Lax-Friedrichs flux:
	\begin{equation}\hat{H}(\sigma^-,\sigma^+,\tau^-,\tau^+;\upsilon) = H\left( \frac{\sigma^-+\sigma^+}{2} , \frac{\tau^-+\tau^+}{2}, \upsilon \right) - \frac{1}{2}\lambda^x(\sigma^+-\sigma^-) - \frac{1}{2}\lambda^z(\tau^+-\tau^-),
	\end{equation}
	where
	\begin{equation}
	\lambda^x = \max_{\sigma,\tau}|H_\sigma| = \frac{\alpha\eta}{\beta r_*},\quad \lambda^z =\max_{\sigma,\tau}|H_\tau| = \frac{\alpha}{\beta r_*}.
	\end{equation}
	The time step for the SSP integrator is given by
	\begin{equation}\label{eq:SSPstep}
	\Delta t  \left( \frac{\lambda^x}{\Delta_x} + \frac{\lambda^z}{\Delta z} \right) \leq 0.6 .
	\end{equation}

	\subsection{Experiments}
	We perform a sequence of exploratory experiments below, modifying the parameters $\eta$ and $\beta$, as well as the throughput function $\alpha$.  In all cases, $\alpha$, $\rho_0$, and $\rho_{\rm{bc}}$ are periodic with respect to $x$ and the parameter $r_*=1$.     Results are presented as two-dimensional color maps or line-outs in the $z$ direction.  In all figures, the horizontal axis corresponds to the $z$-axis.  Profiles of $\alpha$ for each experiment are depicted in \cref{fig:alphas}.

	\begin{figure}
		\begin{subfigure}{.24\linewidth}
						\includegraphics[width=\textwidth]{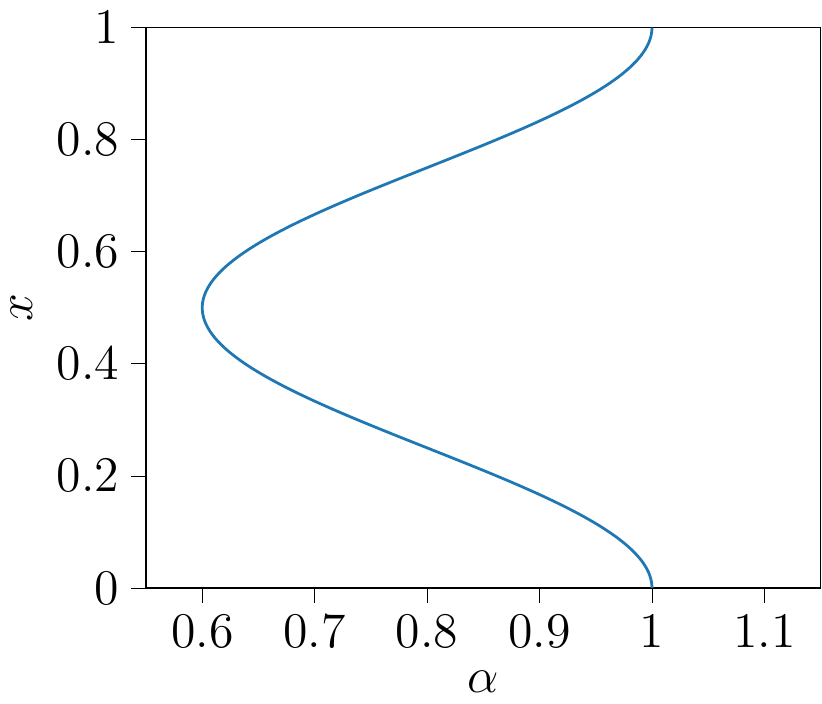}
			\caption{Example 1}
			\label{fig:alpha1_3}
		\end{subfigure}\hfill
		\begin{subfigure}{.24\linewidth}
			\includegraphics[width=\textwidth]{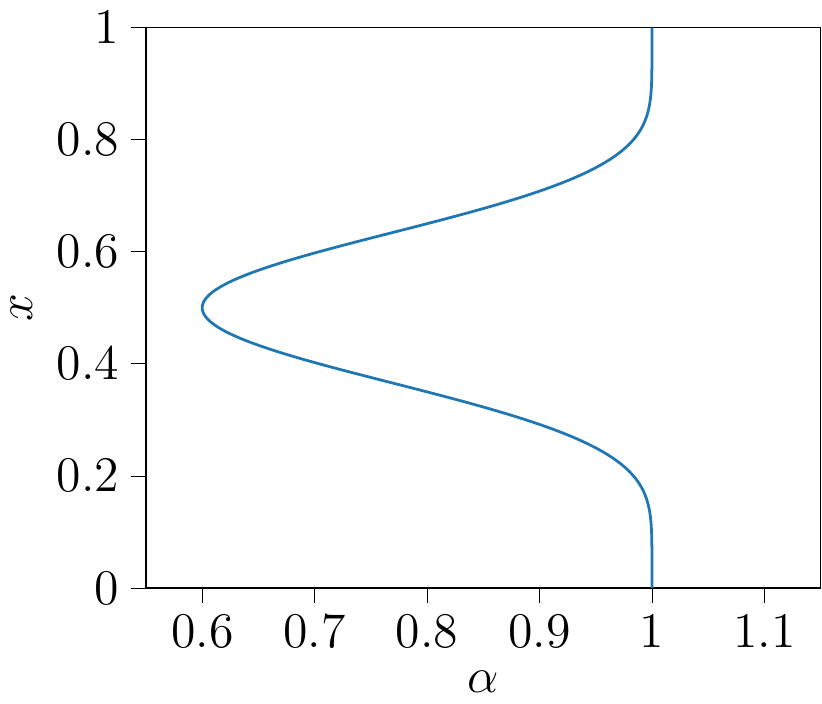}
			\caption{Examples 2-3}
			\label{fig:alpha2_3}	
		\end{subfigure}\hfill
		\begin{subfigure}{.24\linewidth}
			\includegraphics[width=\textwidth]{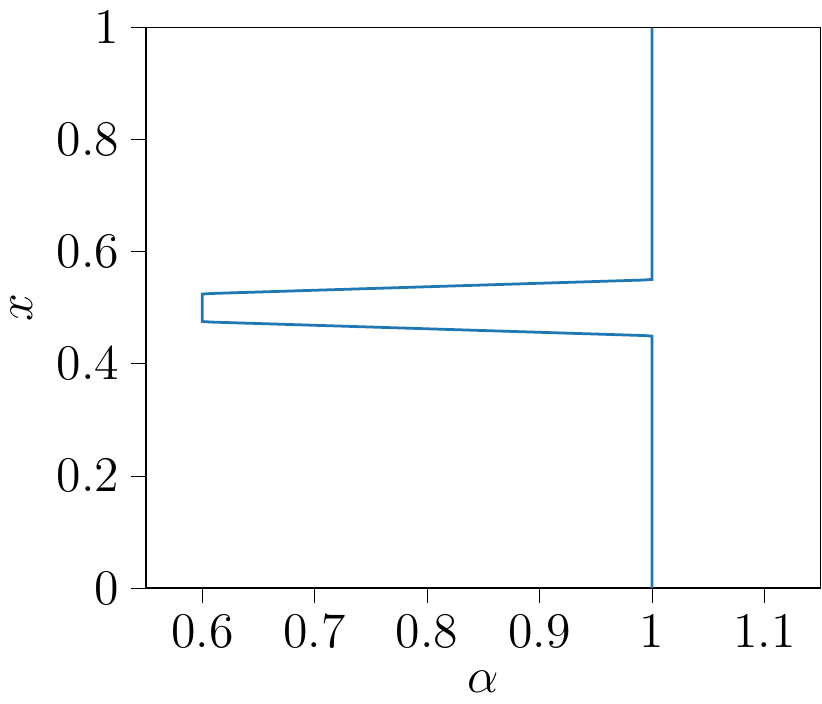}
			\caption{Example 4}
			\label{fig:alpha4}
		\end{subfigure}\hfill
		\begin{subfigure}{.24\linewidth}
			\includegraphics[width=\textwidth]{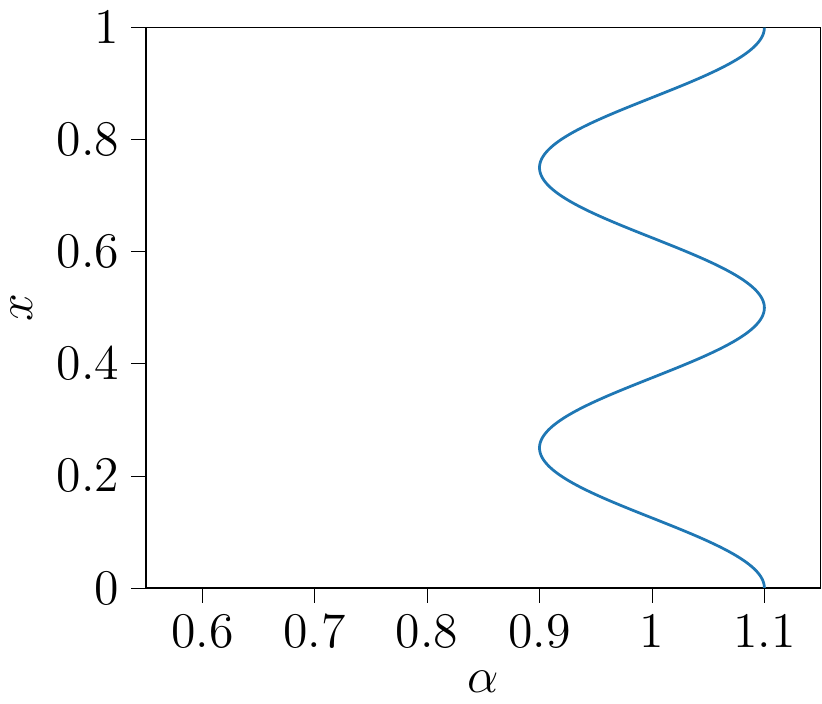}
			\caption{Example 5}
			\label{fig:alpha5}
		\end{subfigure}
		\caption{Profiles of the processor speed $\alpha$ used in the numerical experiments.  The non-standard orientation of the graphs is set to match the axes in the numerical results that follow.}
		\label{fig:alphas}
	\end{figure}
	
	\begin{example}[Agreement between models] The purpose of this example is to demonstrate that the macroscopic model approximates the microscopic model when $\veps$ and $\delta$ are sufficiently small.  We set $\beta=1$ and consider $\eta \in \{ 0.2,1,5 \}$. The initial condition, boundary condition, and processor speed are given by 
		\begin{equation}
		\rho_0(x,z) = 1.5\left(\sin( 2\pi z) \right)^6 \,\chi_{[0,0.5]}(z),
		\quad \rho_{\rm{bc}}(x,t)=0, 
		\quad \alpha(x) = 1-0.4(\sin(\pi x))^2,
		\end{equation}
		respectively.   Both models are simulated up to a final time $t=0.5$.

		For this example, the Hamilton-Jacobi simulation is performed with a $1000 \times 1000$ mesh and a time step chosen according to \eqref{eq:SSPstep} in order to generate a highly resolved numerical solution of the macroscopic model.    For the microscopic model, we use $\imax=1000$ and $\kmax=200$ when $\eta=0.2$, $\imax=\kmax=500$ when $\eta=1$, and $\imax=200$ and $\kmax=1000$ when $\eta=5$.  These solutions to the microscopic model are then used to obtain the piecewise-constant function $r$ on the $1000\times 1000$ mesh from the Hamilton-Jacobi simulation.
		
		Numerical results for $\eta=0.2$, $\eta=1.0$, and $\eta=5.0$ are shown in  \cref{fig:ODEcomparetime02}, \cref{fig:ODEcomparetime1}, and \cref{fig:ODEcomparetime5}, respectively.  
		While the results demonstrate general qualitative agreement between the models, discrepancies develop over time, especially for smaller values of $\eta$; see \cref{fig:Test1Eta02LateError,fig:Test1Eta02LateLine}.    For the worst case scenario ($\eta=0.2$), we increase the size of the discrete model by a factor of $2.5$ ( giving $\imax = 2500$ and $\kmax=500$), at which point the discrepancy between models decreases noticeably; see  \cref{fig:ODEcompareFine}.

		\begin{figure}[h!]
			\begin{subfigure}{.32\linewidth}
				\includegraphics[width=\textwidth]{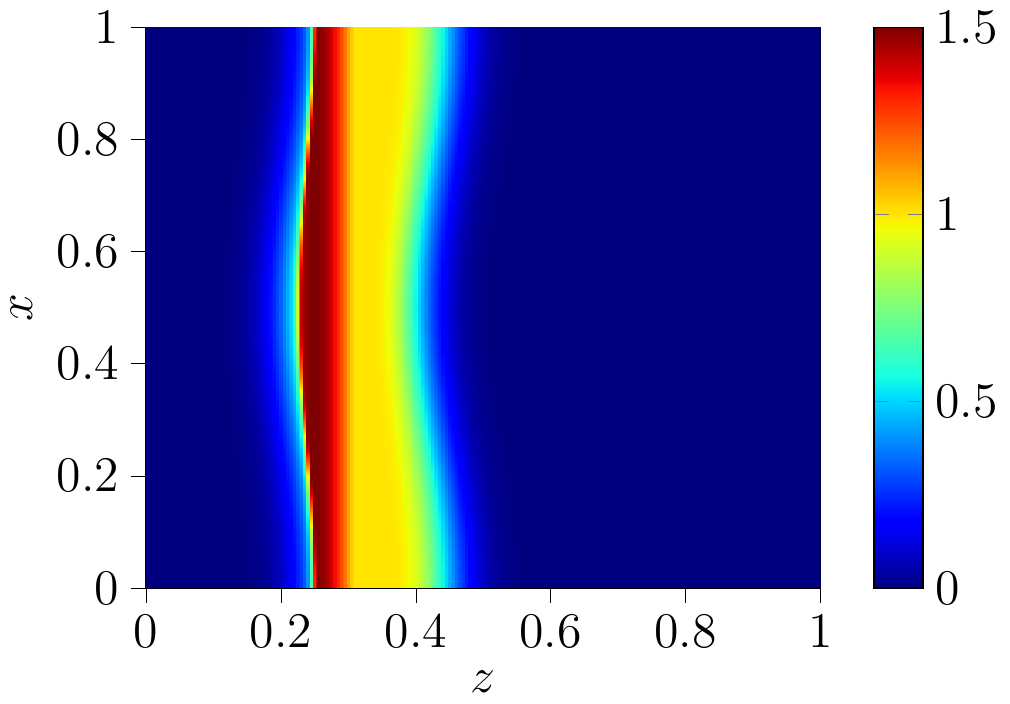}
				\caption{$r$ at $t=0.1$}	
				\label{fig:Test1Eta02EarlyODE}
			\end{subfigure}\hfill
			\begin{subfigure}{.32\linewidth}
				\includegraphics[width=\textwidth]{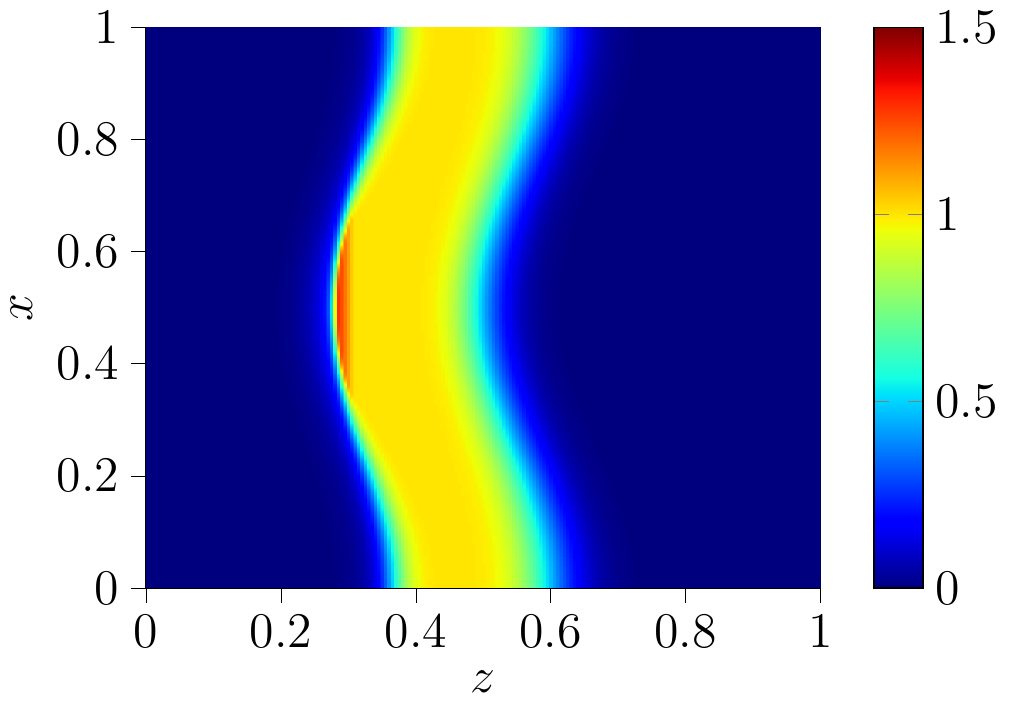}
				\caption{$r$ at $t=0.25$}	
				\label{fig:Test1Eta02MidODE}
			\end{subfigure}\hfill
			\begin{subfigure}{.32\linewidth}
				\includegraphics[width=\textwidth]{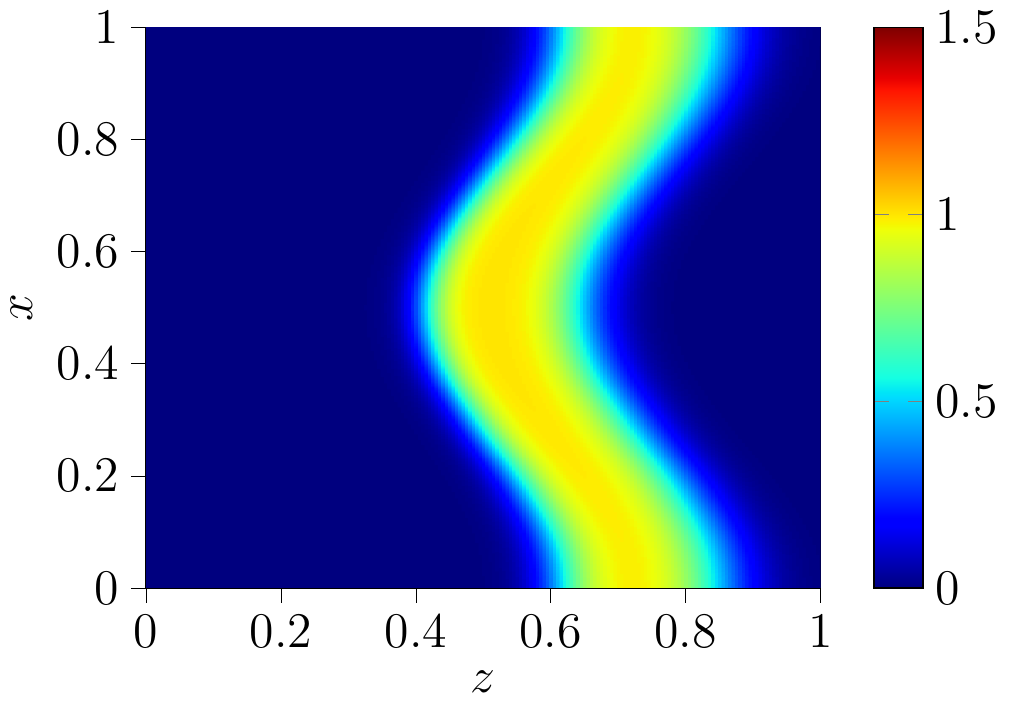}
				\caption{$r$ at $t=0.5$}	
				\label{fig:Test1Eta02LateODE}
			\end{subfigure}\\
			\begin{subfigure}{.32\linewidth}
				\includegraphics[width=\textwidth]{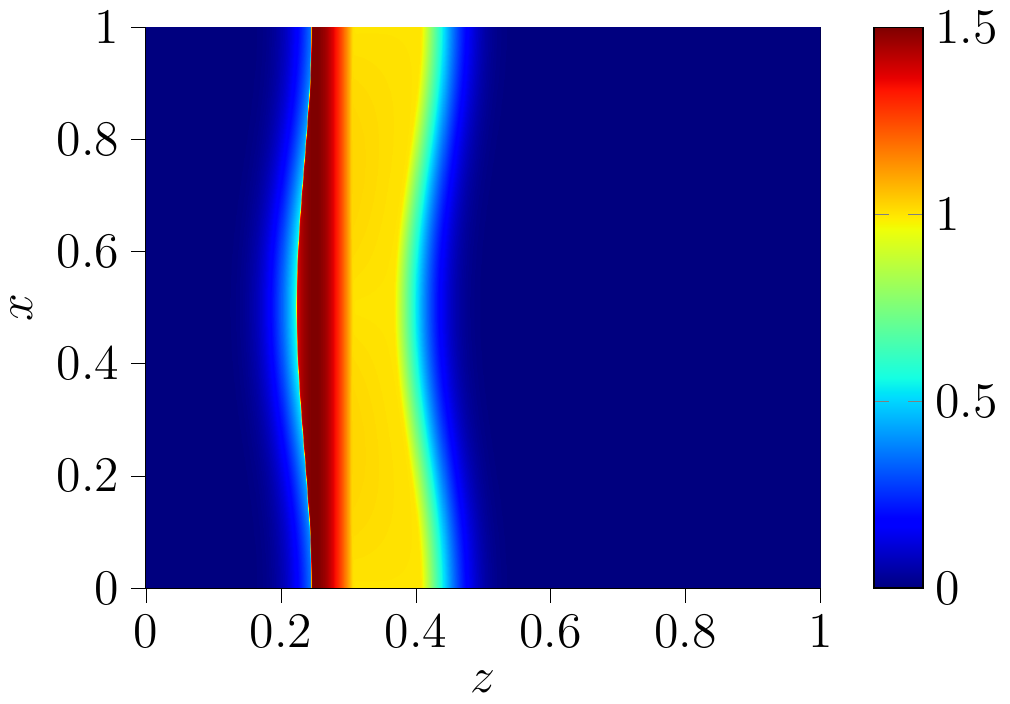}
				\caption{$\rho$ at $t=0.1$ }
				\label{fig:Test1Eta02EarlyPDE}
			\end{subfigure}\hfill
			\begin{subfigure}{.32\linewidth}
				\includegraphics[width=\textwidth]{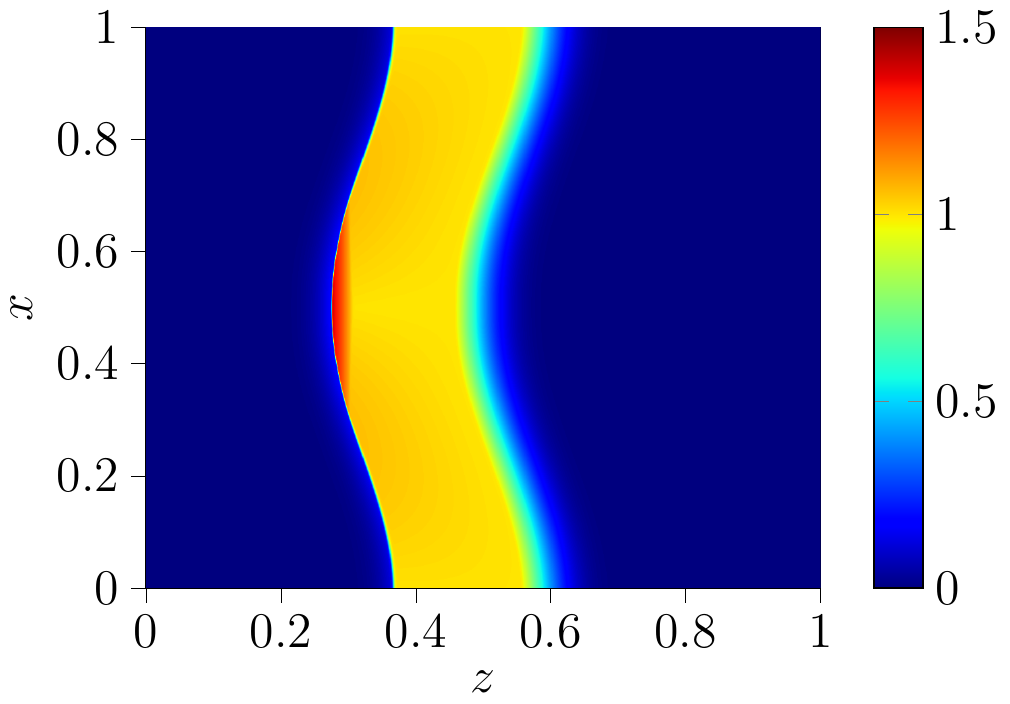}
				\caption{$\rho$ at $t=0.25$ }
				\label{fig:Test1Eta02MidPDE}
			\end{subfigure}\hfill
			\begin{subfigure}{.32\linewidth}
				\includegraphics[width=\textwidth]{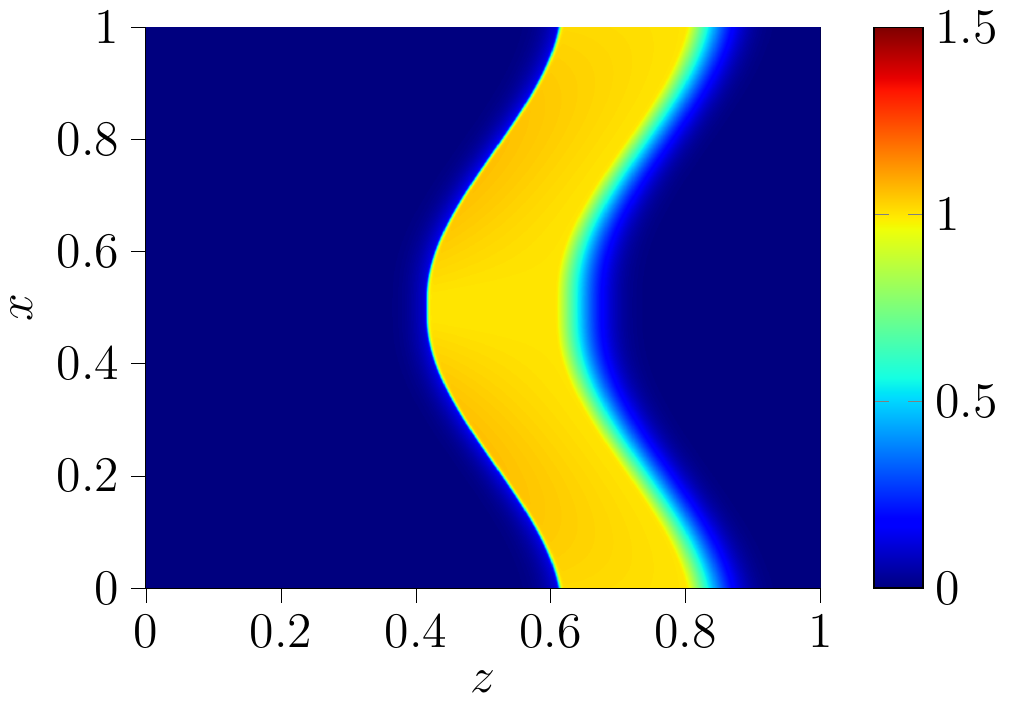}
				\caption{$\rho$ at $t=0.5$ }
				\label{fig:Test1Eta02LatePDE}
			\end{subfigure}\hfill
			\begin{subfigure}{.32\linewidth}
				\includegraphics[width=\textwidth]{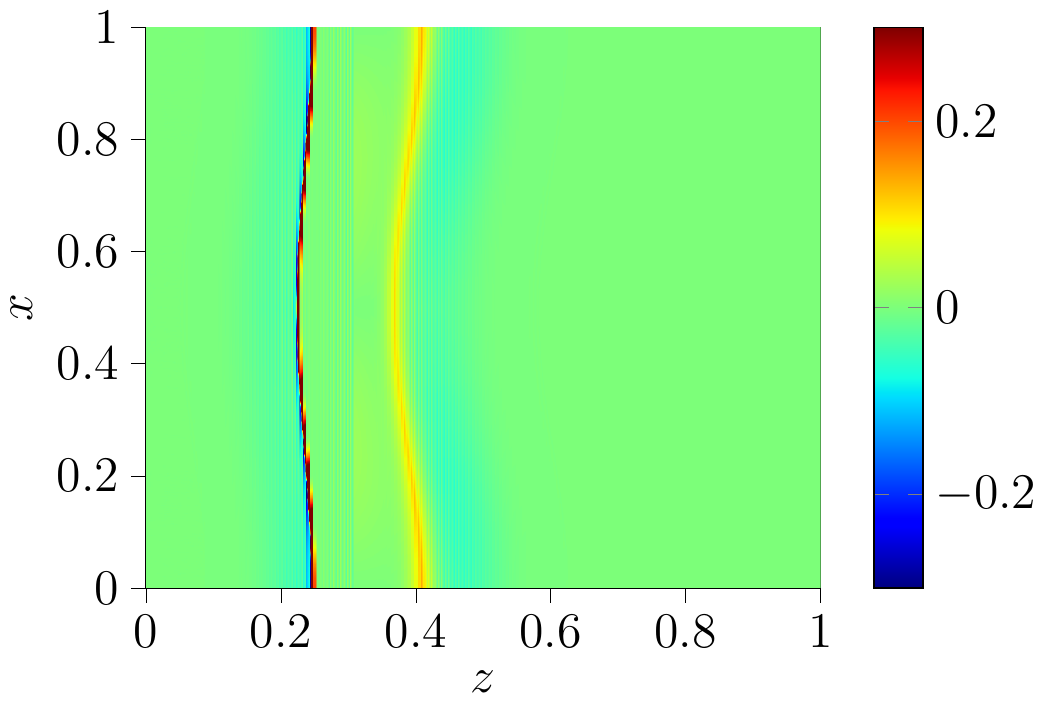}
				\caption{$\rho-r$ at $t=0.1$}
				\label{fig:Test1Eta02EarlyError}
			\end{subfigure}\hfill
			\begin{subfigure}{.32\linewidth}
				\includegraphics[width=\textwidth]{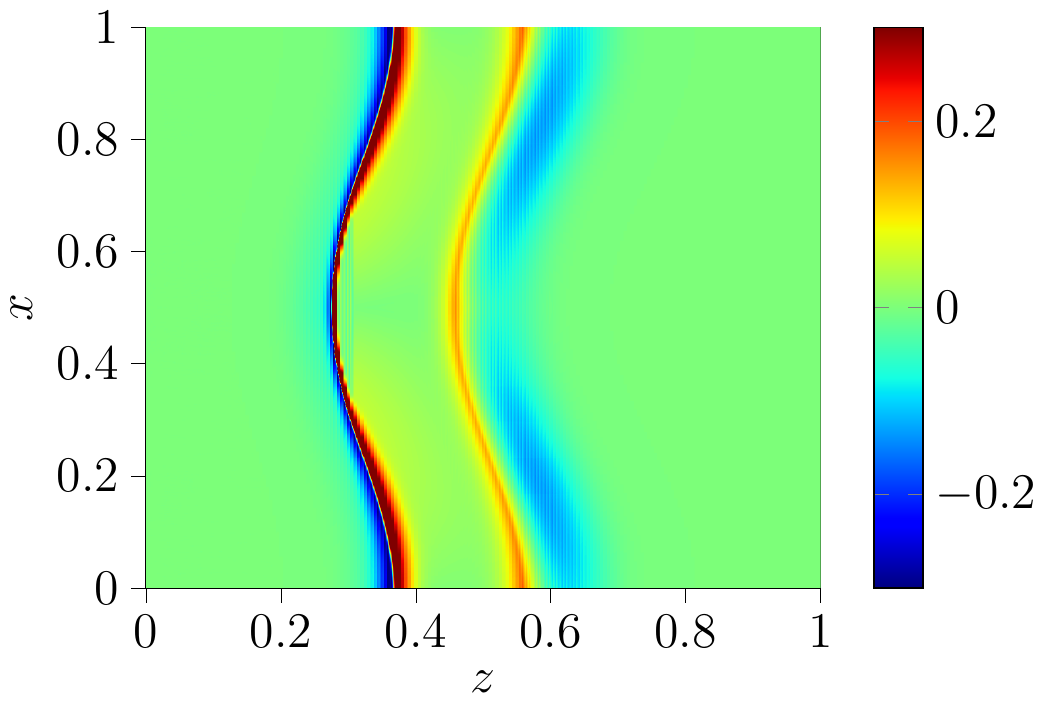}
				\caption{$\rho-r$ at $t=0.25$}
				\label{fig:Test1Eta02MidError}
			\end{subfigure}\hfill
			\begin{subfigure}{.32\linewidth}
				\includegraphics[width=\textwidth]{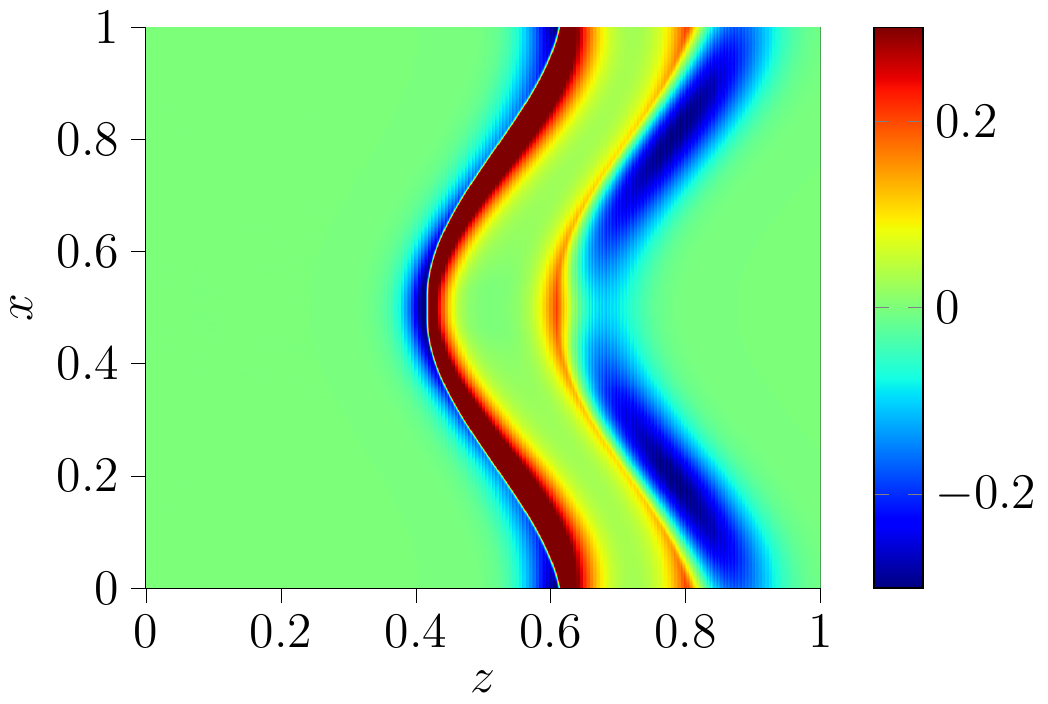}
				\caption{$\rho-r$ at $t=0.5$}
				\label{fig:Test1Eta02LateError}
			\end{subfigure}\\
			\begin{subfigure}{.3\linewidth}
				\includegraphics[width=\textwidth]{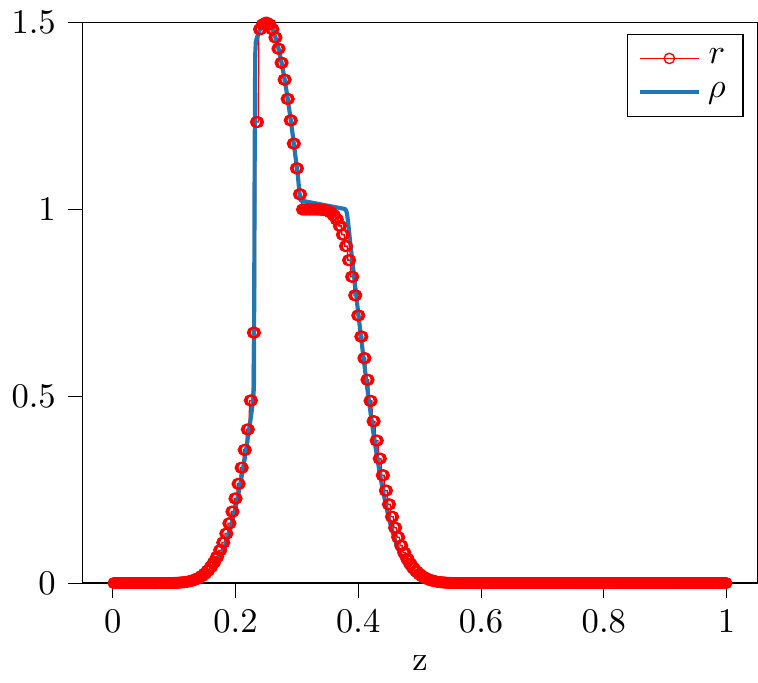}
				\caption{$\rho$,$r$ at $(x,t)=(0.3,0.1)$}
				\label{fig:Test1Eta02EarlyLine}
			\end{subfigure}\hfill
			\begin{subfigure}{.3\linewidth}
				\includegraphics[width=\textwidth]{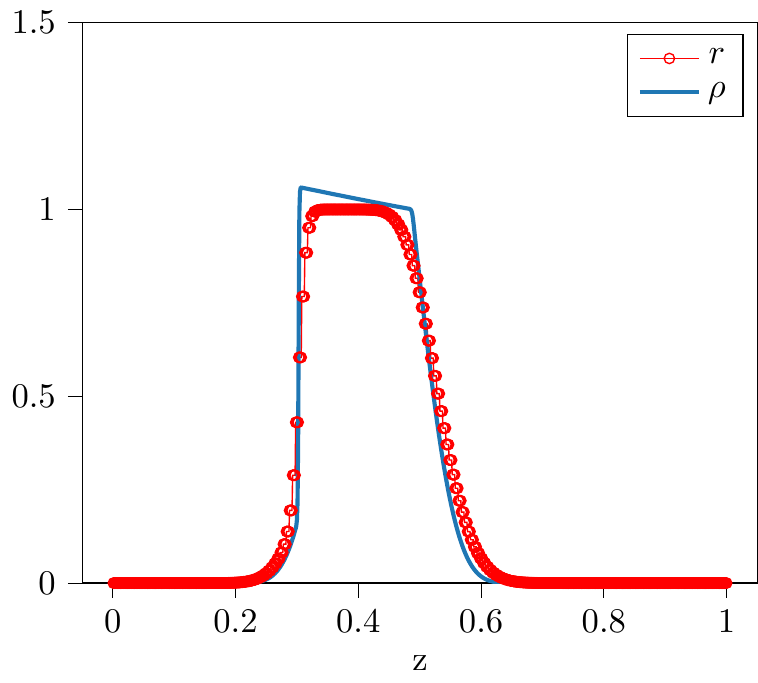}
				\caption{$\rho$,$r$ at $(x,t)=(0.3,0.25)$}
				\label{fig:Test1Eta02MidLine}
			\end{subfigure}\hfill
			\begin{subfigure}{.3\linewidth}
				\includegraphics[width=\textwidth]{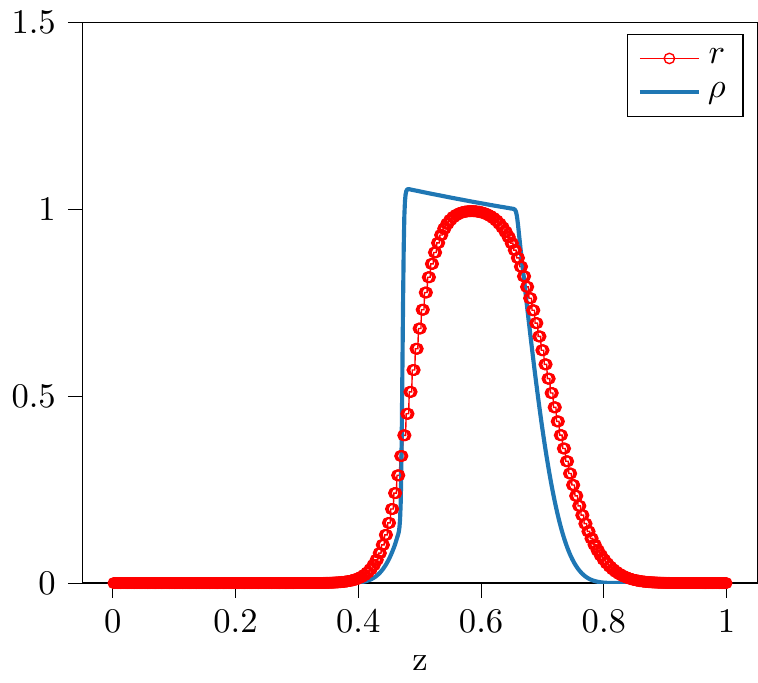}
				\caption{$\rho$,$r$ at $(x,t)=(0.3,0.5)$}
				\label{fig:Test1Eta02LateLine}
			\end{subfigure}
			\caption{Discrete solution $r$ and continuum solution $\rho$ when $\eta=0.2$.  From left to right, columns correspond to solutions
				at $t=0.1$, $t=0.25$, and $t=0.5$  Discrete solution is computed with $(\imax, \kmax) = (1000,200)$.
				Continuum solution is computed on a $10^3 \times 10^3$ mesh.  
			}
			\label{fig:ODEcomparetime02}
		\end{figure}

		\begin{figure}[h!]
			\begin{subfigure}{.32\linewidth}
				\includegraphics[width=\textwidth]{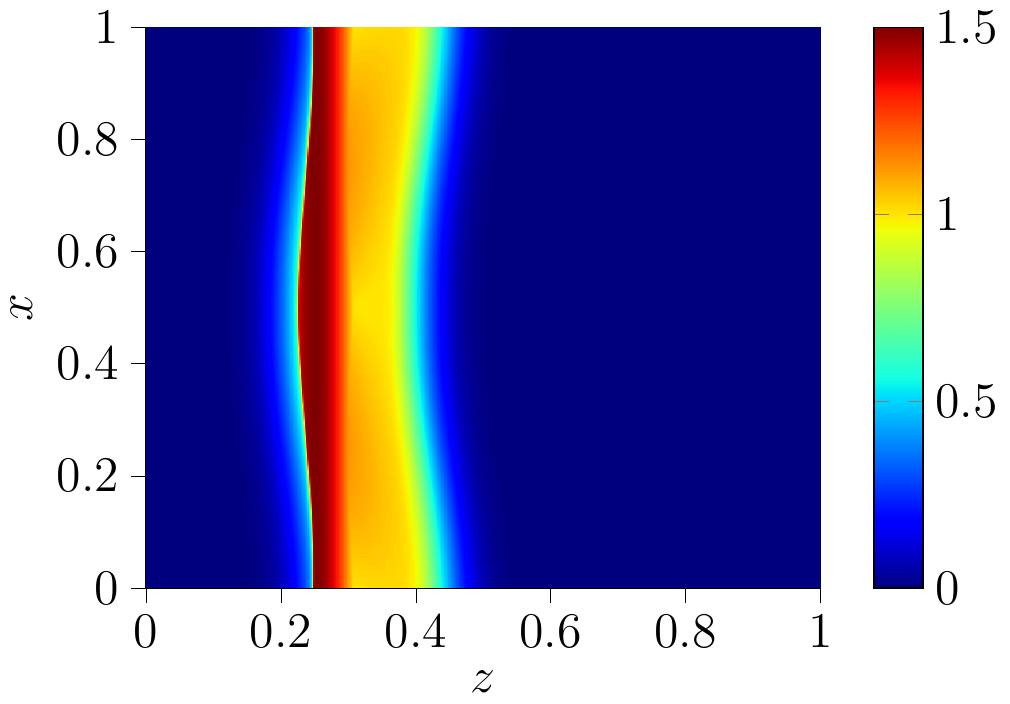}
				\caption{$r$ at $t=0.1$}	
				\label{fig:Test1Eta1EarlyODE}
			\end{subfigure}\hfill
			\begin{subfigure}{.32\linewidth}
				\includegraphics[width=\textwidth]{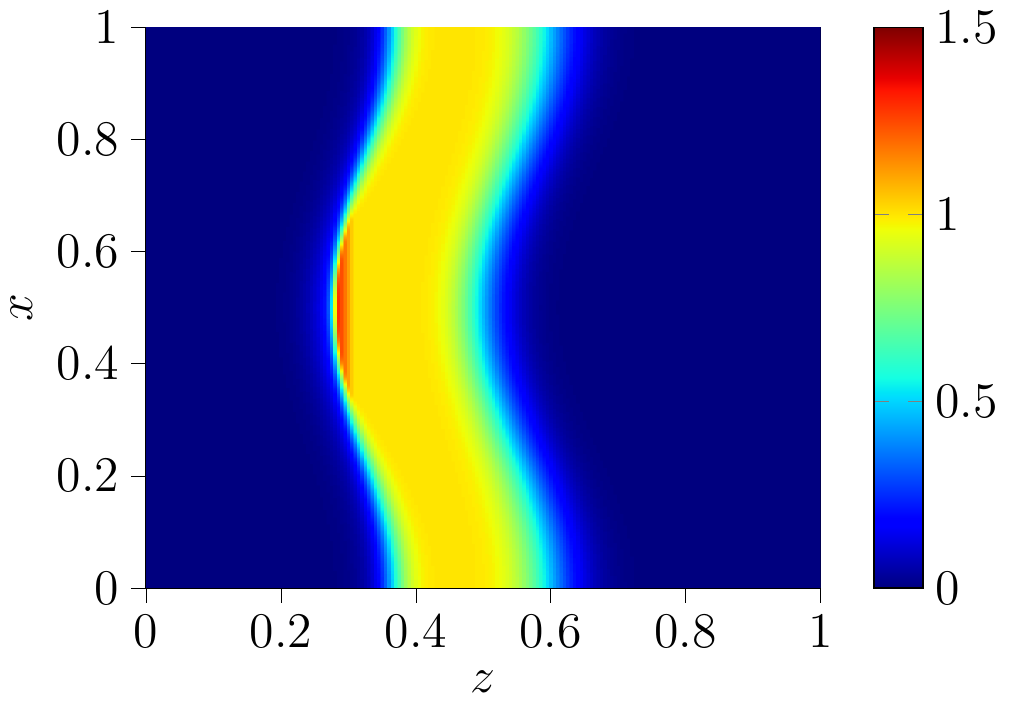}
				\caption{$r$ at $t=0.25$}	
				\label{fig:Test1Eta1MidODE}
			\end{subfigure}\hfill
			\begin{subfigure}{.32\linewidth}
				\includegraphics[width=\textwidth]{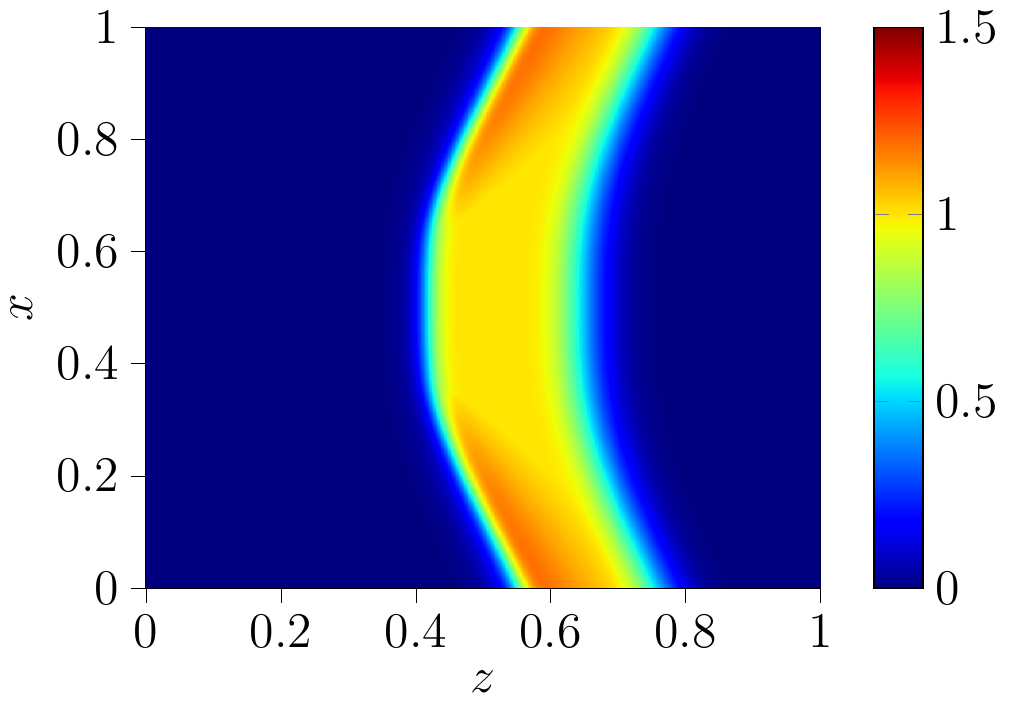}
				\caption{$r$ at $t=0.5$}	
				\label{fig:Test1Eta1LateODE}
			\end{subfigure}\\
			\begin{subfigure}{.32\linewidth}
				\includegraphics[width=\textwidth]{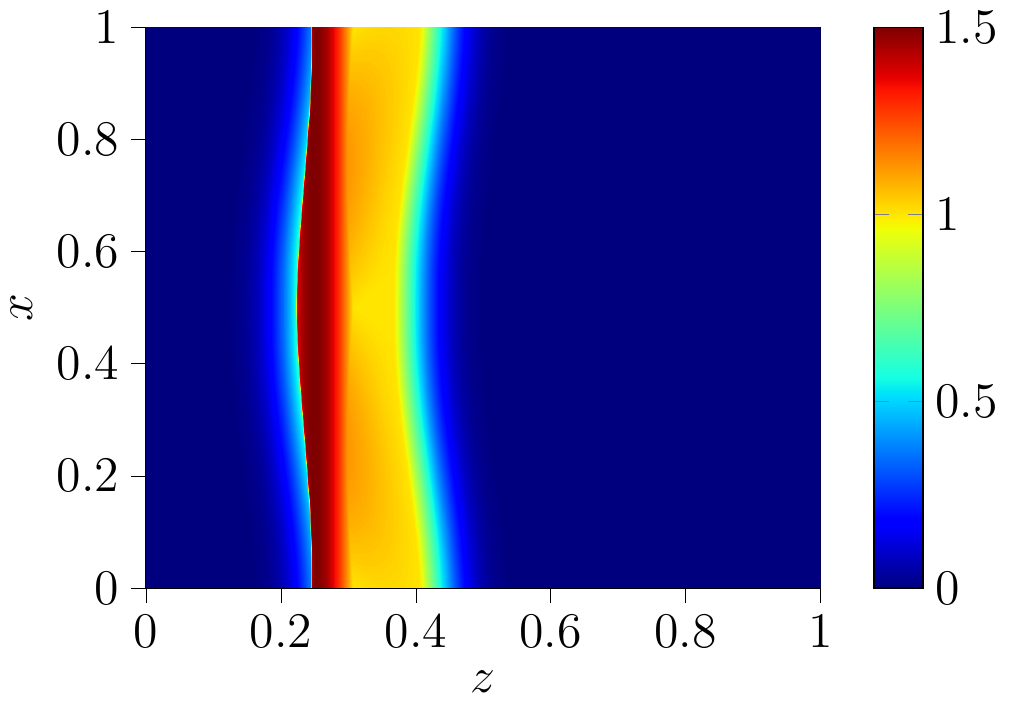}
				\caption{$\rho$ at $t=0.1$ }
				\label{fig:Test1Eta1EarlyPDE}
			\end{subfigure}\hfill
			\begin{subfigure}{.32\linewidth}
				\includegraphics[width=\textwidth]{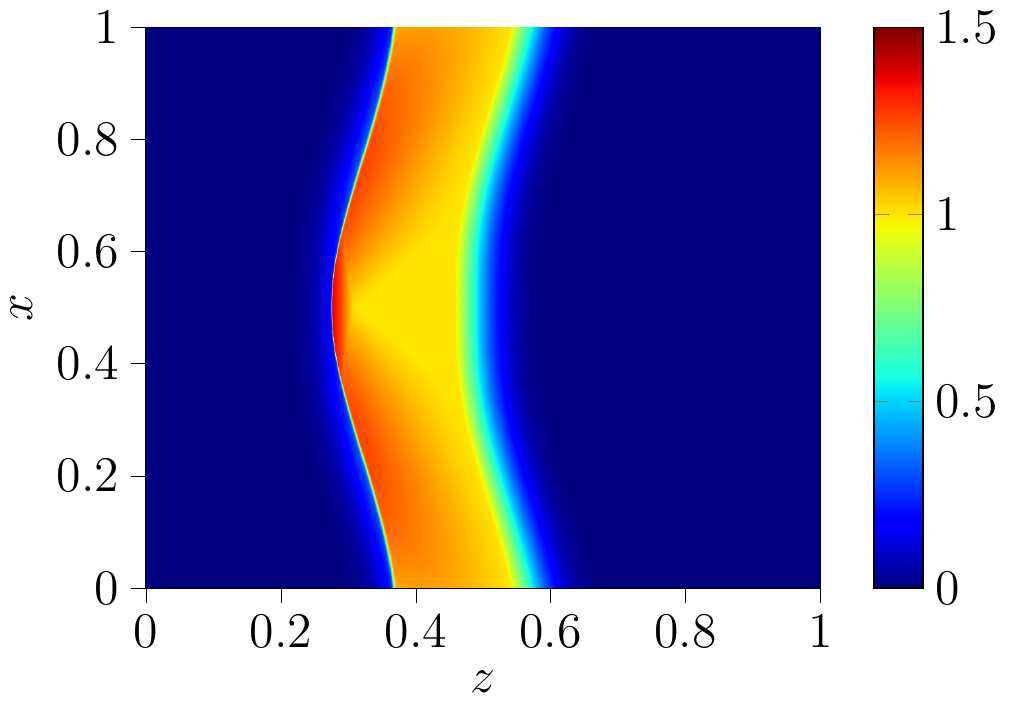}
				\caption{$\rho$ at $t=0.25$ }
				\label{fig:Test1Eta1MidPDE}
			\end{subfigure}\hfill
			\begin{subfigure}{.32\linewidth}
				\includegraphics[width=\textwidth]{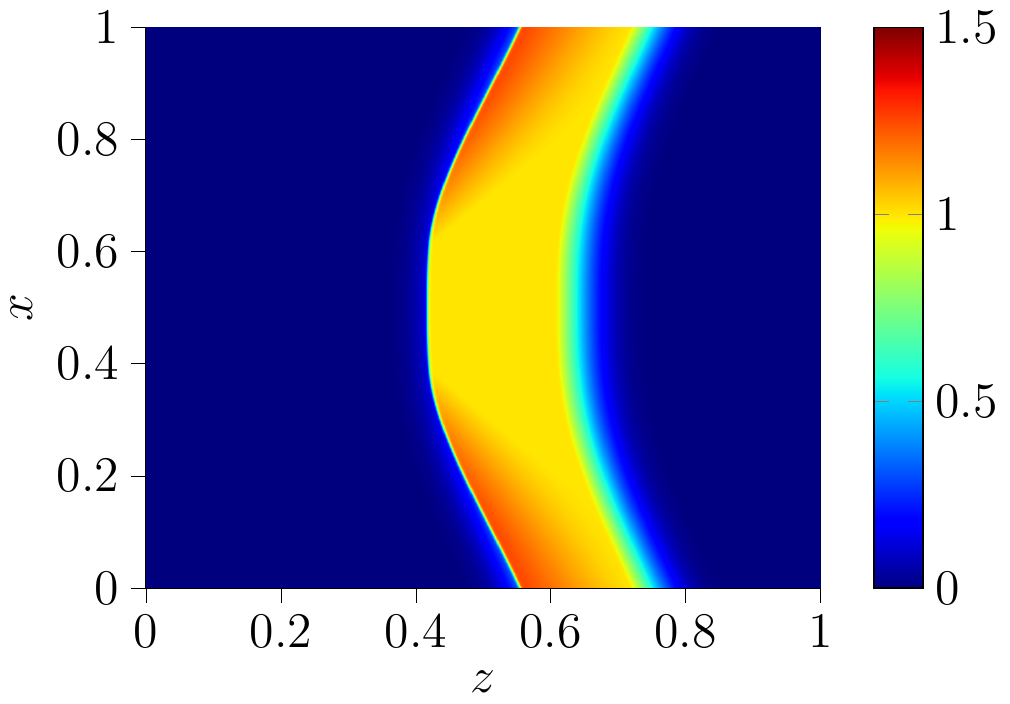}
				\caption{$\rho$ at $t=0.5$ }
				\label{fig:Test1Eta1LatePDE}
			\end{subfigure}\hfill
			\begin{subfigure}{.32\linewidth}
				\includegraphics[width=\textwidth]{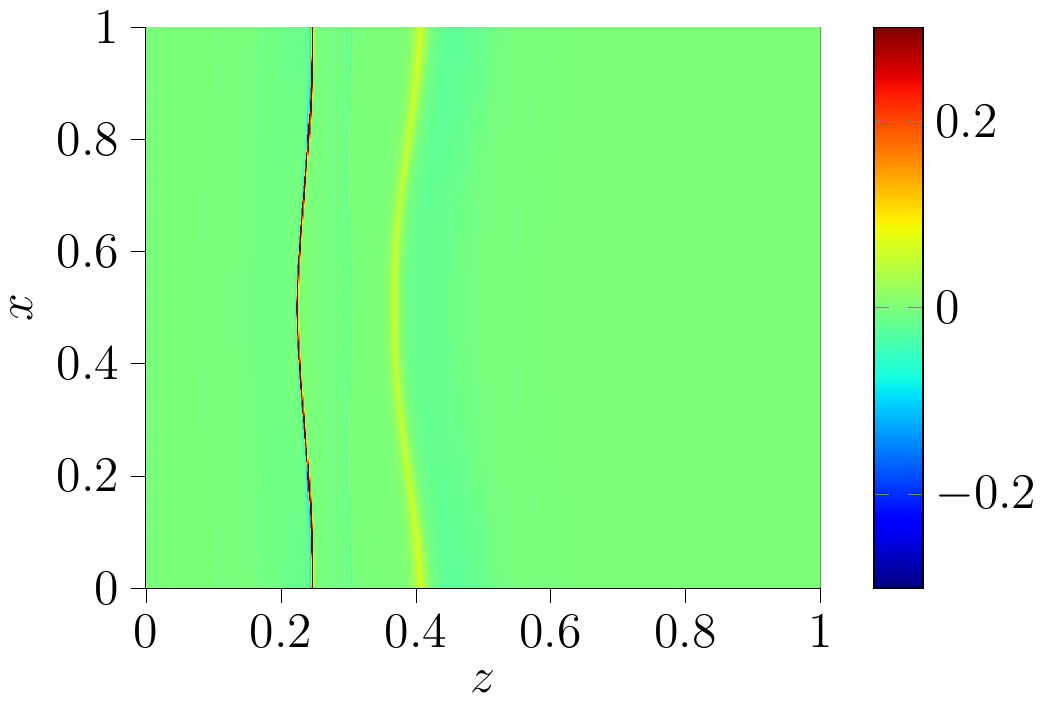}
				\caption{$\rho-r$ at $t=0.1$}
				\label{fig:Test1Eta1EarlyError}
			\end{subfigure}\hfill
			\begin{subfigure}{.32\linewidth}
				\includegraphics[width=\textwidth]{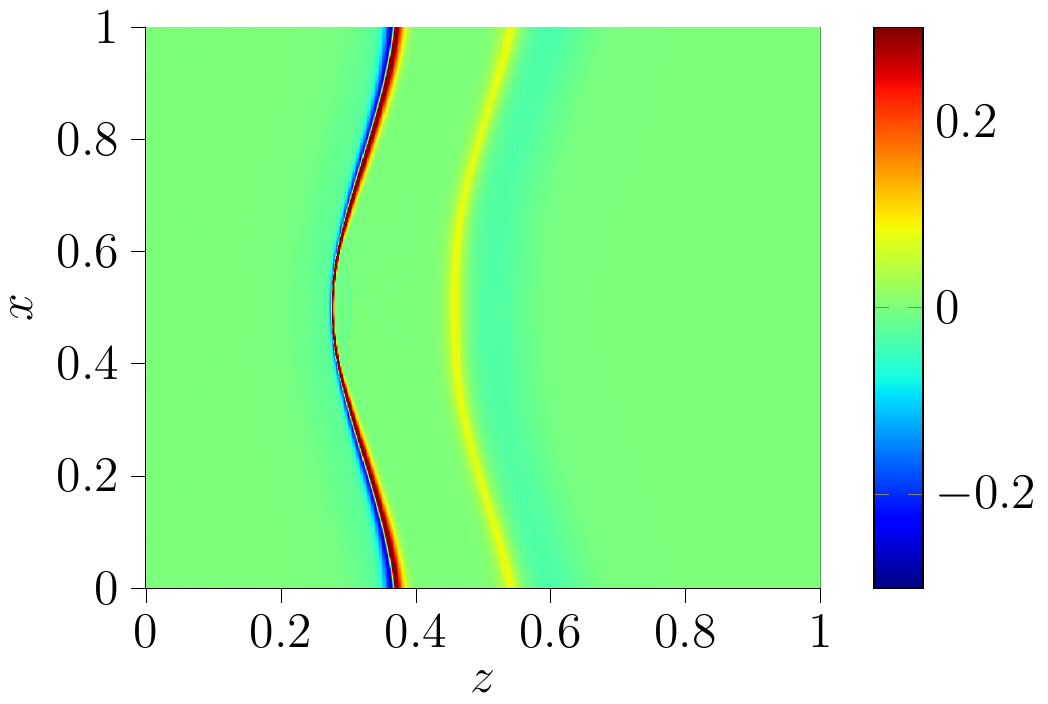}
				\caption{$\rho-r$ at $t=0.25$}
				\label{fig:Test1Eta1MidError}
			\end{subfigure}\hfill
			\begin{subfigure}{.32\linewidth}
				\includegraphics[width=\textwidth]{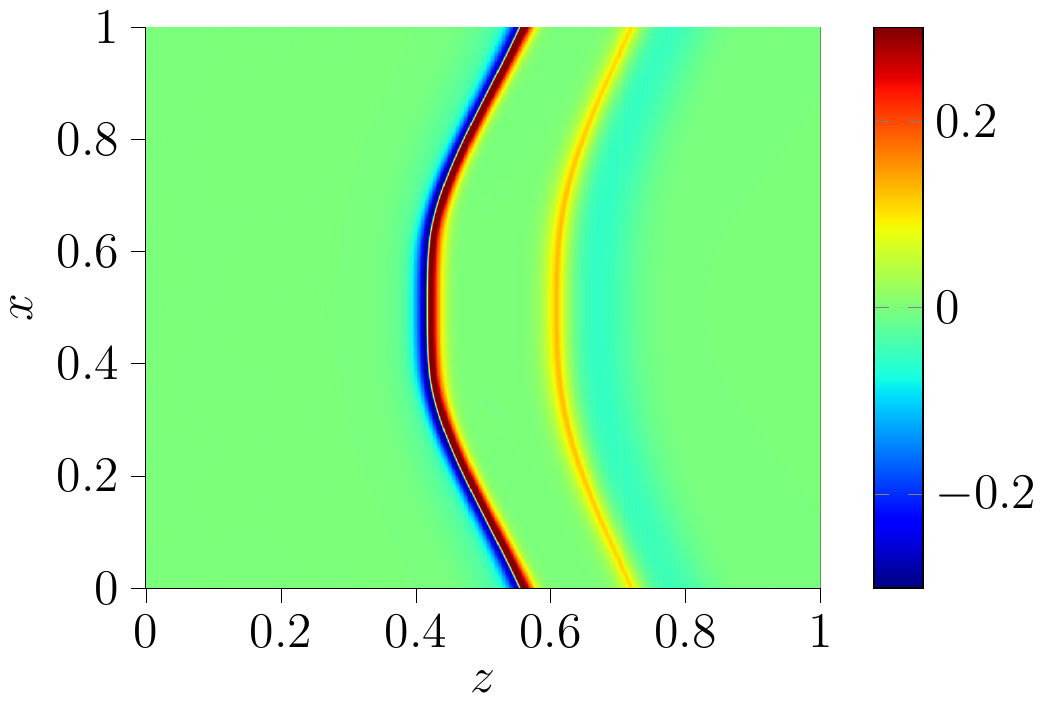}
				\caption{$\rho-r$ at $t=0.5$}
				\label{fig:Test1Eta1LateError}
			\end{subfigure}\\
			\begin{subfigure}{.3\linewidth}
				\includegraphics[width=\textwidth]{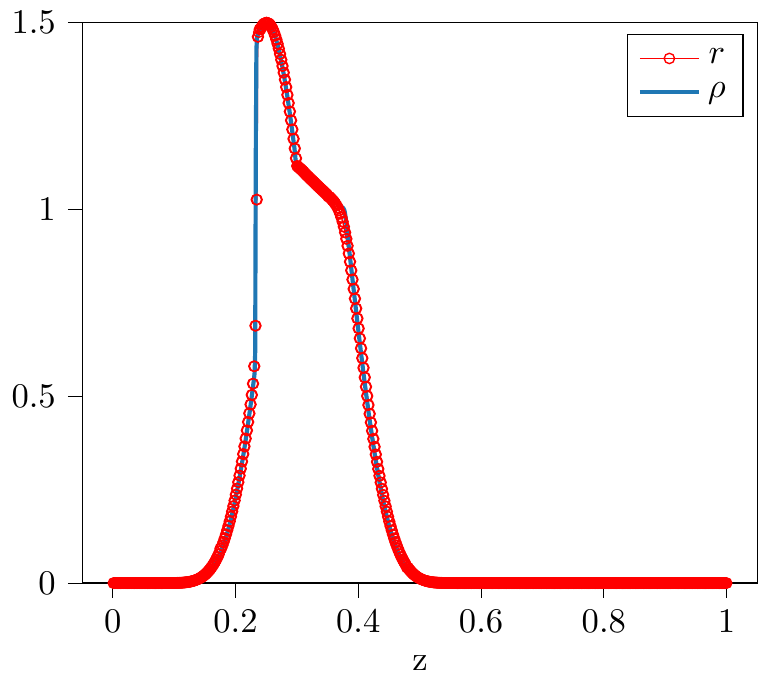}
				\caption{$\rho$,$r$ at $(x,t)=(0.3,0.1)$}
				\label{fig:Test1Eta1EarlyLine}
			\end{subfigure}\hfill
			\begin{subfigure}{.3\linewidth}
				\includegraphics[width=\textwidth]{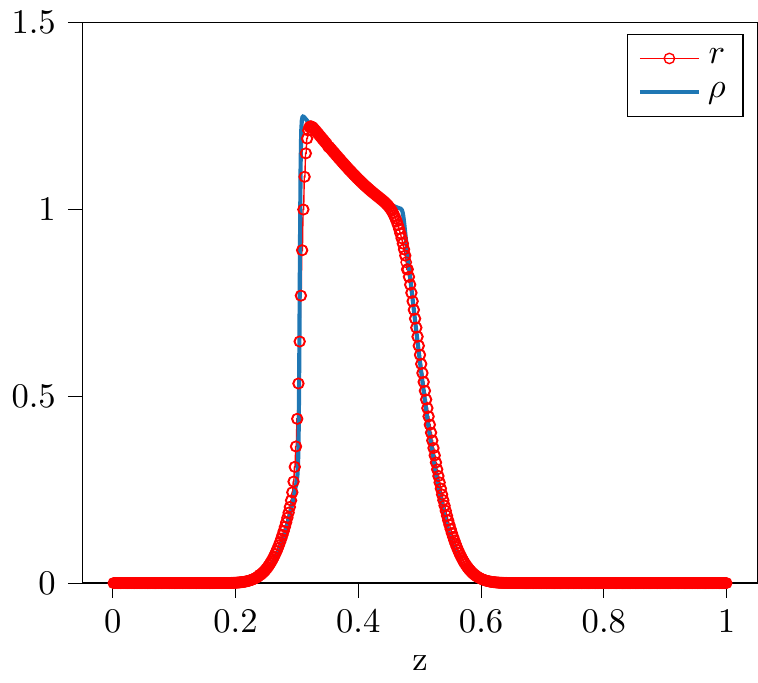}
				\caption{$\rho$,$r$ at $(x,t)=(0.3,0.25)$}
				\label{fig:Test1Eta1MidLine}
			\end{subfigure}\hfill
			\begin{subfigure}{.3\linewidth}
				\includegraphics[width=\textwidth]{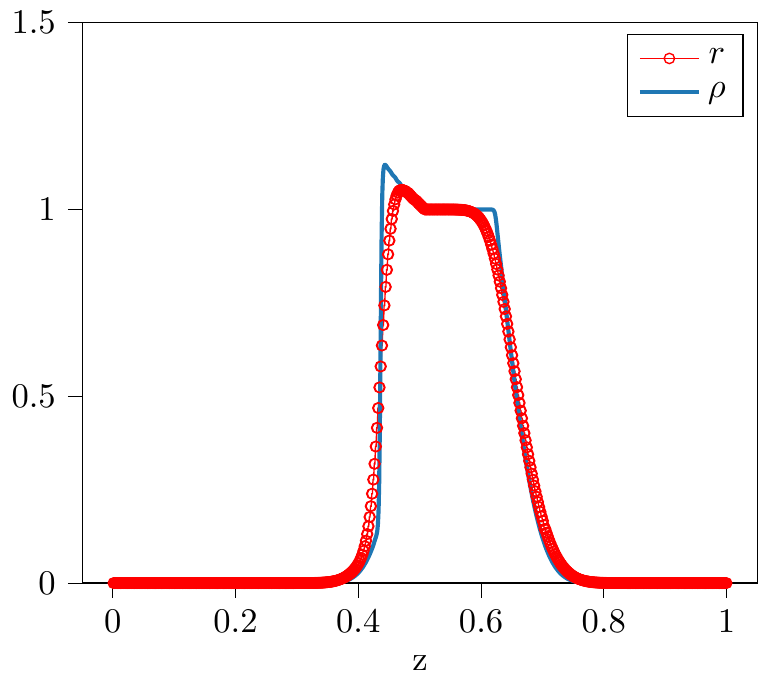}
				\caption{$\rho$,$r$ at $(x,t)=(0.3,0.5)$}
				\label{fig:Test1Eta1LateLine}
			\end{subfigure}
			\caption{Discrete $r$ and continuum $\rho$ solutions when $\eta=1$ case.  From left to right, column correspond to solutions
				at $t=0.1$, $t=0.25$, and $t=0.5$  Discrete solution is computed with $(\imax, \kmax) = (500,500)$.
				Continuum solution is computed on a $10^3 \times 10^3$ mesh.  
			}
			\label{fig:ODEcomparetime1}
		\end{figure}
		\begin{figure}[h!]
			\centering
			\begin{subfigure}{.32\linewidth}
				\includegraphics[width=\textwidth]{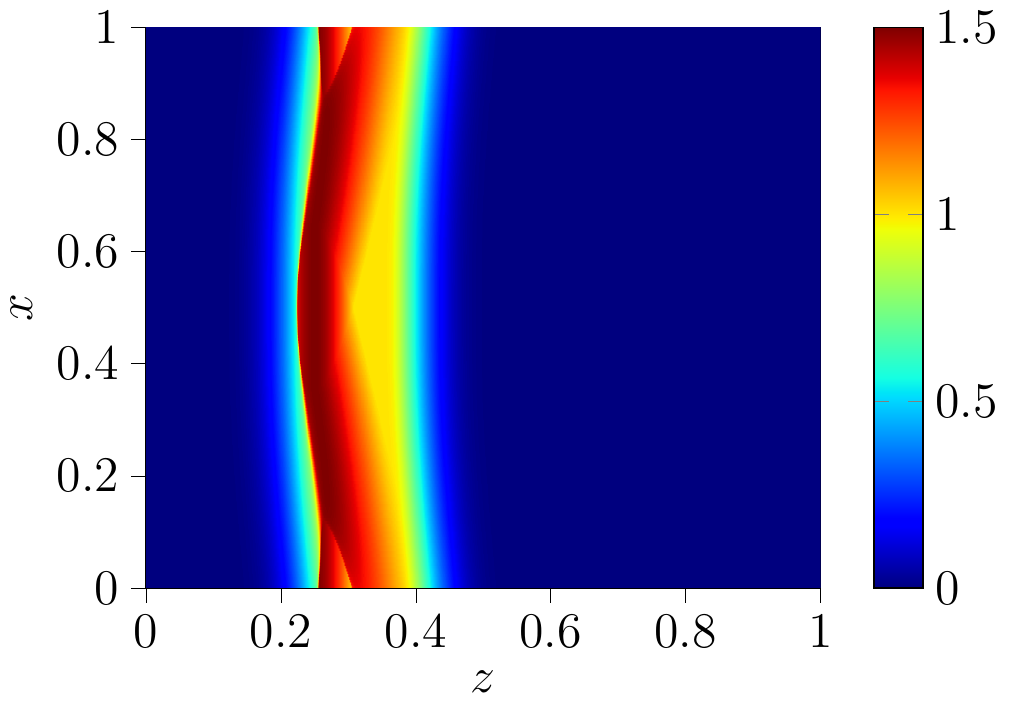}
				\caption{$r$ at $t=0.1$}
				\label{fig:Test1Eta5EarlyODE}
			\end{subfigure}\hfill
			\begin{subfigure}{.32\linewidth}
				\includegraphics[width=\textwidth]{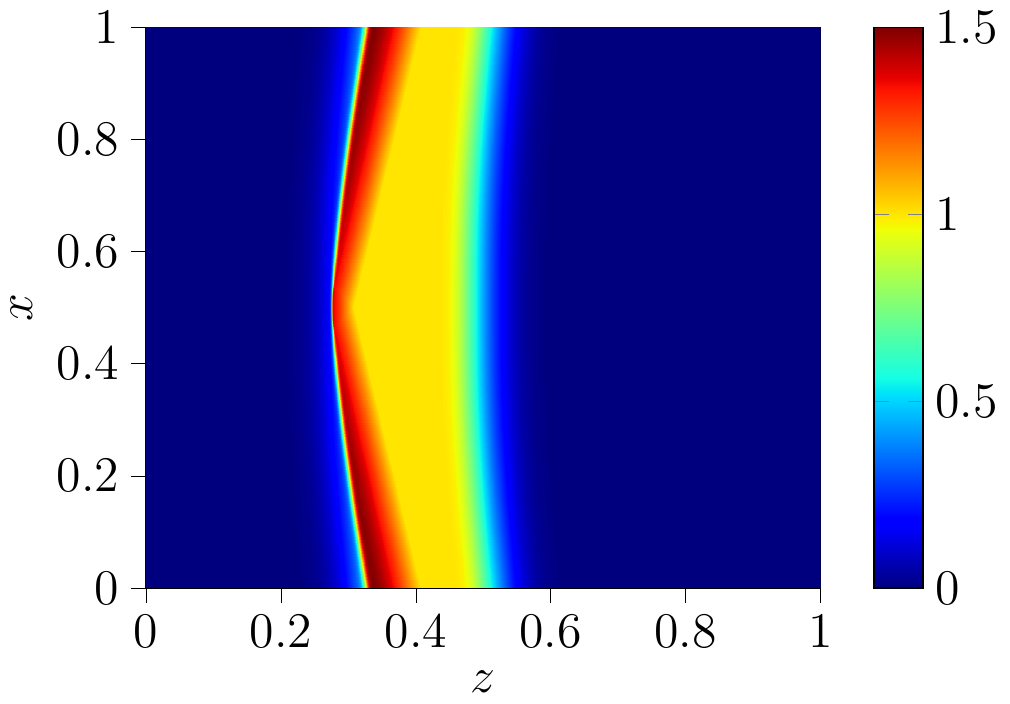}
				\caption{$r$ at $t=0.25$ }	
				\label{fig:Test1Eta5MidODE}
			\end{subfigure}\hfill
			\begin{subfigure}{.32\linewidth}
				\includegraphics[width=\textwidth]{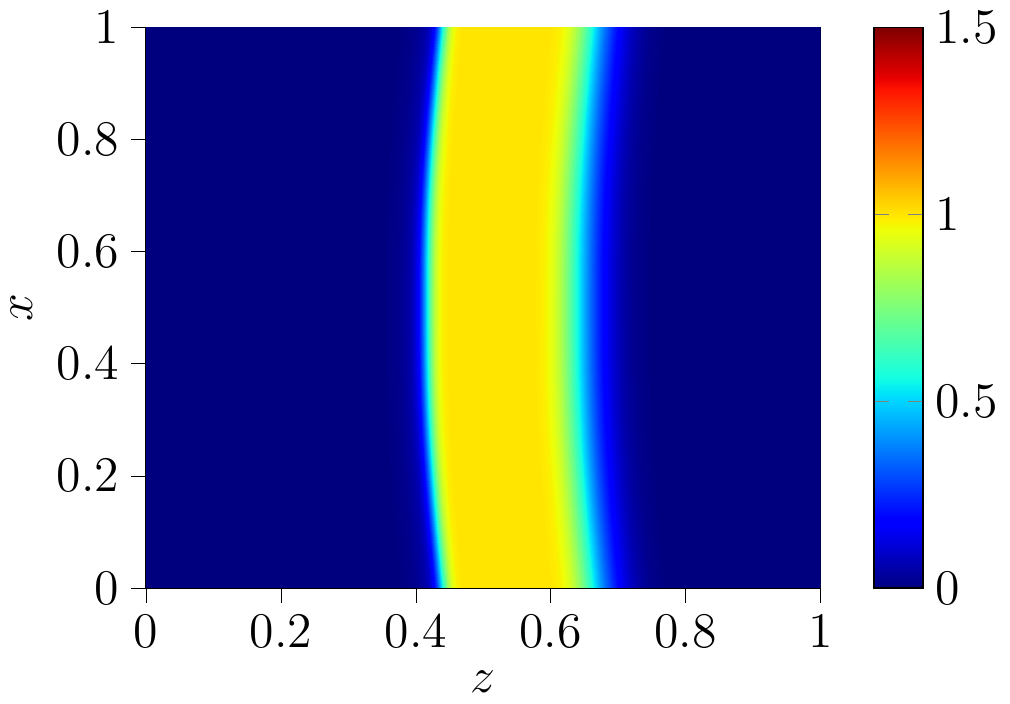}
				\caption{$r$ at $t=0.5$}	
				\label{fig:Test1Eta5LateODE}
			\end{subfigure}\hfill
			\begin{subfigure}{.32\linewidth}
				\includegraphics[width=\textwidth]{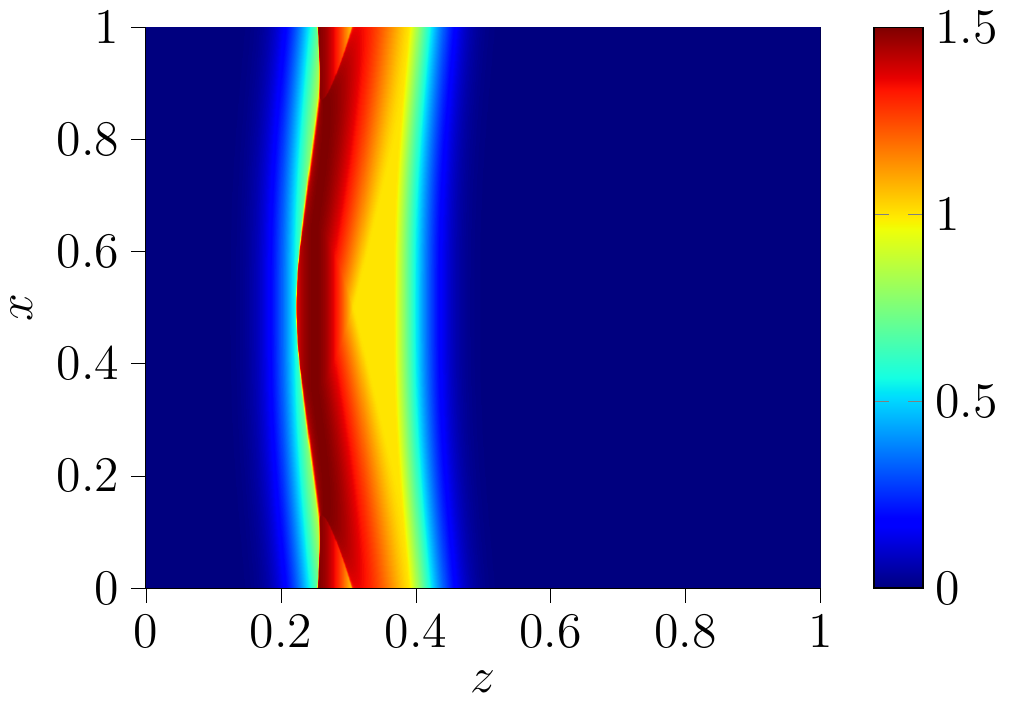}
				\caption{$\rho$ at $t=0.1$}
				\label{fig:Test1Eta5EarlyPDE}
			\end{subfigure}\hfill
			\begin{subfigure}{.32\linewidth}
				\includegraphics[width=\textwidth]{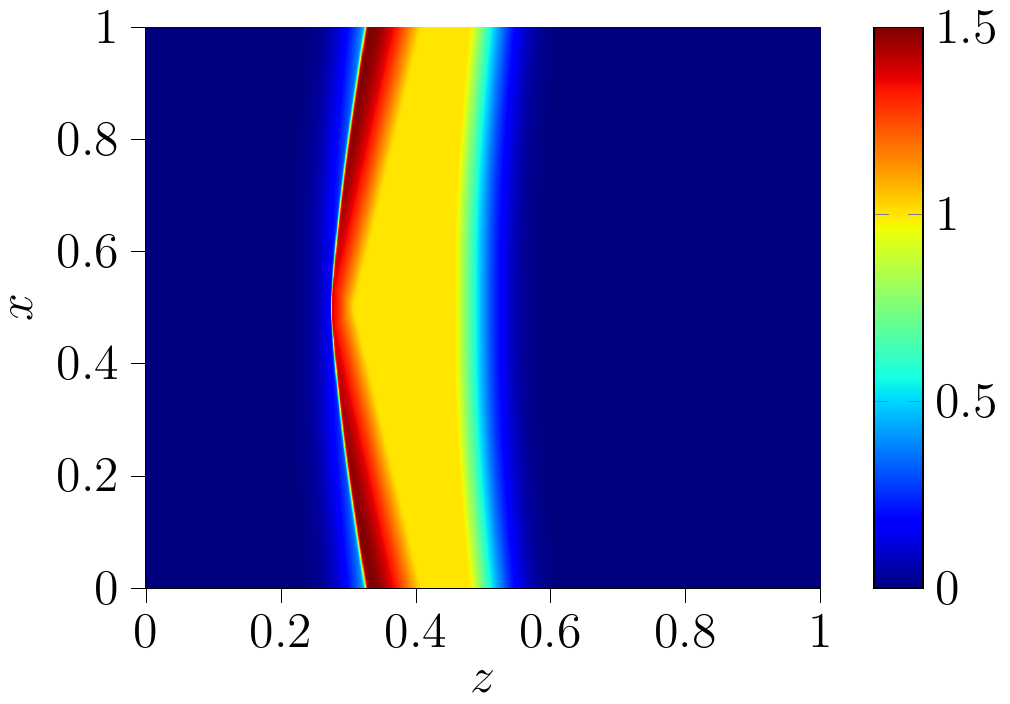}
				\caption{$\rho$ at $t=0.25$ }
				\label{fig:Test1Eta5MidPDE}
			\end{subfigure}\hfill
			\begin{subfigure}{.32\linewidth}
				\includegraphics[width=\textwidth]{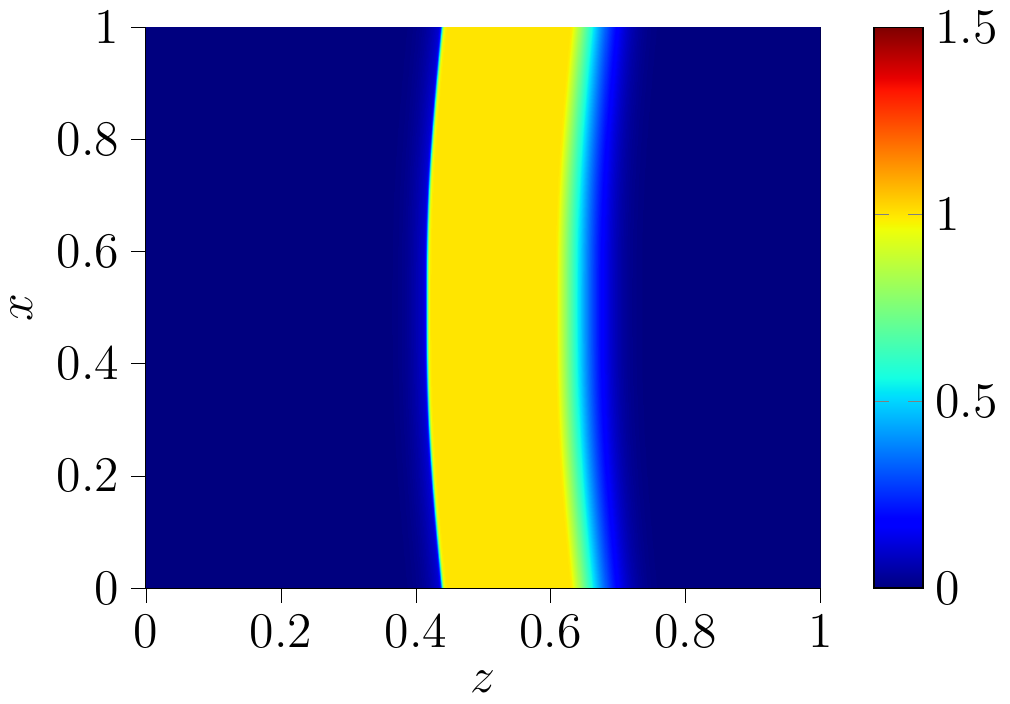}
				\caption{$\rho$ at $t=0.5$ }
				\label{fig:Test1Eta5LatePDE}
			\end{subfigure}\hfill
			\begin{subfigure}{.32\linewidth}
				\includegraphics[width=\textwidth]{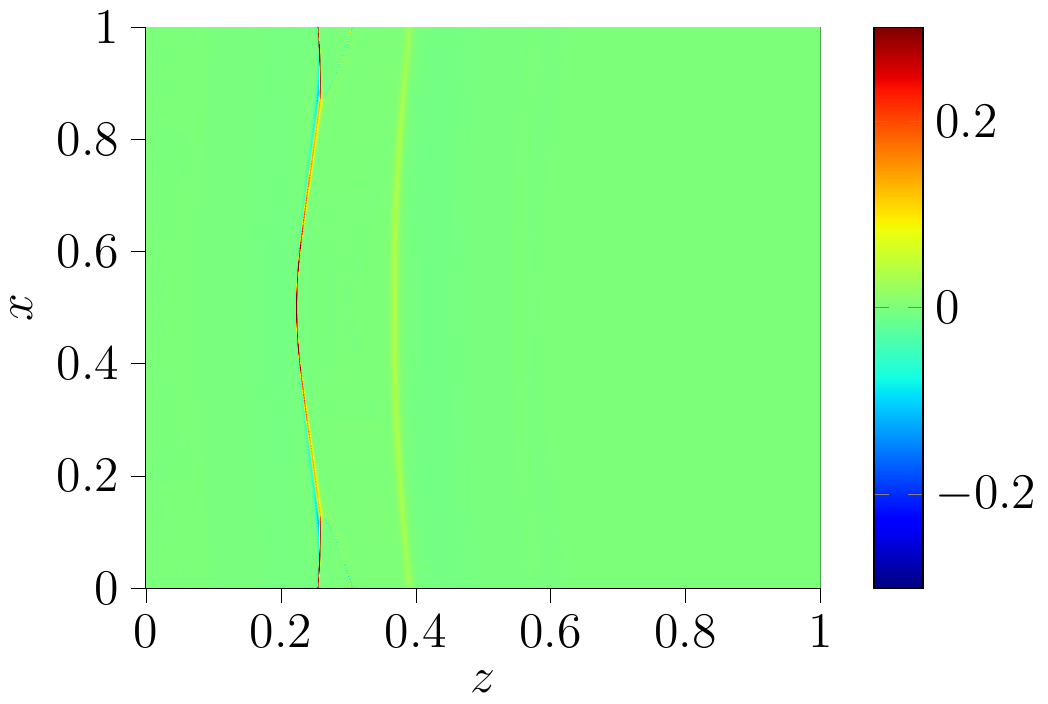}
				\caption{$\rho-r$ at $t=0.1$}
				\label{fig:Test1Eta5EarlyError}
			\end{subfigure}\hfill
			\begin{subfigure}{.32\linewidth}
				\includegraphics[width=\textwidth]{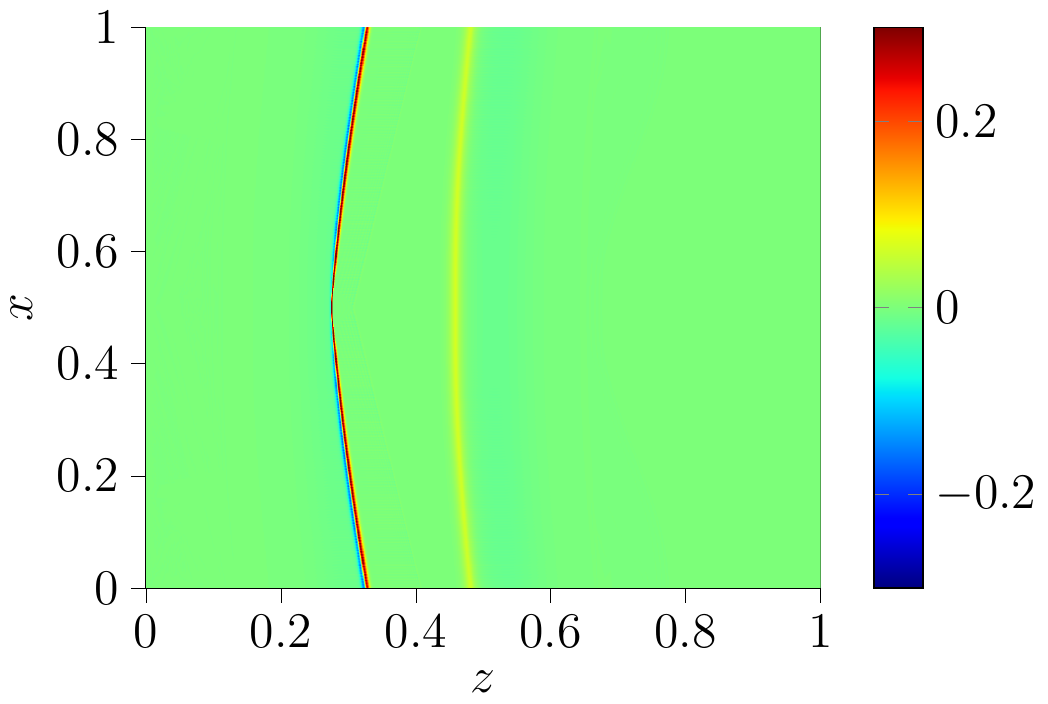}
				\caption{$\rho-r$ at $t=0.25$}
				\label{fig:Test1Eta5MidError}
			\end{subfigure}\hfill
			\begin{subfigure}{.32\linewidth}
				\includegraphics[width=\textwidth]{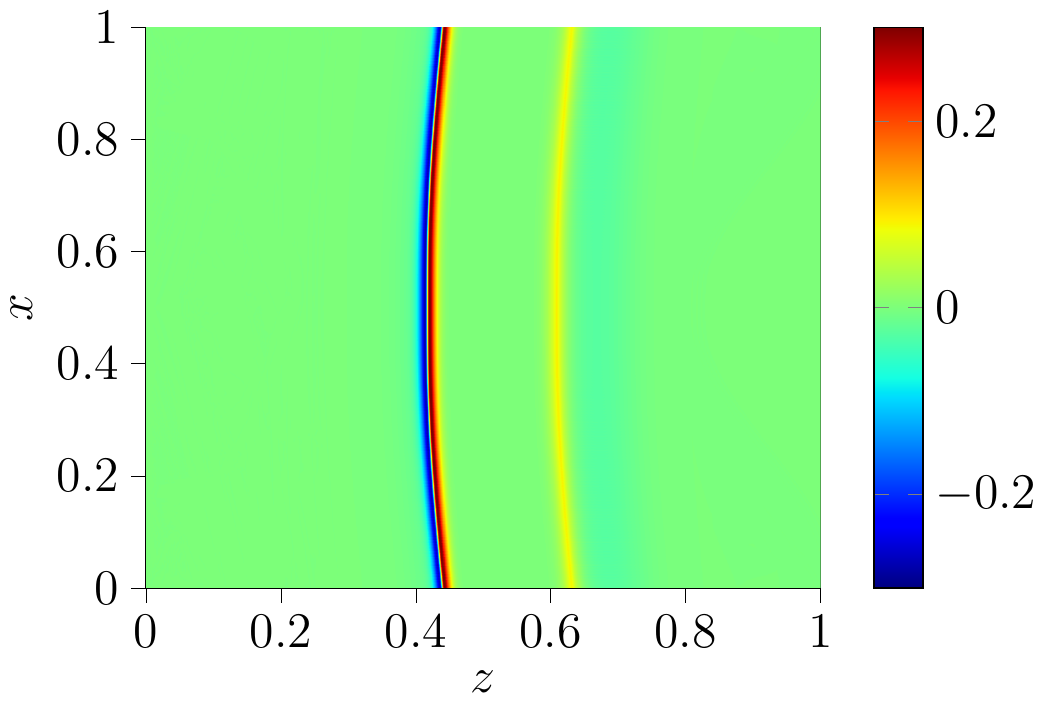}
				\caption{$(\rho-r)$ at $t=0.50$}
				\label{fig:Test1Eta5LateError}
			\end{subfigure}\hfill
			\begin{subfigure}{.3\linewidth}
				\includegraphics[width=\textwidth]{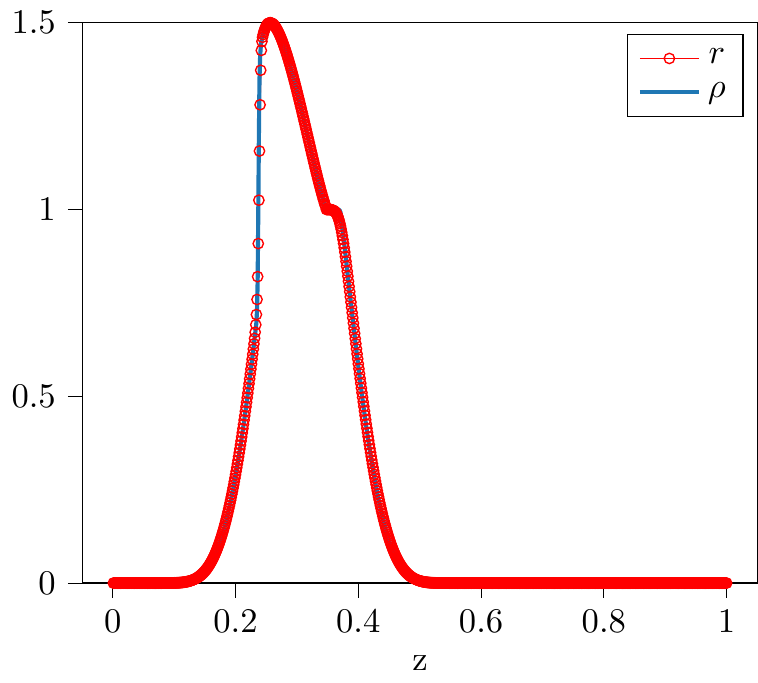}
				\caption{$\rho$,$r$ at $(x,t)=(0.3,0.1)$}
				\label{fig:Test1Eta5EarlyLine}
			\end{subfigure}\hfill
			\begin{subfigure}{.3\linewidth}
				\includegraphics[width=\textwidth]{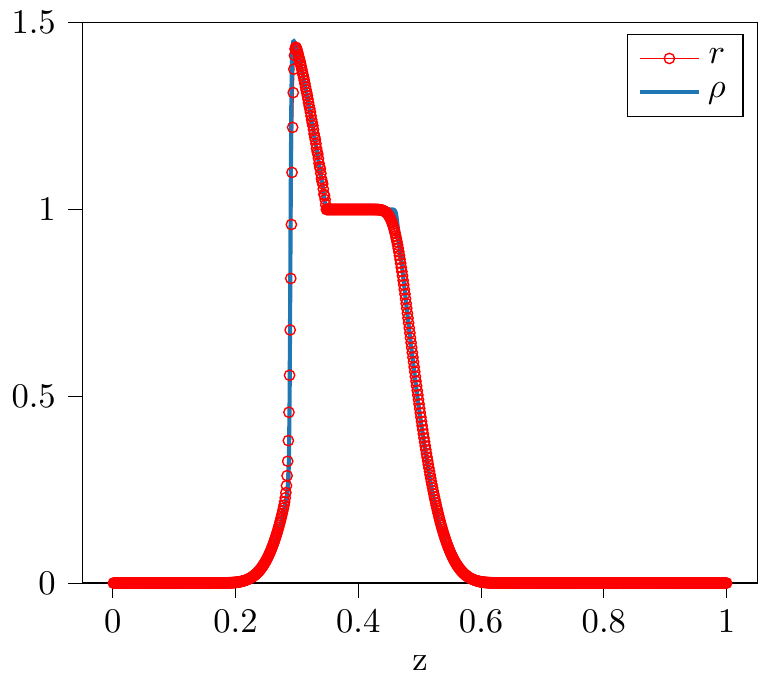}
				\caption{$\rho$,$r$ at $(x,t)=(0.3,0.25)$}
				\label{fig:Test1Eta5MidLine}
			\end{subfigure}\hfill
			\begin{subfigure}{.3\linewidth}
				\includegraphics[width=\textwidth]{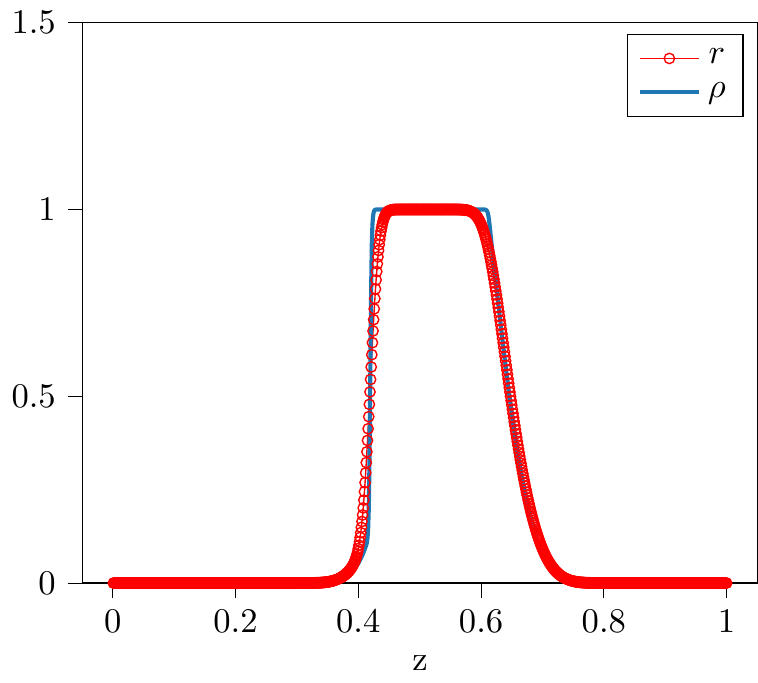}
				\caption{$\rho$,$r$ at $(x,t)=(0.3,0.5)$}
				\label{fig:Test1Eta5LateLine}
			\end{subfigure}
			\caption{Discrete $r$ and continuum $\rho$ solutions when $\eta=5$ case.  From left to right, column correspond to solutions
				at $t=0.1$, $t=0.25$, and $t=0.5$  Discrete solution is computed with $(\imax, \kmax) = (200,1000)$.
				Continuum solution is computed on a $10^3 \times 10^3$ mesh.  
			}
			\label{fig:ODEcomparetime5}
		\end{figure}
		
		\begin{figure}[h!]
			\begin{subfigure}{.3\linewidth}
				\includegraphics[width=\textwidth]{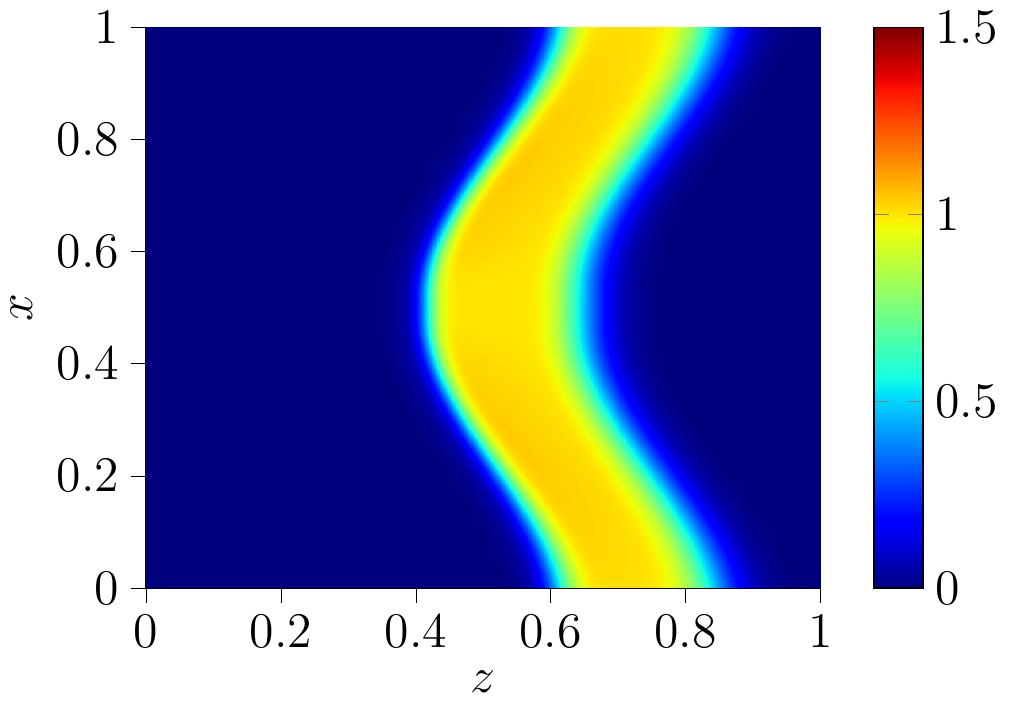}
				\caption{$r$ at $t=0.5$}
			\end{subfigure}\hfill
			\begin{subfigure}{.3\linewidth}
				\includegraphics[width=\textwidth]{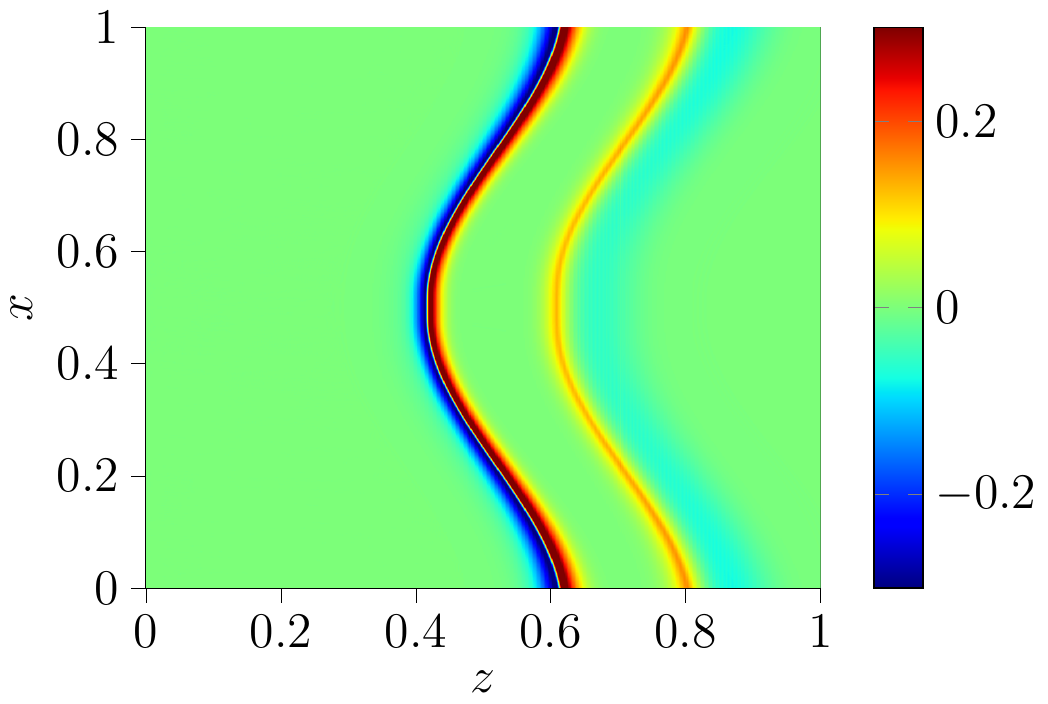}
				\caption{$\rho-r$ at $t=0.5$}
			\end{subfigure}\hfill
			\begin{subfigure}{.3\linewidth}
				\includegraphics[width=\textwidth]{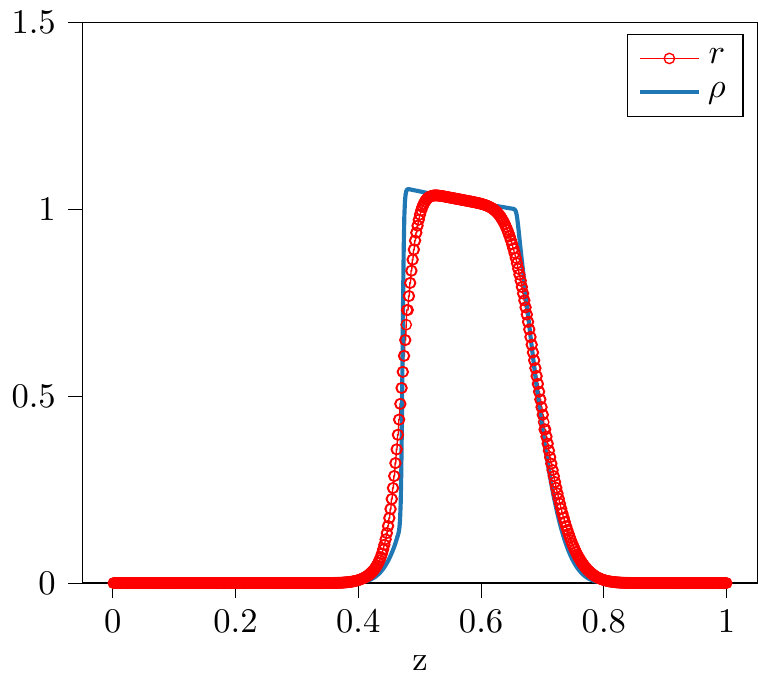}
				\caption{$\rho$,$r$ at $(x,t)=(0.3,0.5)$}
			\end{subfigure}
			\caption{Comparison of the discrete model with $(\imax, \kmax) = (2500,500)$ and the continuum model for $\eta=0.2$ at time $t=0.5$.  As expected, the discrete model shows better agreement with the continuum model than the previous version with only $(\imax, \kmax) = (1000,200)$ processors and stages; cf. \cref{fig:ODEcomparetime02}}
			\label{fig:ODEcompareFine}
		\end{figure}
	
	\end{example}

	For the remaining examples,  the Hamilton-Jacobi simulations are performed on a coarser mesh of $100\times 100$.

	\begin{example}[Variations in $\eta$]
		In this example, we examine the effect of $\eta$ on solutions to the macroscopic model while $\beta = 1.0$ is fixed.
		The initial condition, boundary condition, and processor speed are given by 
		\begin{equation}\label{eq:etaVarTest}
		\rho_0(x,z) = 1.5\chi_{z\leq 0.2}(x,z) \quad \rho_{\rm{bc}}(x,t) =0, \quad  \alpha(x) = 1-0.4(\sin(\pi x))^6,
		\end{equation}
		respectively.
		It is expected that the slower processor speed around $x=0.5$ will slow down neighboring processors due to neighbor-based throttling, encoded in the definition of $w_2$ in \eqref{eq:w2}.  Moreover, the effect should become more global in $x$ as $\eta$ increases,  since larger values of $\eta$ correspond to a larger number of stages per processor.  Indeed as the stages increase, interactions between neighbors begin to have a cumulative global effect.  This trend can be observed by comparing results across the first three rows of \cref{fig:etaVar} and in the line-outs in the final row.

		\begin{figure}[h!]\centering
			\begin{subfigure}{0.32\linewidth}
				\includegraphics[width=\textwidth]{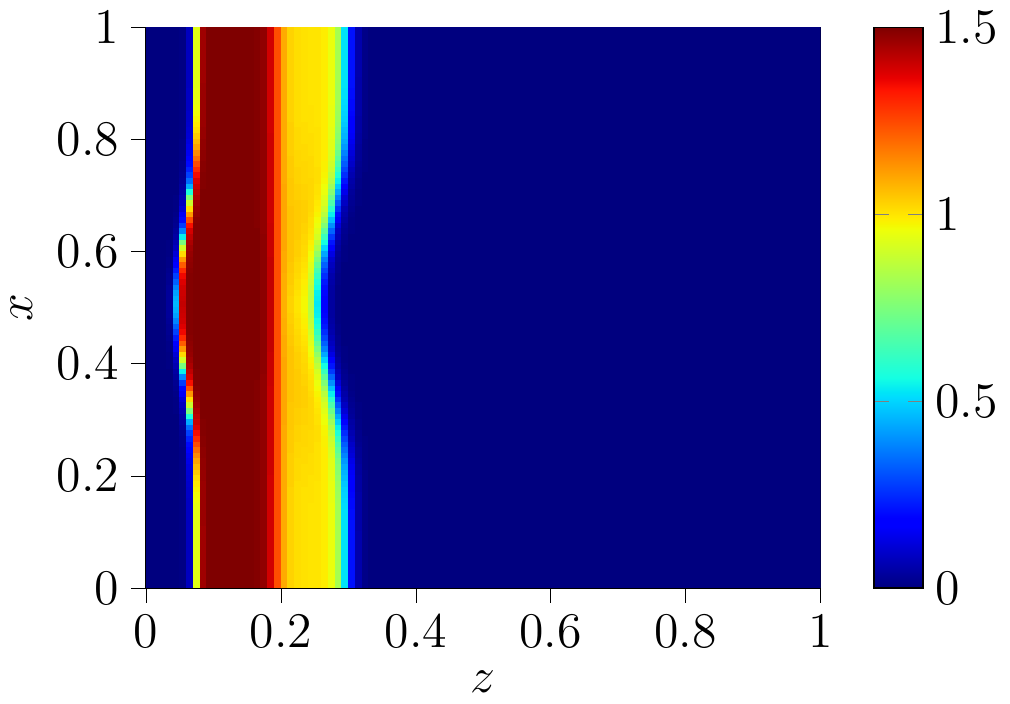}
				\caption{$\rho(x,z,0.1)$ when $\eta=0.2$}
				\label{fig:Test2Eta0_2Time0_1}
			\end{subfigure}
			\begin{subfigure}{0.32\linewidth}
				\includegraphics[width=\textwidth]{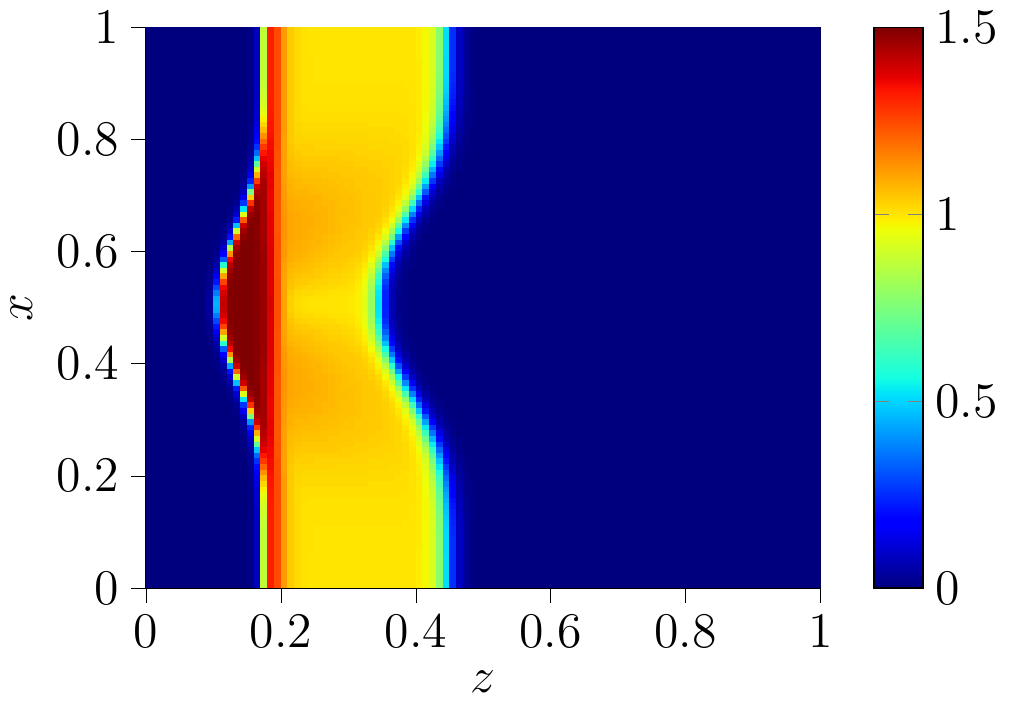}
				\caption{$\rho(x,z,0.25)$ when $\eta=0.2$}
				\label{fig:Test2Eta0_2Time0_25}
			\end{subfigure}
			\begin{subfigure}{0.32\linewidth}
				\includegraphics[width=\textwidth]{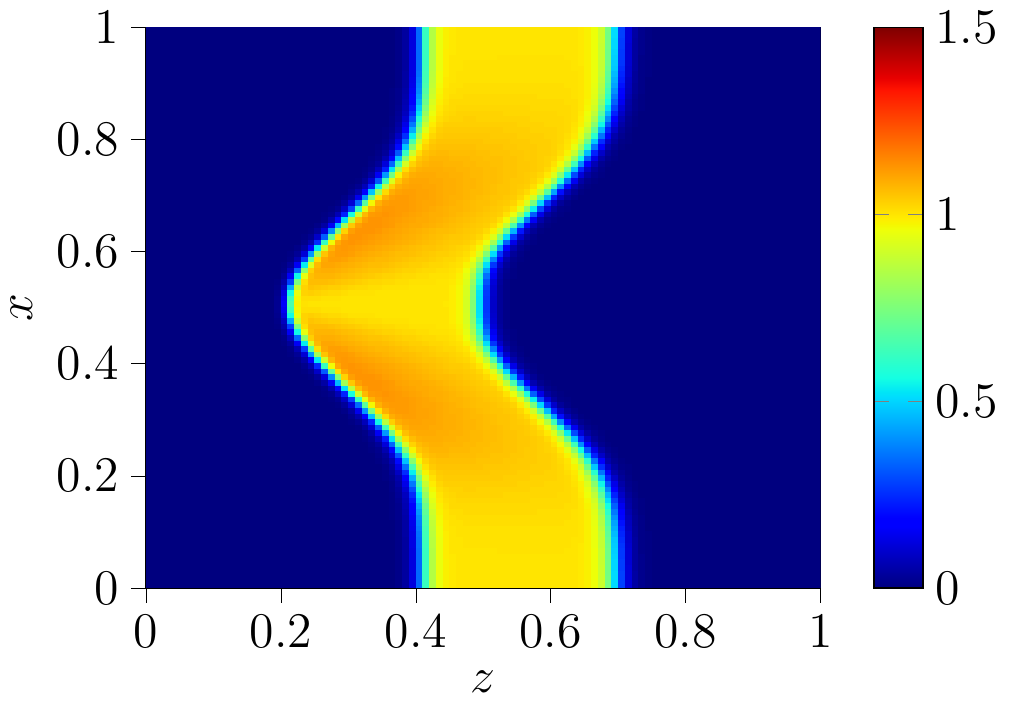}
				\caption{$\rho(x,z,0.5)$ when $\eta=0.2$}
				\label{fig:Test2Eta0_2Time0_5}
			\end{subfigure}
			\\
			\begin{subfigure}{0.32\linewidth}
				\includegraphics[width=\textwidth]{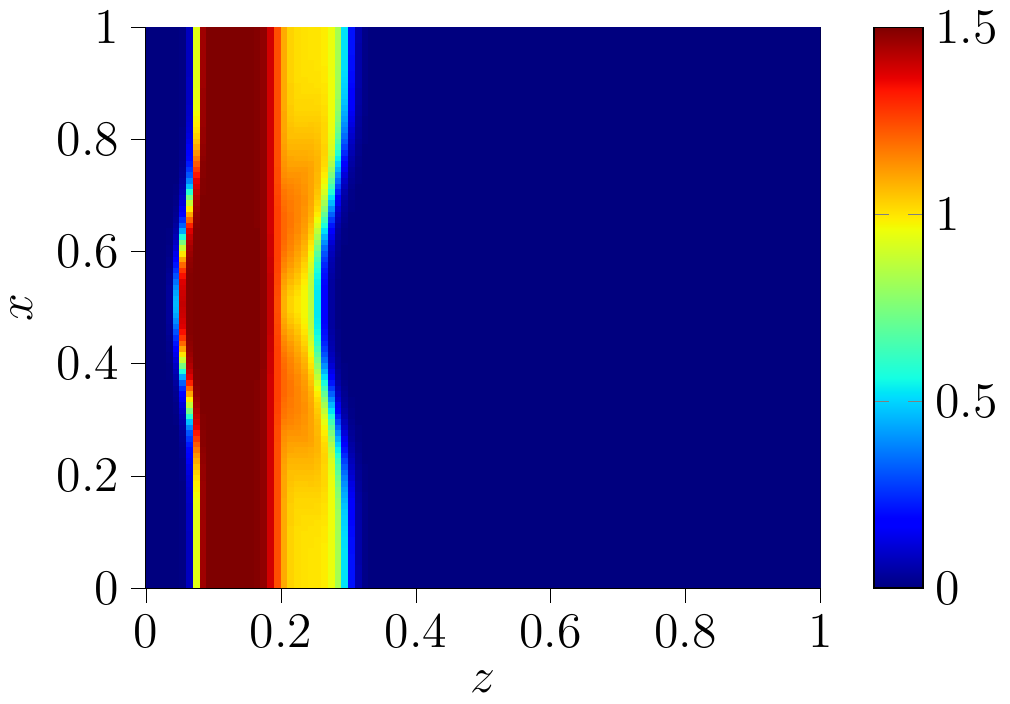}
				\caption{$\rho(x,z,0.1)$ when $\eta=1.0$}
				\label{fig:Test2Eta1_0Time0_1}
			\end{subfigure}
			\begin{subfigure}{0.32\linewidth}
				\includegraphics[width=\textwidth]{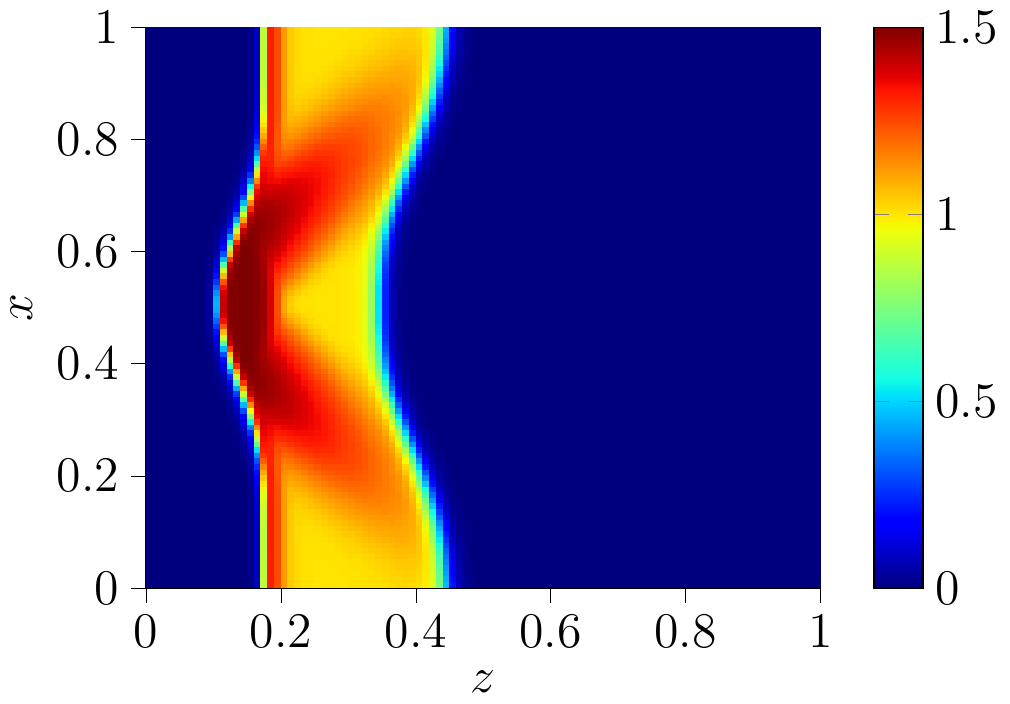}
				\caption{$\rho(x,z,0.25)$ when $\eta=1.0$}
				\label{fig:Test2Eta1_0Time0_25}
			\end{subfigure}
			\begin{subfigure}{0.32\linewidth}
				\includegraphics[width=\textwidth]{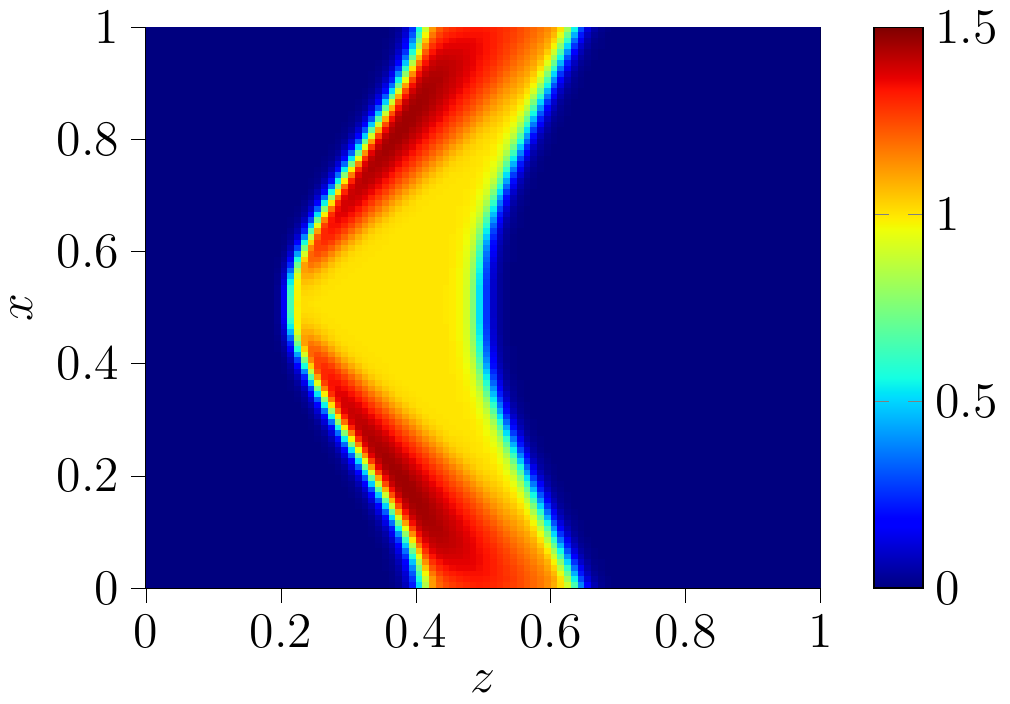}
				\caption{$\rho(x,z,0.5)$ when $\eta=1.0$}
				\label{fig:Test2Eta1_0Time0_5}
			\end{subfigure}
			\\
			\begin{subfigure}{0.32\linewidth}
				\includegraphics[width=\textwidth]{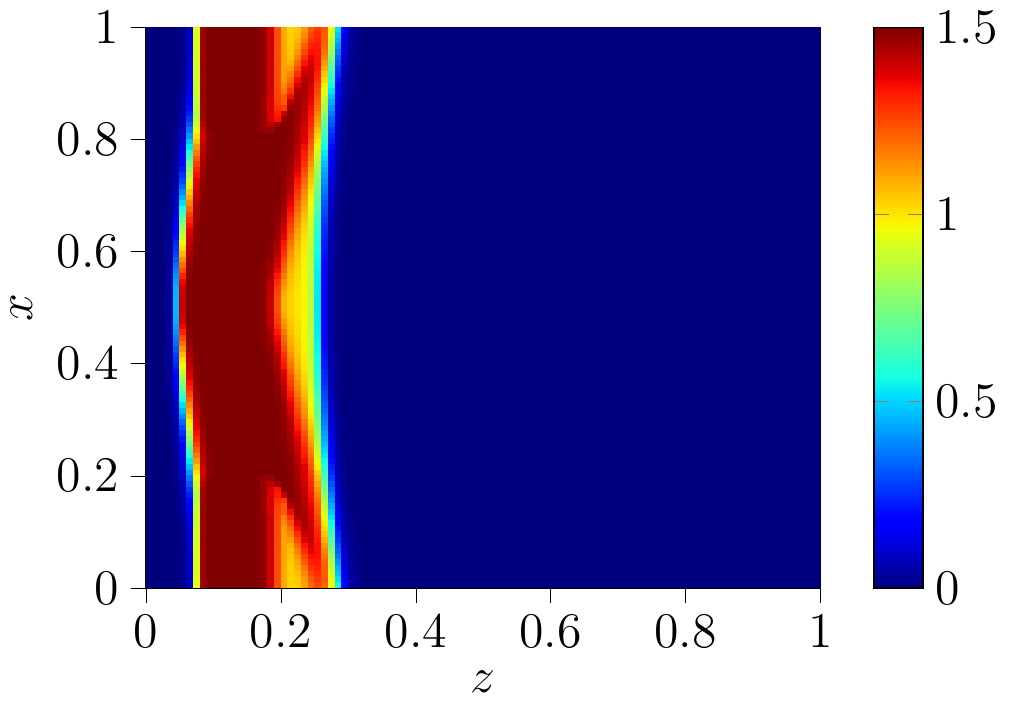}
				\caption{$\rho(x,z,0.1)$ when $\eta=5.0$}
				\label{fig:Test2Eta5_0Time0_1}
			\end{subfigure}
			\begin{subfigure}{0.32\linewidth}
				\includegraphics[width=\textwidth]{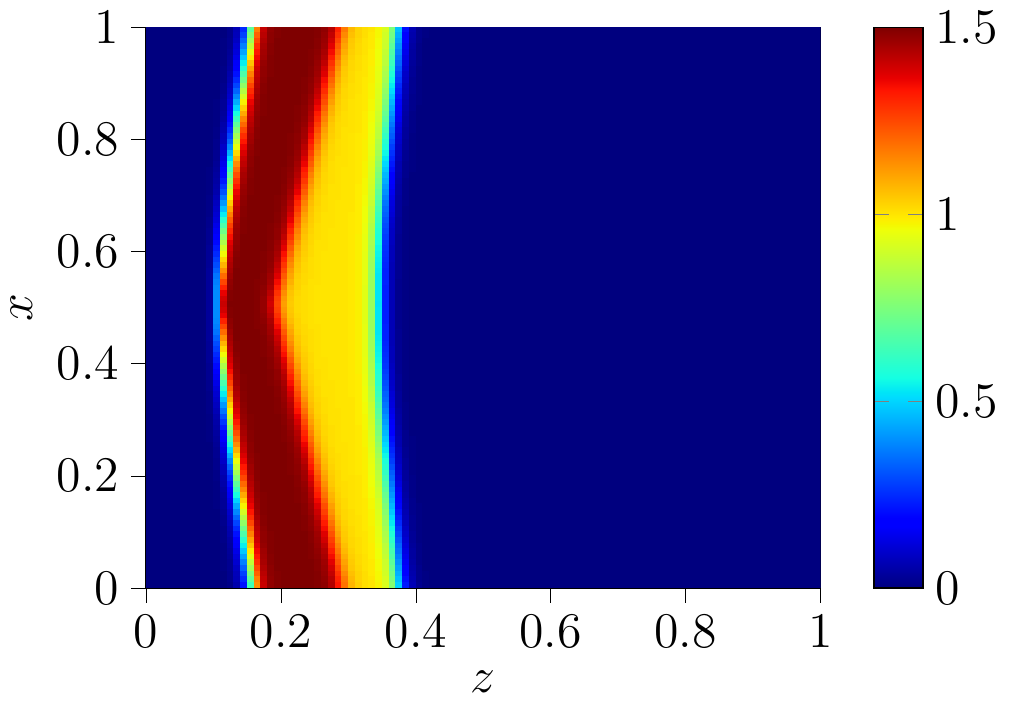}
				\caption{$\rho(x,z,0.25)$ when $\eta=5.0$}
				\label{fig:Test2Eta5_0Time0_25}
			\end{subfigure}
			\begin{subfigure}{0.32\linewidth}
				\includegraphics[width=\textwidth]{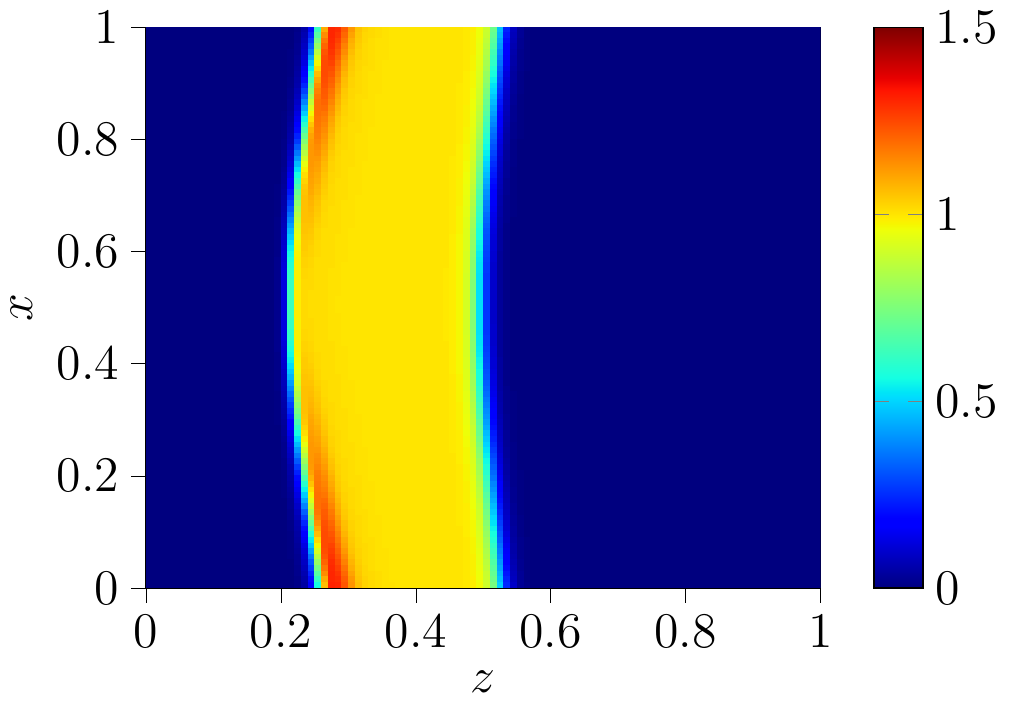}
				\caption{$\rho(x,z,0.5)$ when $\eta=5.0$}
				\label{fig:Test2Eta5_0Time0_5}
			\end{subfigure}\\
			\begin{subfigure}{0.32\linewidth}
				\includegraphics[width=\textwidth]{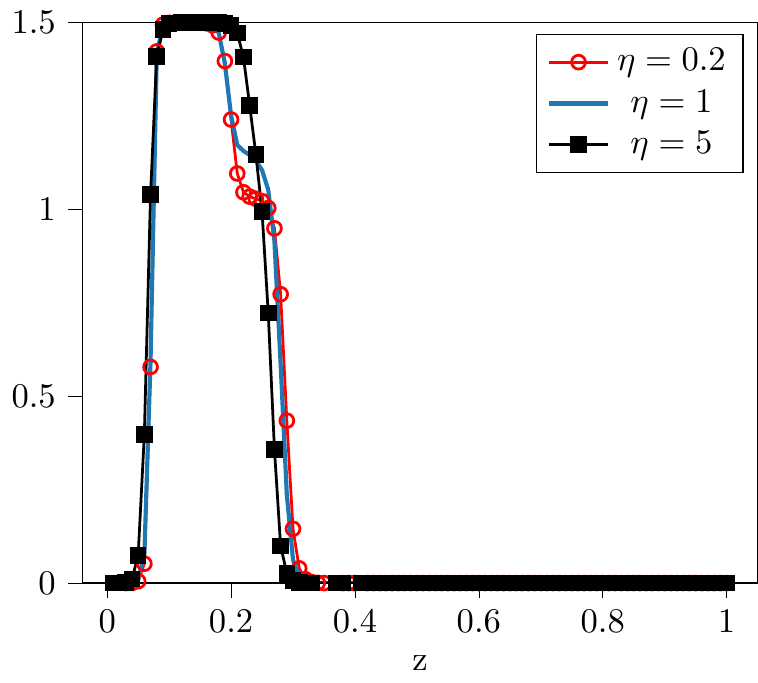}
				\caption{ $\rho(0.3,z,0.1)$}
				\label{fig:Test2LineTime0_1}
			\end{subfigure}
			\begin{subfigure}{0.32\linewidth}
				\includegraphics[width=\textwidth]{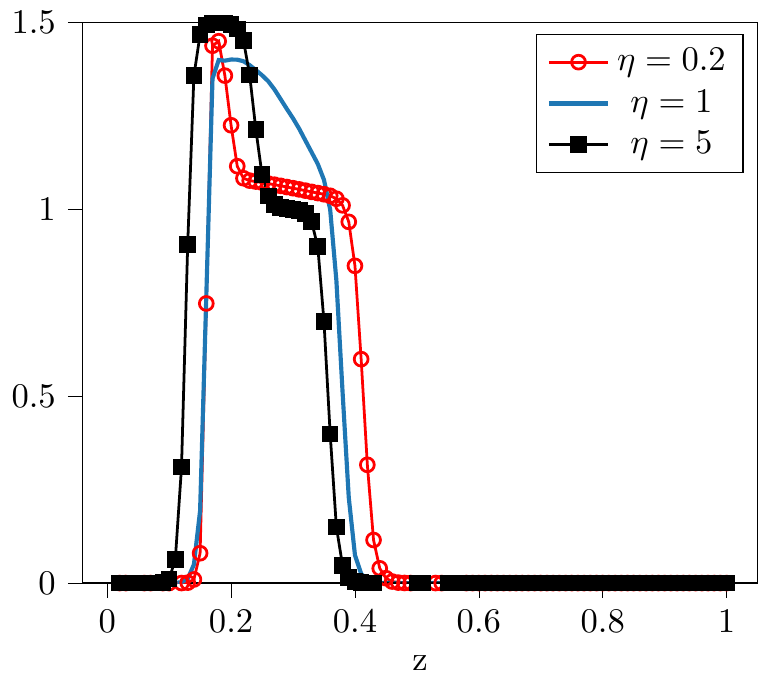}
				\caption{$\rho(0.3,z,0.25)$}
				\label{fig:Test2LineTime0_25}
			\end{subfigure}
			\begin{subfigure}{0.32\linewidth}
				\includegraphics[width=\textwidth]{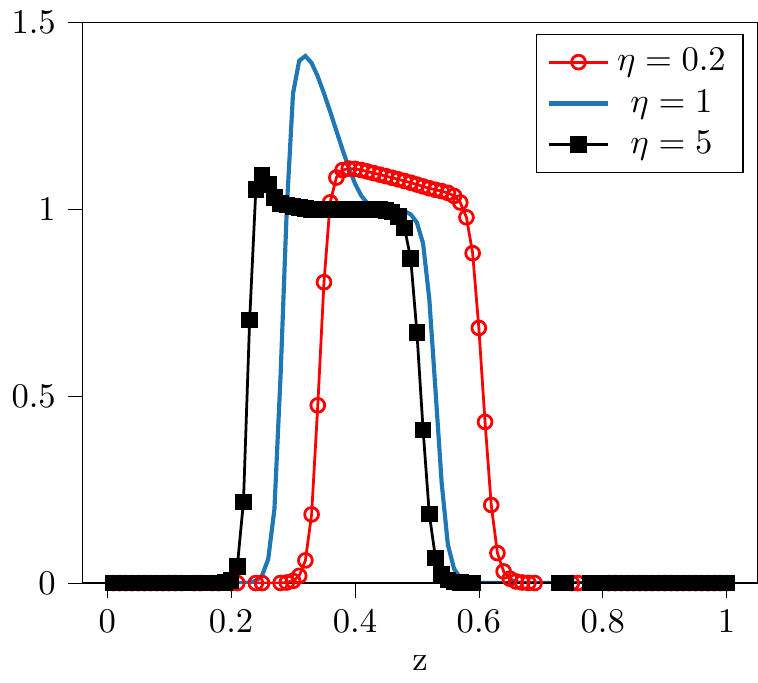}
				\caption{$\rho(0.3,z,0.5)$ }
				\label{fig:Test2LineTime0_5}		
			\end{subfigure}
			\caption{The effects on $\rho$ due to variations in $\eta$.  As $\eta$ increases the throttling  effect of a local slowdown spreads more quickly, and data is not processed as quickly.}
			\label{fig:etaVar}
		\end{figure}
	\end{example}

	\begin{example}[Variations in $\beta$]\label{ex:var_beta}
		In this example, we examine the effect of $\beta$ on solutions to the macroscopic model, while holding $\eta = 1.0$ fixed.   
		The initial condition, boundary condition, and processor speed are again given by \eqref{eq:etaVarTest}.
		
		Based on the definition of the function $w_2$ in \eqref{eq:w2}, the expectation is that smaller values of $\beta$ will lead to reduced throttling effects. Such behavior is confirmed by the numerical results in \cref{fig:betaVar}.
		\begin{figure}[h!]
			\centering
			\begin{subfigure}{0.32\linewidth}
				\includegraphics[width=\textwidth]{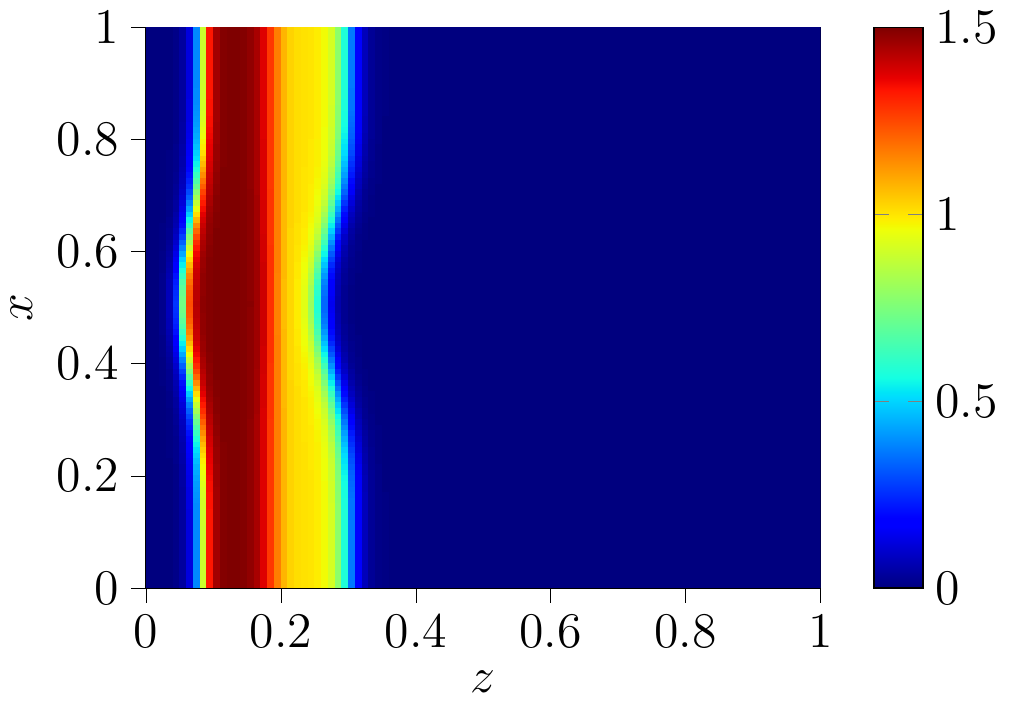}
				\caption{$\rho(x,z,0.1)$  with $\beta=0.1$}
			\end{subfigure}
			\begin{subfigure}{0.32\linewidth}
				\includegraphics[width=\textwidth]{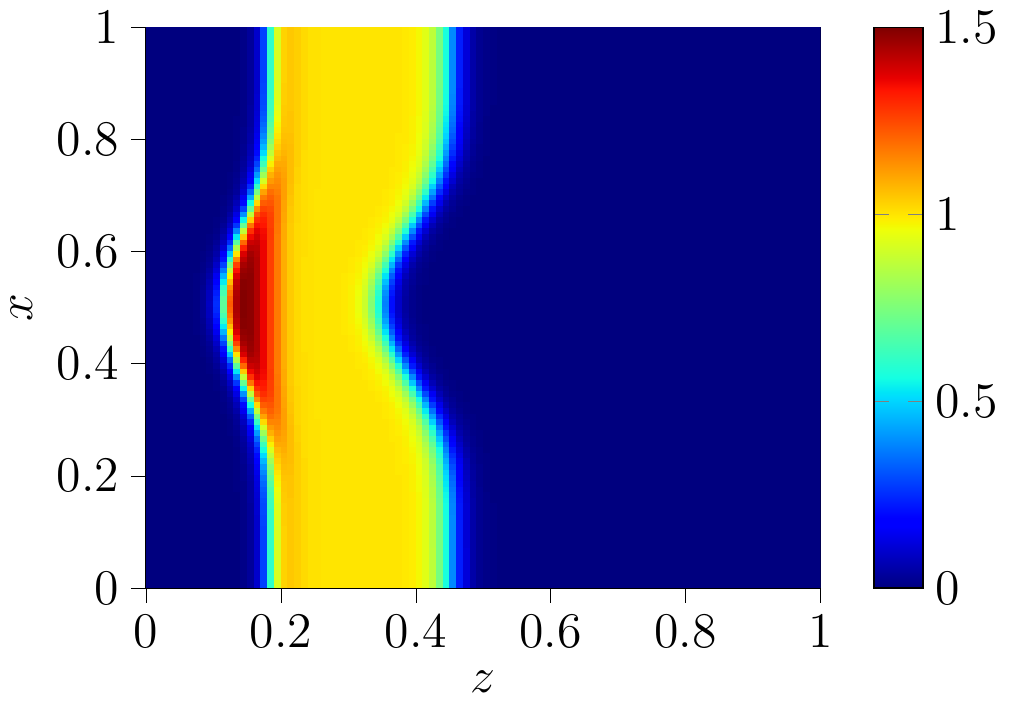}
				\caption{$\rho(x,z,0.25)$  with $\beta=0.1$}
			\end{subfigure}
			\begin{subfigure}{0.32\linewidth}
				\includegraphics[width=\textwidth]{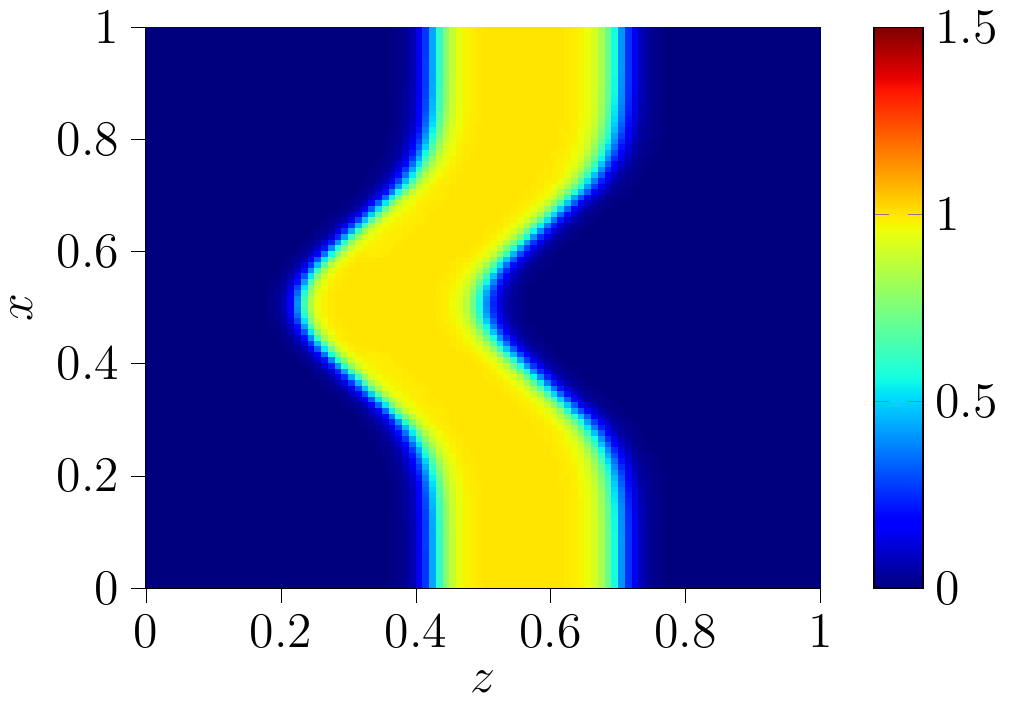}
				\caption{$\rho(x,z,0.5)$  with $\beta=0.1$}
			\end{subfigure}
			\\
			\begin{subfigure}{0.32\linewidth}
				\includegraphics[width=\textwidth]{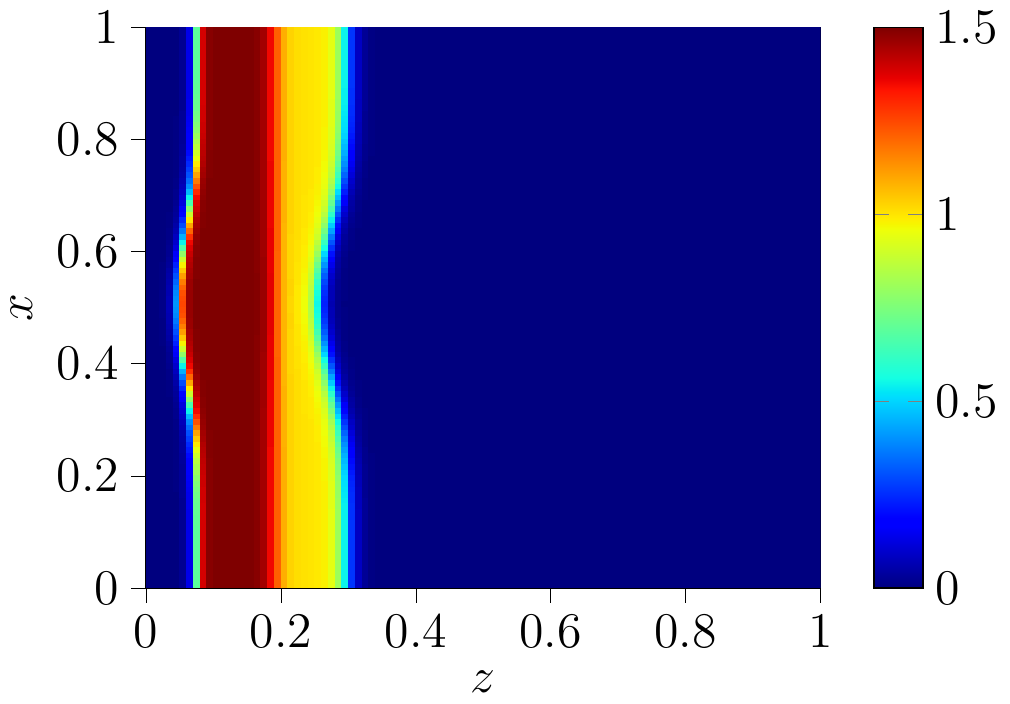}
				\caption{$\rho(x,z,0.1)$  with $\beta=0.5$}
			\end{subfigure}
			\begin{subfigure}{0.32\linewidth}
				\includegraphics[width=\textwidth]{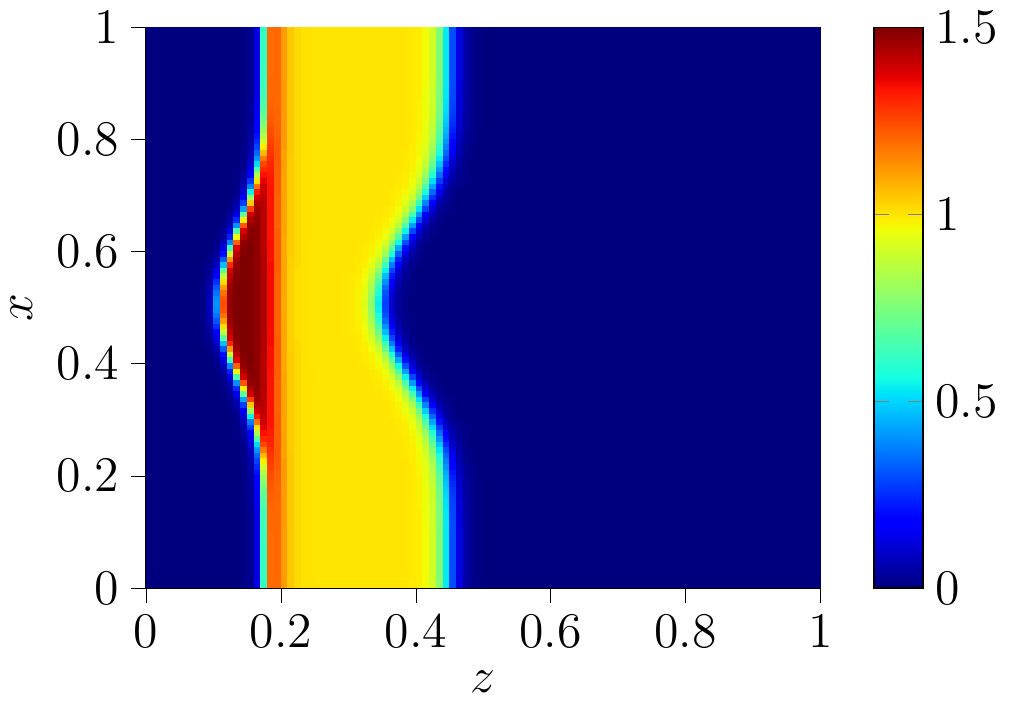}
				\caption{$\rho(x,z,0.25)$  with $\beta=0.5$}
			\end{subfigure}
			\begin{subfigure}{0.32\linewidth}
				\includegraphics[width=\textwidth]{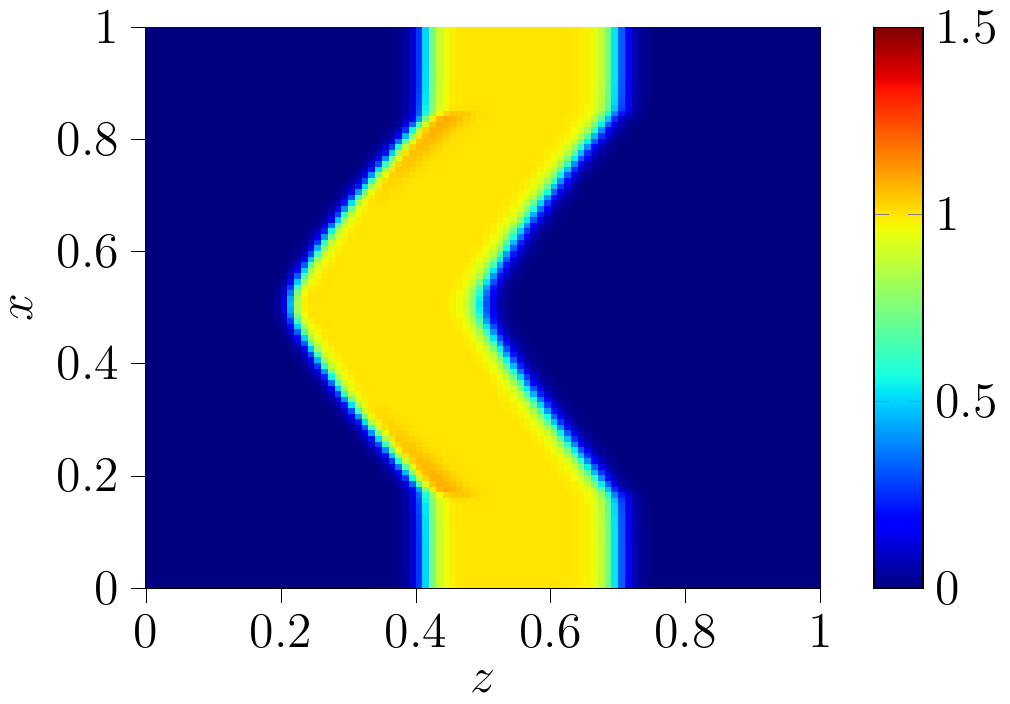}
				\caption{$\rho(x,z,0.5)$  with $\beta=0.5$}
			\end{subfigure}
			\\
			\begin{subfigure}{0.32\linewidth}
				\includegraphics[width=\textwidth]{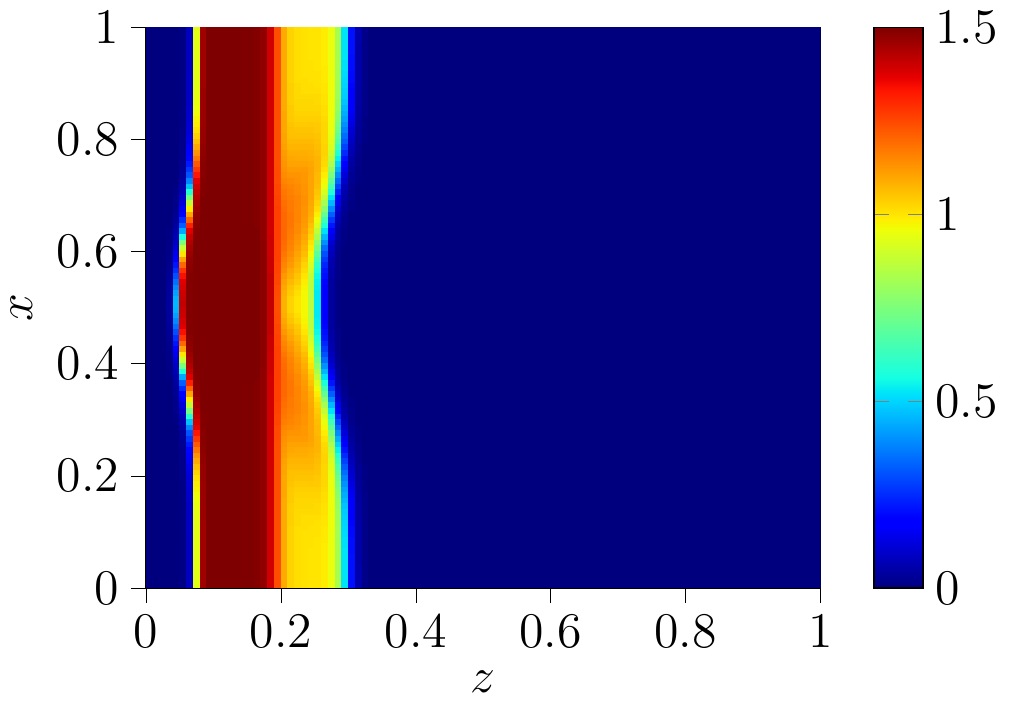}
				\caption{$\rho(x,z,0.1)$  with $\beta=1$}
			\end{subfigure}
			\begin{subfigure}{0.32\linewidth}
				\includegraphics[width=\textwidth]{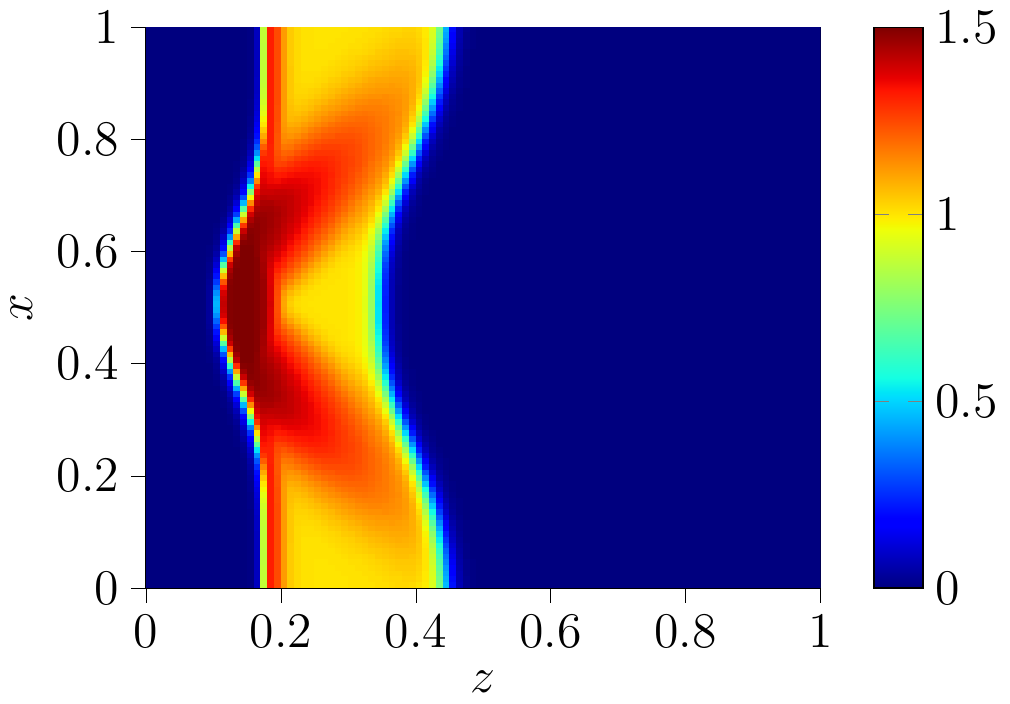}
				\caption{$\rho(x,z,0.25)$  with $\beta=1$}
			\end{subfigure}
			\begin{subfigure}{0.32\linewidth}
				\includegraphics[width=\textwidth]{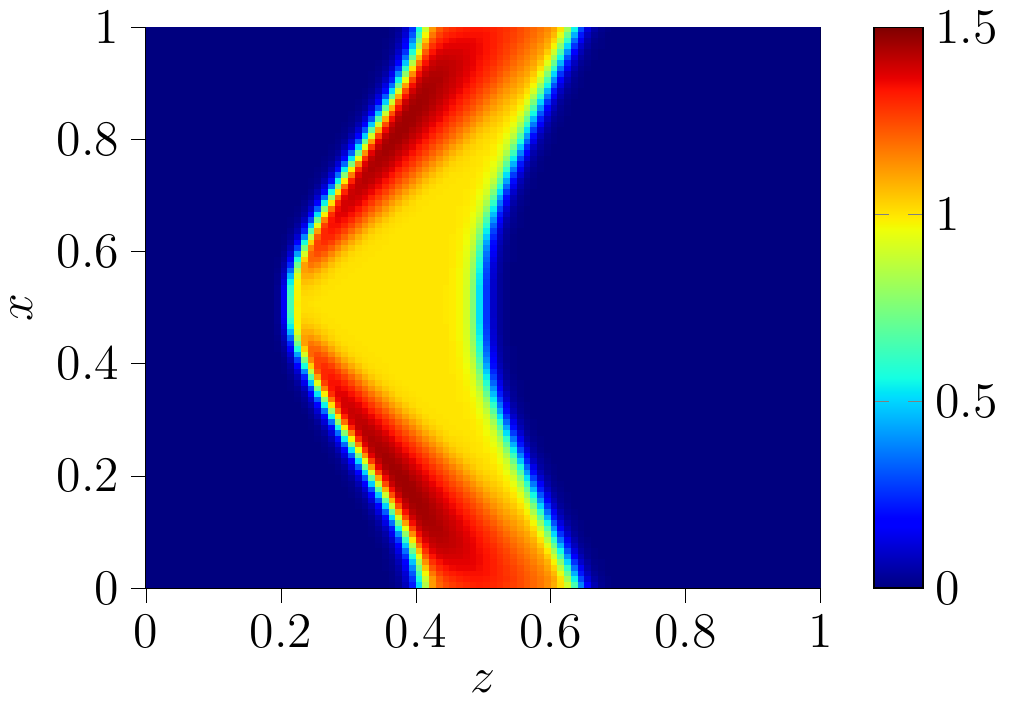}
				\caption{$\rho(x,z,0.5)$  with $\beta=1$}
			\end{subfigure}\\
			\begin{subfigure}{0.32\linewidth}
				\includegraphics[width=\textwidth]{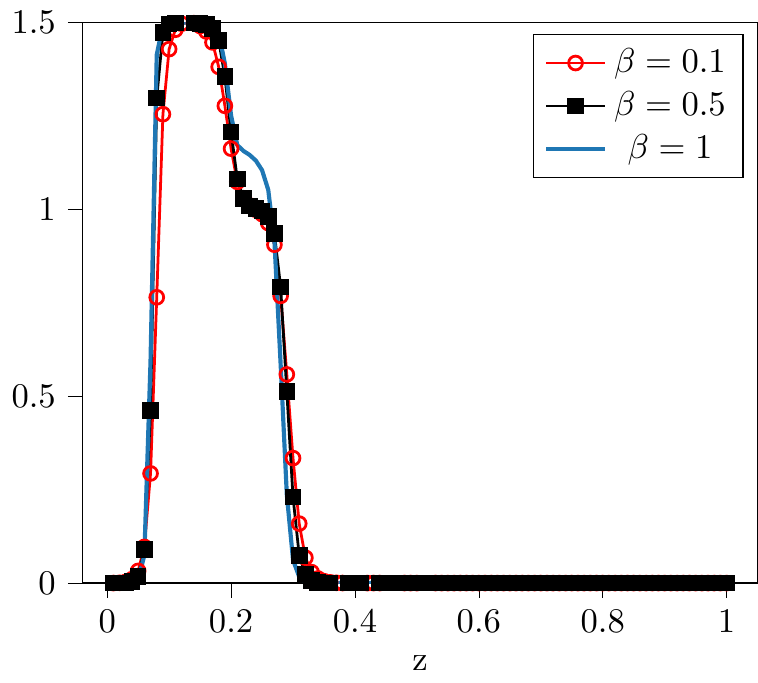}
				\caption{ $\rho(0.3,z,0.1)$}
			\end{subfigure}
			\begin{subfigure}{0.32\linewidth}
				\includegraphics[width=\textwidth]{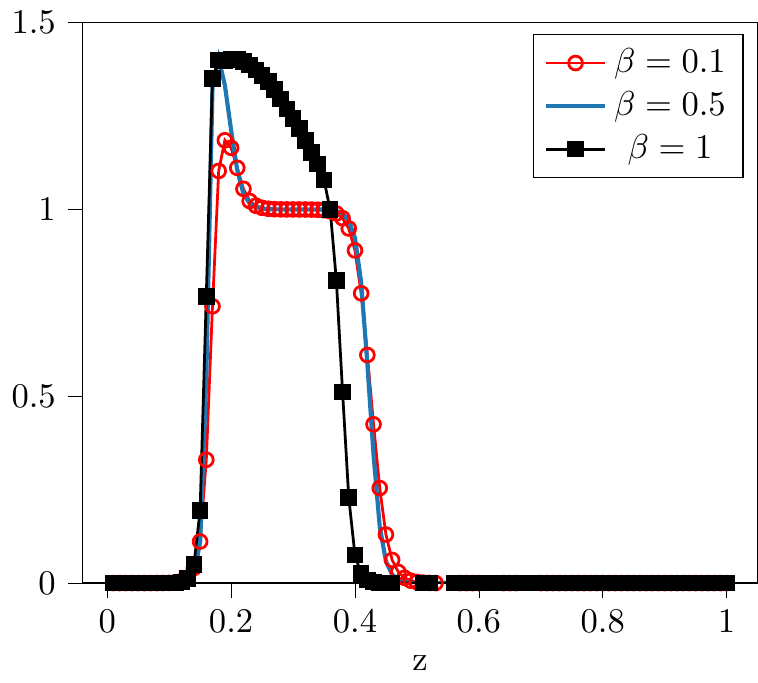}
				\caption{$\rho(0.3,z,0.25)$}
			\end{subfigure}
			\begin{subfigure}{0.32\linewidth}
				\includegraphics[width=\textwidth]{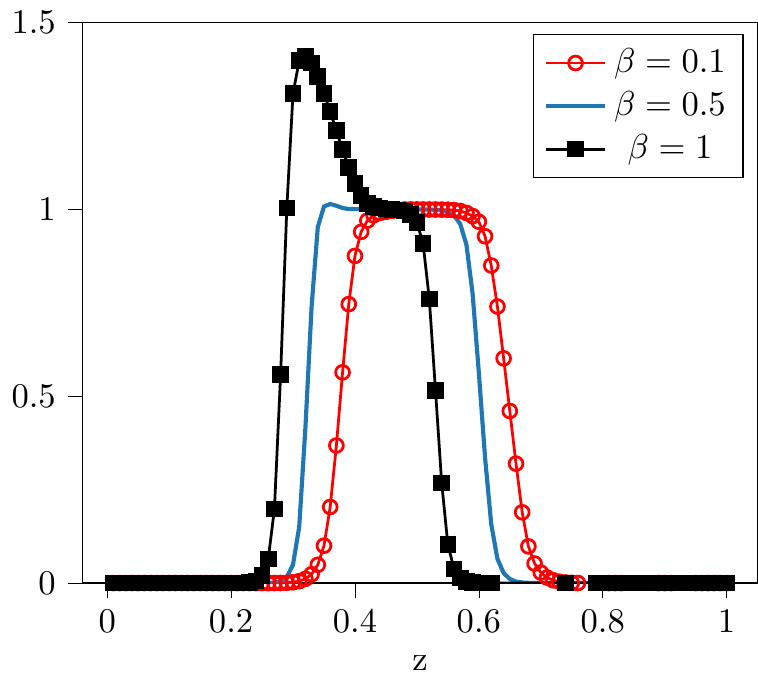}
				\caption{$\rho(0.3,z,0.5)$ }
			\end{subfigure}
			\caption{Plots of the solution $\rho$ from Example \ref{ex:var_beta} for different values of $\beta$.  Larger values of $\beta$ lead to more throttling.}
			\label{fig:betaVar}
		\end{figure}
		%\begin{figure}
		%	\centering
		%	\begin{subfigure}{0.32\linewidth}
		%		\includegraphics[width=\linewidth]{./ODETest2_Beta0_1_early}
		%		\caption{$r(0.1)$  with $\beta=0.1$}
		%	\end{subfigure}
		%	\begin{subfigure}{0.32\linewidth}
		%		\includegraphics[width=\linewidth]{./ODETest2_Beta0_1_mid}
		%		\caption{$rho(0.25)$  with $\beta=0.1$}
		%	\end{subfigure}
		%	\begin{subfigure}{0.32\linewidth}
		%		\includegraphics[width=\linewidth]{./ODETest2_Beta0_1_late}
		%		\caption{$r(0.5)$  with $\beta=0.1$}
		%	\end{subfigure}
		%\\
		%	\begin{subfigure}{0.32\linewidth}
		%	\includegraphics[width=\linewidth]{./ODETest2_Beta0_5_early}
		%	\caption{$r(x0.1)$  with $\beta=0.5$}
		%\end{subfigure}
		%\begin{subfigure}{0.32\linewidth}
		%	\includegraphics[width=\linewidth]{./ODETest2_Beta0_5_mid}
		%	\caption{$r(0.25)$  with $\beta=0.5$}
		%\end{subfigure}
		%\begin{subfigure}{0.32\linewidth}
		%	\includegraphics[width=\linewidth]{./ODETest2_Beta0_5_late}
		%	\caption{$r(0.5)$  with $\beta=0.5$}
		%\end{subfigure}
		%\\
		%	\begin{subfigure}{0.32\linewidth}
		%	\includegraphics[width=\linewidth]{./ODETest2_Beta1_0_early}
		%	\caption{$r(0.1)$  with $\beta=1$}
		%\end{subfigure}
		%\begin{subfigure}{0.32\linewidth}
		%	\includegraphics[width=\linewidth]{./ODETest2_Beta1_0_mid}
		%	\caption{$r(0.25)$  with $\beta=1$}
		%\end{subfigure}
		%\begin{subfigure}{0.32\linewidth}
		%	\includegraphics[width=\linewidth]{./ODETest2_Beta1_0_late}
		%	\caption{$r(0.5)$  with $\beta=1$}
		%\end{subfigure}
		%	\caption{Plots of the solution $r$ from Example \ref{ex:var_beta} for different values of $\beta$.}
		%	\label{fig:ODEbetaVar}
		%\end{figure}
		
	\end{example}

	\begin{example}[Highly localized slowdown]
		In this example, we explore the effects of a highly localized slowdown in processor speed when $\eta=\beta=1$.  The initial and boundary conditions are given in \eqref{eq:etaVarTest}, while the processor speed is given by
		$\alpha(x) = 1-0.4c(x)$, where
		\begin{equation}
		c(x) = \begin{cases}
		0 & |x-.5|>.05\\
		40x-18& x\in[.45,.475]\\
		-40x+22& x\in[.525,.55]\\
		1 & |x-.5|<.025
		\end{cases}.
		\end{equation}
		In particular, $\alpha\neq1$ only on the interval $(0.45,0.55)$.  
		Simulation results from this example are shown in \cref{fig:slowdown}.  At early times, slower processors in the center of the $x$ domain prohibit neighboring processors from moving data to later stages of the calculation (i.e. along the $z$-direction).  The result is a buildup of data in the neighboring processors.  As time progresses, the build-up of data spreads as throttled processors near the initial slowdown around $x=0.5$ begin to effect neighbors further away.  Eventually these buildups dissipate as the slower processors begin catch up with their throttled neighbors.
		\begin{figure}[h!]
			\centering
			\begin{subfigure}{0.32\linewidth}
				\includegraphics[width=\textwidth]{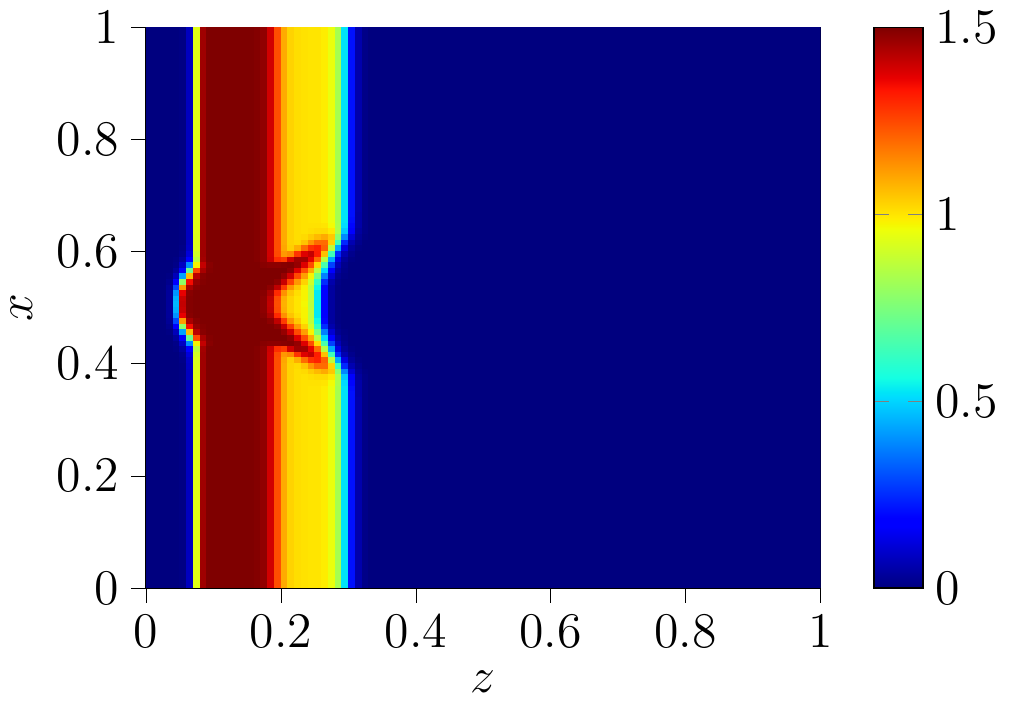}
				\caption{$\rho(x,z,0.1)$}
			\end{subfigure}
			\begin{subfigure}{0.32\linewidth}
				\includegraphics[width=\textwidth]{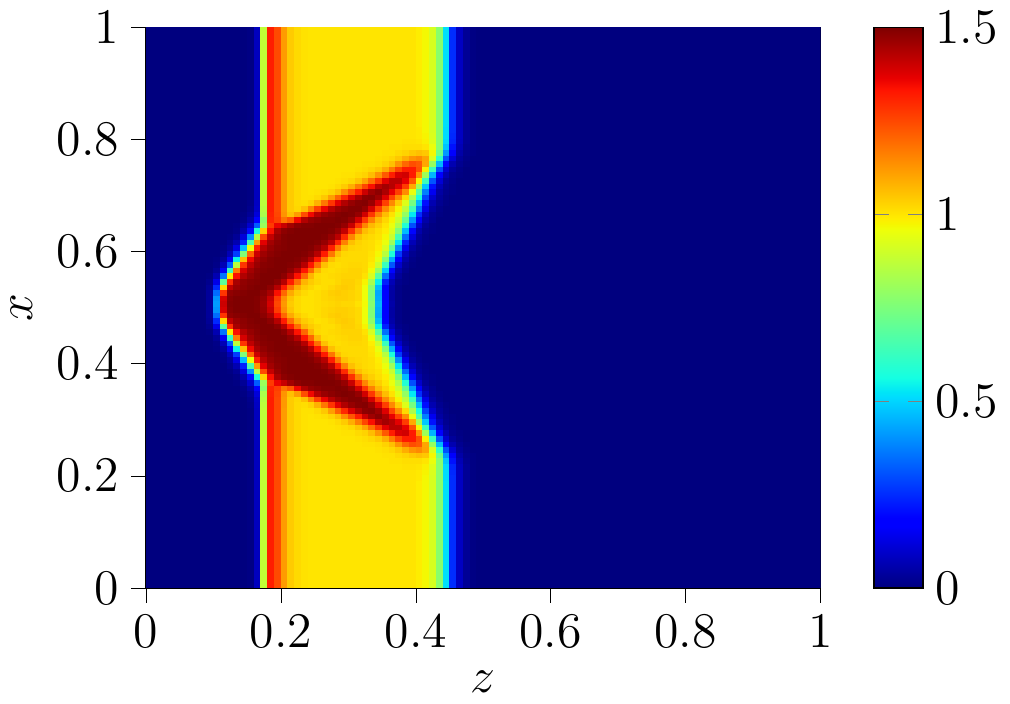}
				\caption{$\rho(x,z,0.25)$}
			\end{subfigure}
			\begin{subfigure}{0.32\linewidth}
				\includegraphics[width=\textwidth]{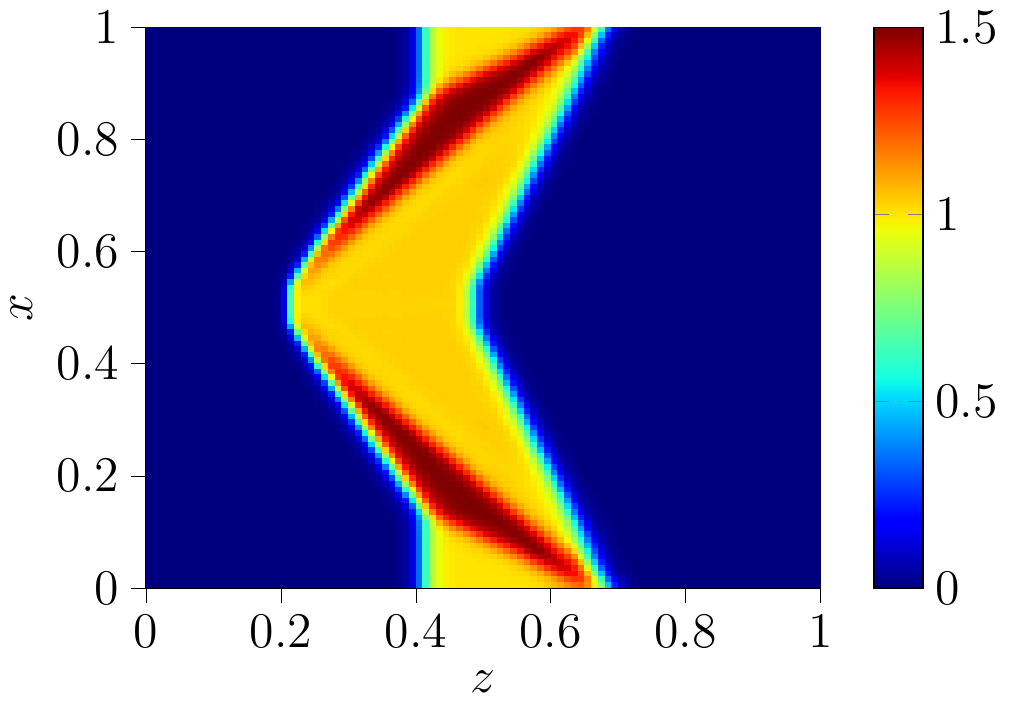}
				\caption{$\rho(x,z,0.5)$}
			\end{subfigure}
			\caption{The effect of a highly localized slowdown on $\rho$}
			\label{fig:slowdown}
		\end{figure}
		%\begin{figure}
		%	\centering
		%		\begin{subfigure}{0.32\linewidth}
		%		\includegraphics[width=\linewidth]{./ODETest4_early}
		%		\caption{$r(0.1)$}
		%	\end{subfigure}
		%	\begin{subfigure}{0.32\linewidth}
		%		\includegraphics[width=\linewidth]{./ODETest4_mid}
		%		\caption{$r(0.25)$}
		%	\end{subfigure}
		%	\begin{subfigure}{0.32\linewidth}
		%		\includegraphics[width=\linewidth]{./ODETest4_late}
		%		\caption{$r(0.5)$}
		%	\end{subfigure}
		%	\caption{The effect of a highly localized slowdown on $r$}
		%	\label{fig:ODEslowdown}
		%\end{figure}
	\end{example}

	\begin{example}[Long-term behavior ]
		In previous examples, we have observed that under some conditions, solutions eventually resemble a traveling profile of the form
		\begin{equation}
		\label{eq:long-time profile}
		\rho_*(x,z,t) = \chi_{[\zeta_0(x), \zeta_1(x)]} (z-st),
		\end{equation}
		where $s$ is a positive constant and the profiles $\xi_0$ and $\xi_1$ are constant in time and satisfy $\zeta_1(x) < \zeta_1(x)$ for all $x \in [0,1)$.  Our intuition is that for a wide range of conditions, traveling profiles are of this type will arise after  sufficiently long times, if the $z$ domain is extended to $(0,\infty)$. Moreover the shape of $\zeta_1$ and $\zeta_2$ is closely related to the initial data and the shape of $\alpha$.
		\footnote{A more systematic study of such profiles in special case can be found in \cite{hauck2019qualitative}.}
		Rather than make a precise conjecture at this point, we instead provide an example which further demonstrates our intuition. 
		Initial and boundary conditions are given in \eqref{eq:etaVarTest}.  Because the domain in $z$ is limited, we introduce relatively small variations in $\alpha$, which allow the system to settle faster:
		\begin{equation}
		\alpha(x) = 1+0.1\cos(4\pi x).
		\end{equation}
		Simulation results for this example are presented in \cref{fig:alphaVar}.  When $t=0.5$, the solution has nearly settled to a profile of the form \eqref{eq:long-time profile}, with cusps that appear where the waves caused by throttling meet, at $x=0.5$ and at the periodic boundary. In particular the solution has the periodicity of $\alpha$.
		\begin{figure}[h!]
			\centering
			\begin{subfigure}{0.32\linewidth}
				\includegraphics[width=\textwidth]{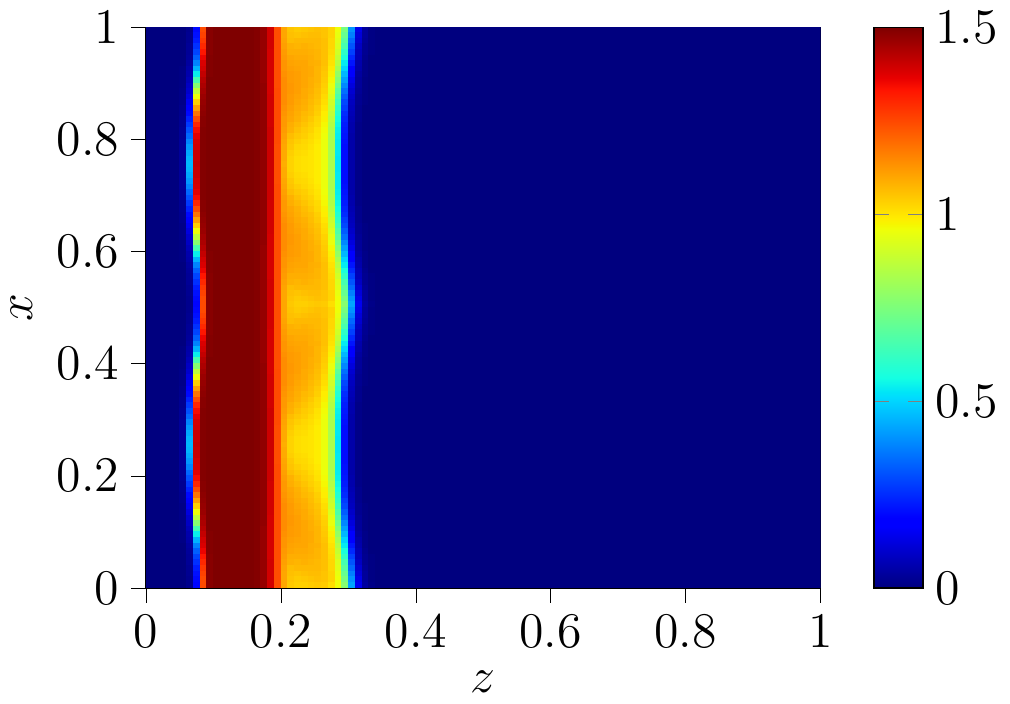}
				\caption{$\rho(x,z,0.1)$}
			\end{subfigure}
			\begin{subfigure}{0.32\linewidth}
				\includegraphics[width=\textwidth]{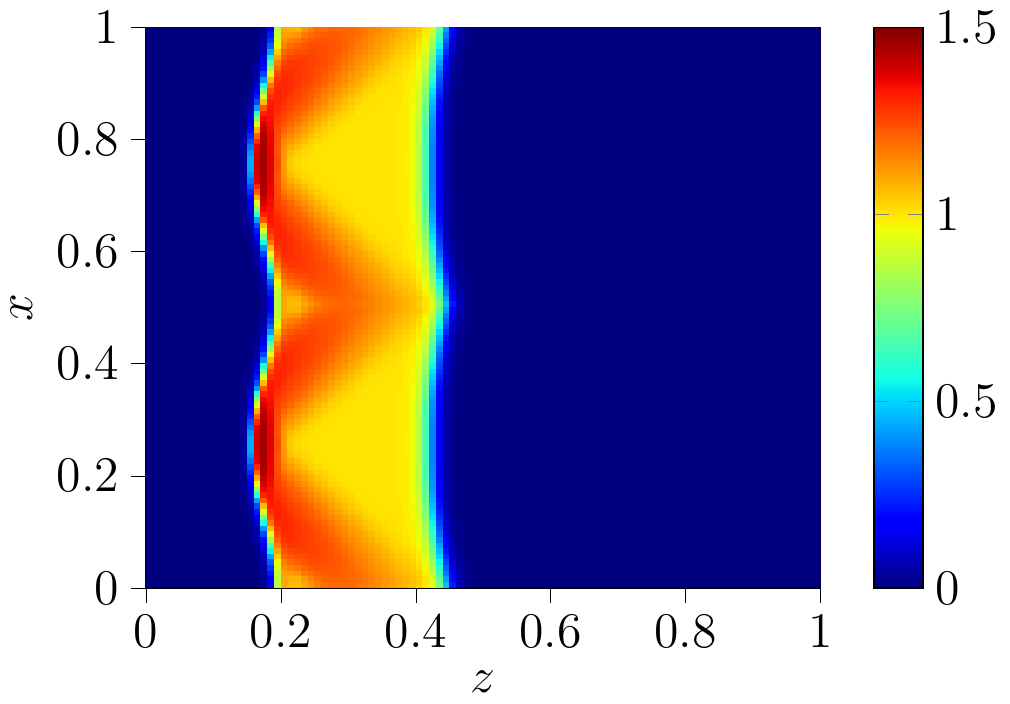}
				\caption{$\rho(x,z,0.25)$}
			\end{subfigure}
			\begin{subfigure}{0.32\linewidth}
				\includegraphics[width=\textwidth]{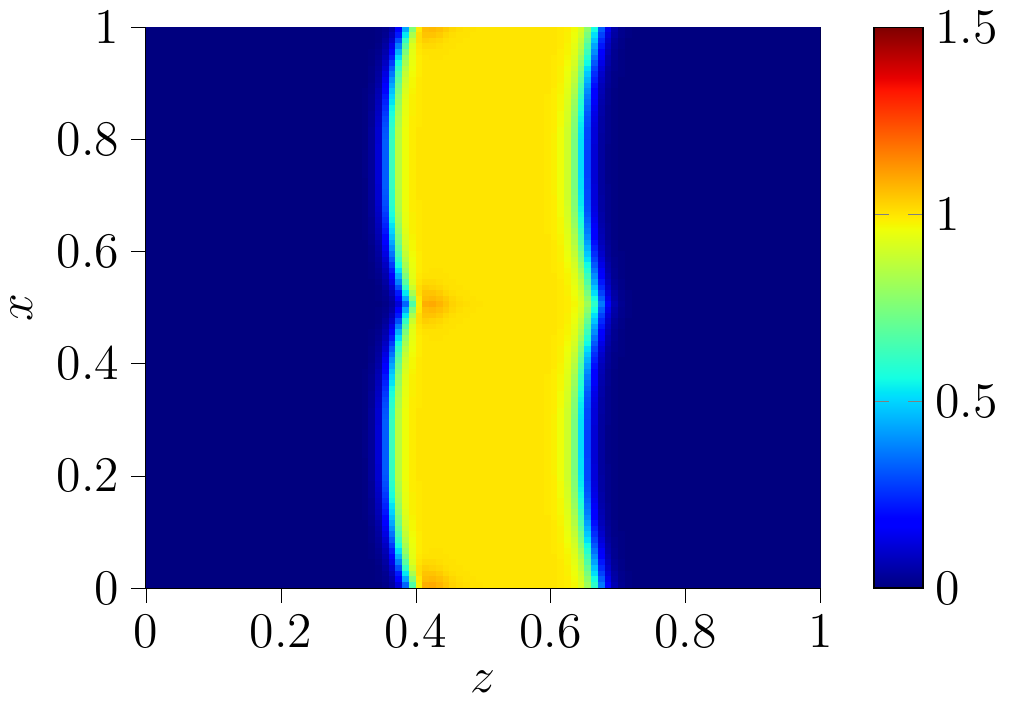}
				\caption{$\rho(x,z,0.5)$}
			\end{subfigure}
			\caption{The effect of small variation in processor speed on $\rho$.  After sufficiently time, a profile emerges with the periodicity of $\alpha$.}
			\label{fig:alphaVar}
		\end{figure}

	\end{example}

\section{Conclusion}
\label{sec:conclusion}

We have presented a simple discrete model of a network of processors in a high performance computing environment where the computational throughput depends on the on the availability of data from neighboring processors.  
This discrete, microscopic-level model has been then used to derive a continuum-level model which treats computational progress as an Eulerian fluid flow.
Currently, the existence and uniqueness of solutions to the partial differential equation in this fluid model is open.  
However, a Hamilton-Jacobi model is available for which we can establish the existence and uniqueness of continuous viscosity solutions; the solution for the governing equation corresponds to the total amount of data that has been processed through a particular stage in the computation. 
Numerical experiments have shown that this continuum model can capture the asymptotic behavior of the discrete model.  
Additionally, we have used these experiments to give an initial understanding of solutions' dependence on parameters associated with the parallelism of the modeled computation as well as the effects heterogeneities in processing capacity.

In future work, we intend to explore control strategies for $\alpha$ that can alleviate bottlenecks caused by local slowdowns in the processor speed.  We will also extend the model to allow for more complicated interactions, including stochastic effects, and explore strategies for optimal communication.  Finally, we hope to tune the parameters of the model with data taken from processor components of a real supercomputer and then compare predictions of the macroscopic model with the real global behavior of the supercomputer.

% Information that is shared between the article and the supplement
% (title and author information, macros, packages, etc.) goes into
% ex_shared.tex. If there is no supplement, this file can be included
% directly.

% Optional PDF information
\ifpdf
\hypersetup{
	pdftitle={An Example Article},
	pdfauthor={D. Doe, P. T. Frank, and J. E. Smith}
}
\fi

% The next statement enables references to information in the
% supplement. See the xr-hyperref package for details.

%\externaldocument{ex_supplement}

% FundRef data to be entered by SIAM
%<funding-group specific-use="FundRef">
%<award-group>
%<funding-source>
%<named-content content-type="funder-name"> 
%</named-content> 
%<named-content content-type="funder-identifier"> 
%</named-content>
%</funding-source>
%<award-id> </award-id>
%</award-group>
%</funding-group>

%	\input{appendix}
	\appendix
	
	\section{Hamilton-Jacobi Theory}
	\label{sec:appendix}
	
	We recall a few standard definitions from the theory of nonlinear second-order Hamilton-Jacobi equations as in, for instance, \cite{CraIshLio92,Bar13}:
	\begin{definition}[Degenerately elliptic function]
		Let $F:\R^n\times\R^n\times\mathcal{S}^n\rightarrow\R$ be given, where $\mathcal{S}^n$ is the set of symmetric $n\times n$ matrices.  Then we say $F$ is degenerately elliptic if $F(x,r,p,X)\leq F(x,r,p,Y)$ whenever $Y\leq X.$ %Here $Y\leq X$ means $\langle p,Yp\rangle\leq\langle p,Xp\rangle$ for every $p\in\R^N$. 
	\end{definition} 
	\begin{definition}[Modulus function]
		We call a function $\sigma:[0,\infty)\rightarrow[0,\infty)$ a {modulus function} if $\sigma(0)=0$ and it is nondecreasing.
	\end{definition}  
	In contrast to the notational convention in \cref{sec:continuum}, we follow in this appendix the convention of the viscosity literature and refer to time- and space-dependent functions as $u(t,x)$.
	\begin{definition}[Parabolic boundary]
		If $U=(0,T]\times D$ where $D\subset\R^n$ and $T\geq 0$, then $\p_PU:=\{0\}\times D\cup[0,T]\times \p D$ is called the parabolic boundary of $U$.
	\end{definition}
	\begin{definition}[Semicontinuous envelope]
		The upper (respectively, lower) semicontinuous envelopes of a function $u:V\rightarrow[-\infty,\infty]$ are 
		\begin{align*}
		u^*(x)&=\lim_{r\downarrow0}\sup\{u(y):y\in V,|y-x|\leq r\},\\
		u_*(x)&=\lim_{r\downarrow0}\inf\{u(y):y\in V,|y-x|\leq r\}.
		\end{align*}
		They are, respectively, the smallest upper semicontinuous function greater than $u$ and the largest lower semicontinuous function less than $u$.  
	\end{definition} 
	\begin{definition}[Viscosity solutions]
		Let $f:U\times\R\times\R^n\times \mathcal{S}^n\rightarrow\R$ be given.
		An upper (resp. lower ) semicontinuous function $u$ is a viscosity subsolution (resp. supersolution) of  
		\begin{subequations}
			\label{eq:gen_HJ}
			\begin{align}
			u_t+f(t,x,u,\nabla_x u,\nabla_x^2 u) &=0,\qquad (t,x) \in U,\label{eq:gen_HJ_int}\\
			h(t,x,u,\nabla_xu)&=0,\qquad (t,x)\in \p_PU,\label{eq:gen_HJ_BC}
			\end{align}
		\end{subequations}
		on $(0,T]\times D$ if at every $(t,x)\in(0,T]\times D$,  when $u-\psi$ is locally maximized (resp. minimized) at $(t,x)$ and $\psi$ is $C^2( (0,T]\times D)$ we have
		\begin{subequations}
			\label{eq:sub_soln}
			\begin{align}
			\psi_t(t,x)+f(t,x,u,\nabla_x\psi(t,x),\nabla_x^2\psi(t,x))&\leq 0, \quad (t,x)\in U \label{eq:sub_soln_int}\\
			\min\{\psi_t(t,x)+f(t,x,u,\nabla_x\psi(t,x),\nabla_x^2\psi(t,x)),
			\label{eq:sub_soln_bndry}\\
			h(t,x,u(t,x),\nabla_x\psi(t,x))\}&\leq 0, \quad (t,x)\in\p_PU \nonumber
			\end{align}
		\end{subequations}
		(respectively,
		\begin{subequations}
			\label{eq:super_soln}
			\begin{align}
			\psi_t(t,x)+f(t,x,u,\nabla_x\psi(t,x),\nabla_x^2\psi(t,x))&\geq 0,\quad (t,x)\in U
			\label{eq:super_soln_int}\\
			\max\{\psi_t(t,x)+f(t,x,u,\nabla_x\psi(t,x),\nabla_x^2\psi(t,x)),
			\label{eq:super_soln_bndry}\\
			h(t,x,u(t,x),\nabla_x\psi(t,x))\}&\geq 0, \quad (t,x)\in \p_PU. \nonumber
			\end{align}
		\end{subequations}
		A function $u$ is a viscosity solution of \eqref{eq:gen_HJ} if its upper semicontinuous envelope is a viscosity subsolution and its lower semicontinuous envelope is a viscosity supersolution.
	\end{definition}
	%\begin{definition}
	%	A lower-semicontinuous function $u$ is a viscosity supersolution of \eqref{eq:gen_HJ}-\eqref{eq:gen_HJ_BC} if at every $(t,x)\in(0,T]\times D$,  when $u-\phi$ is locally minimized at $(t,x)$ and $\phi$ is $C^2( (0,T]\times D)$ we have
	%	\begin{align*}
	%	\phi_t(t,x)+f(t,x,u,D_x\phi(t,x),D_x^2\phi(t,x))&\geq 0\qquad (t,x)\in (0,T]\times D\\
	%	\min(\phi_t(t,x)+f(t,x,u,D_x\phi(t,x),D_x^2\phi(t,x)),H(t,x,u(t,x),D_x\phi(t,x)))&\geq 0 \qquad (t,x)\in (0,T]\times\p_D.
	%	\end{align*}
	%\end{definition}
	%\begin{definition}
	%	A function $u$ is   
	%\end{definition}
	
	Next, we recall the following general comparison theorem, which is Theorem 4.1 of \cite{GigGotIsh91}.
	\begin{lemma}[Generalized comparison principle]\label{thm:Gig_comp}
		Consider the system \cref{eq:gen_HJ} on $U = (0,T)\times \Omega$ where $\Omega\subseteq\R^n$ is a possibly unbounded domain and $T>0$.  Assume that $f$ satisfies the following assumptions:
		\begin{enumerate}
			\item $f$ is continuous on $U\times\R\times(\R^n\setminus\{0\})\times \mathcal{S}^n$.
			\item $f$ is degenerately elliptic.
			\item $-\infty<f_*(t,x,r,0,O)=f^*(t,x,r,0,O)<\infty$ for all $(t,x,r)\in U\times \R$, where $O$ is the zero matrix.
			\item For every $R>0$, we have 
			\begin{equation}
			\sup\{|f(t,x,r,p,X)|:|p|,|X|\leq R,(t,x,r,p,S)\in U\times\R\times(\R^n\setminus\{0\})\times \mathcal{S}^n\}
			\end{equation}
			is finite.
			\item For every $H>0$, there is a constant $c_0$ such that $r\mapsto f(t,x,r,p,X) +c_0r$ is nondecreasing for all $(t,x,r,p,X)\in U\times\R\times(\R^n\setminus\{0\})\times \mathcal{S}^n$ with $|r|\leq H$.  
			\item For every $R>\rho>0$ there is a modulus function $\sigma_{R_\rho}$ such that 
			\begin{equation}
			|f(t,x,r,p,X)-f(t,x,r,q,Y)|\leq \sigma_{R_\rho}(|p-q|+|X-Y|)
			\end{equation} for $(t,x,r)\in U\times\R,$ $\rho\leq|p|,~|q|\leq R$, and $|X|,|Y|\leq R$.
			\item There is a constant $\rho_0>0$ and a modulus function $\sigma_1$ such that 
			\begin{align}
			f^*(t,x,r,p,X)-f^*(t,x,r,0,O)&\leq\sigma_1(|p|+|X|)\\
			f_*(t,x,r,p,X)-f_*(t,x,r,0,O)&\geq-\sigma_1(|p|+|X|)
			\end{align}
			for $(t,x,r)\in U\times \R$ and $|p|,|X|\leq\rho_0$.
			\item There is a modulus function $\sigma_2$ such that 
			\begin{equation}
			|f(t,x,r,p,X)-f(t,y,r,p,X)|\leq \sigma_2(|x-y|(|p|+1))
			\end{equation}
			for any $y\in \Omega$ and $(t,x,r,p,X)\in U\times\R\times(\R^n\setminus\{0\})\times \mathcal{S}^n.$ 
		\end{enumerate} 
		Then if $u^-$ and $u^+$ are viscosity subsolutions and supersolutions of \eqref{eq:gen_HJ}, respectively, such that for some $K>0$ independent of $t,x,y\in(0,T]\times\Omega\times\Omega$:
		\begin{itemize}
			\item $u^-(t,x)\leq K(|x|+1)$ and $u^+(t,x)\geq -K(|x|+1);$
			\item $(u^-)^*(t,x)-(u^+)_*(t,y)\leq m_T(|x-y|)$ on $\p_p\big( (0,T]\times(\Omega\times\Omega)\big);$
			\item $(u^-)^*(t,x)-(u^+)_*(t,y)\leq K(|x-y|+1)$ on $\p_p\big( (0,T]\times(\Omega\times\Omega)\big).$
		\end{itemize}  
		Then there is a modulus function $\sigma$ such that
		\begin{equation}
		\label{eq:comparison_formula}
		(u^-)^*(t,x)-(u^+)_*(t,y)\leq \sigma(|x-y|) .
		\end{equation}
		
	\end{lemma}
	This generalized comparison principle, coupled with an argument which uses the framework given in \cite{CraIshLio92}, known as Perron's method, gives the existence and uniqueness of a viscosity solution to \eqref{eq:HJ_prelim}.  Specifically, we note the parabolic version of this framework uses a result like the following, which is Lemma 2.3.15 from \cite{ImbSil13}
	\begin{lemma}[Perron process for parabolic equations]\label{lem:max_subsol}
		Consider \eqref{eq:gen_HJ} where $f$ is degenerate elliptic and continuous.  Assume that $u^+$ and $u^-$ are viscosity supersolutions and subsolutions, respectively.  Then there exists a viscosity solution $u$ such that $u^-\leq u\leq u^+$.
	\end{lemma}
	
	The results above are summarized in the following theorem.
	
	\begin{theorem}[Unique viscosity solution]
		\label{thm:visc_solution}
		Suppose that $f$ satisfies the conditions of \cref{thm:Gig_comp} and \cref{lem:max_subsol} and that there exists a viscosity supersolution $u^+$ and a viscosity subsolution $u^-$ to \eqref{eq:gen_HJ}.  Then there exists a unique continuous viscosity solution to \eqref{eq:gen_HJ}.
	\end{theorem}
	
	\begin{proof}
		According to \cref{lem:max_subsol}, there exists a viscosity solution $u$ to \eqref{eq:gen_HJ}.  To show uniqueness and continuity, let $v$ be another viscosity solution.  By definition, $u^*$ and $v^*$ are subsolutions and $u_*$ and $v_*$ are supersolutions.  Then \eqref{eq:comparison_formula}, combined with the properties of envelopes imply that
		\begin{equation}
		u^* = (u^*)^* \leq (v_*)_* = v_* \leq v^* = (v^*)^* \leq (u_*)_* = u_* \leq u^*.
		\end{equation}
		Thus $u = v$ and $u^*= u_*$ so that $u$ is continuous.
	\end{proof}

	\section*{Acknowledgments}
	C.D.H and R.C.B. would like to thank Michael Herty for many helpful discussions.
	
	\bibliographystyle{siamplain}
	\bibliography{references}

\begin{thebibliography}{10}

\bibitem{TOP500}
{\em Top500 list}.
\newblock \url{https://www.top500.org/lists/2019/06/}, June 2019.

\bibitem{ArmMarRin03}
{\sc D.~Armbruster, D.~Marthaler, and C.~Ringhofer}, {\em Kinetic and fluid
  model hierarchies for supply chains}, Multiscale Modeling \& Simulation, 2
  (2003), pp.~43--61.

\bibitem{AwRas00}
{\sc A.~Aw and M.~Rascle}, {\em Resurrection of "second order" models of
  traffic flow}, SIAM Journal on Applied Mathematics, 60 (2000), pp.~916--938.

\bibitem{BanHerKla06}
{\sc M.~K. Banda, M.~Herty, and A.~Klar}, {\em Gas flow in pipeline networks},
  Networks and Heterogeneous Media, 1 (2006), pp.~41--56.

\bibitem{Bar93}
{\sc G.~Barles}, {\em Fully non-linear neumann type boundary conditions for
  second-order elliptic and parabolic equations}, Journal of Differential
  Equations, 106 (1993), pp.~90 -- 106.

\bibitem{Bar13}
{\sc G.~Barles}, {\em An introduction to the theory of viscosity solutions for
  first-order hamilton--jacobi equations and applications}, in Hamilton-Jacobi
  Equations: Approximations, Numerical Analysis and APplications, P.~Loreti and
  N.~A. Tchou, eds., Springer Berlin Heidelberg, 2013.

\bibitem{BiaTon12}
{\sc S.~Bianchini and D.~Tonon}, {\em {SBV} regularity for {Hamilton}--{Jacobi}
  equations with {Hamiltonian} depending on $(t, x)$}, SIAM Journal on
  Mathematical Analysis, 44 (2012), pp.~2179--2203.

\bibitem{BroGasHer11}
{\sc J.~Brouwer, I.~Gasser, and M.~Herty}, {\em Gas pipeline models revisited:
  Model hierarchies, nonisothermal models, and simulations of networks},
  Multiscale Modeling \& Simulation, 9 (2011), pp.~601--623.

\bibitem{CanFra14}
{\sc P.~Cannarsa and H.~Frankowska}, {\em From pointwise to local regularity
  for solutions of {Hamilton}--{Jacobi} equations}, Calculus of Variations and
  Partial Differential Equations, 49 (2014), pp.~1061--1074.

\bibitem{Coc03}
{\sc B.~Cockburn}, {\em Continuous dependence and error estimation for
  viscosity methods}, Acta Numerica, 12 (2003), pp.~127--180.

\bibitem{CraIshLio92}
{\sc M.~G. Crandall, H.~Ishii, and P.-L. Lions}, {\em User's guide to viscosity
  solutions of second order partial differential equations}, Bulletin of the
  American Mathematical Society, 27 (1992), pp.~1--67.

\bibitem{CulKarPat93}
{\sc D.~Culler, R.~Karp, D.~Patterson, A.~Sahay, K.~E. Schauser, E.~Santos,
  R.~Subramonian, and T.~von Eicken}, {\em Logp: Towards a realistic model of
  parallel computation}, SIGPLAN Not., 28 (1993), pp.~1--12.

\bibitem{CulKarPat96}
{\sc D.~E. Culler, R.~M. Karp, D.~Patterson, A.~Sahay, E.~E. Santos, K.~E.
  Schauser, R.~Subramonian, and T.~von Eicken}, {\em Logp: A practical model of
  parallel computation}, Commun. ACM, 39 (1996), pp.~78--85.

\bibitem{DeeVahJuv15}
{\sc E.~Deelman, K.~Vahi, G.~Juve, M.~Rynge, S.~Callaghan, P.~J. Maechling,
  R.~Mayani, W.~Chen, R.~F. da~Silva, M.~Livny, and K.~Wenger}, {\em Pegasus, a
  workflow management system for science automation}, Future Generation
  Computer Systems, 46 (2015), pp.~17 -- 35,
  \href{http://dx.doi.org/https://doi.org/10.1016/j.future.2014.10.008}{doi:\nolinkurl{https://doi.org/10.1016/j.future.2014.10.008}},
  \url{http://www.sciencedirect.com/science/article/pii/S0167739X14002015}.

\bibitem{DonHitBel14}
{\sc J.~Dongarra, J.~Hittinger, J.~Bell, L.~Chac\'{o}n, R.~Falgout, M.~Heroux,
  P.~Hovland, E.~Ng, C.~Webster, and S.~Wild}, {\em Applied mathematics
  research for exascale computing}, tech. report, U.S. Department of Energy,
  Office of Science, Advanced Scientific Computing Research Program, 2014.

\bibitem{DosBarDoe14}
{\sc S.~Dosanjh, R.~Barrett, D.~Doerfler, S.~Hammond, K.~Hemmert, M.~Heroux,
  P.~Lin, K.~Pedretti, A.~Rodrigues, T.~Trucano, and J.~Luitjens}, {\em
  Exascale design space exploration and co-design}, Future Generation Computer
  Systems, 30 (2014), pp.~46 -- 58.

\bibitem{GigGotIsh91}
{\sc Y.~Giga, S.~Goto, H.~Ishii, and M.-H. Sato}, {\em Comparison principle and
  convexity preserving properties for singular degenerate parabolic equations
  on unbounded domains}, Indiana Univ. Math. J., 40 (1991), pp.~443--470.

\bibitem{HaiNorWan93}
{\sc E.~Hairer, S.~P. N{\o}rsett, and G.~Wanner}, {\em Solving Ordinary
  Differential Equations I: Nonstiff Problems}, vol.~8 of Springer Series in
  Computational Mathematics, Springer-Verlag Berlin Heidelberg, 2010.

\bibitem{hauck2019qualitative}
{\sc C.~Hauck, M.~Herty, and G.~Visconti}, {\em Qualitative properties of
  mathematical model for data flow}.
\newblock submitted.

\bibitem{ImbSil13}
{\sc C.~Imbert and L.~Silvestre}, {\em An introduction to fully nonlinear
  parabolic equations}, in An Introduction to the K\"{a}hler-Ricci Flow,
  S.~Boucksom, P.~Eyssidieux, and V.~Guedj, eds., vol.~2086, Springer
  International Publishing, 2013, pp.~7--88.

\bibitem{JakKar02}
{\sc E.~R. Jakobsen and K.~H. Karlsen}, {\em Continuous dependence estimates
  for viscosity solutions of fully nonlinear degenerate parabolic equations},
  Journal of Differential Equations, 183 (2002), pp.~497 -- 525.

\bibitem{KlaWeg99a}
{\sc A.~Klar and R.~Wegener}, {\em A hierarchy of models for multilane
  vehicular traffic. {I}. {M}odeling}, SIAM J. Appl. Math., 59 (1999),
  pp.~983--1001 (electronic).

\bibitem{Kun13}
{\sc J.~M. Kunkel}, {\em Simulating parallel programs on application and system
  level}, Computer Science - Research and Development, 28 (2013), pp.~167--174.

\bibitem{NunFerFil12}
{\sc A.~Nunez, J.~Fernandez, R.~Filguiera, F.~Garcia, and J.~Carretero}, {\em
  Simcan: A flexible, scalable and expandable simulation platform for modelling
  and simulating distributed architectures and applications}, Simulation
  Modelling Practice and Theory, 20 (2012), pp.~12--32.

\bibitem{PiePuc97}
{\sc A.~Pietracaprina and G.~Pucci}, {\em The complexity of deterministic pram
  simulation on distributed memory machines}, Theory of Computing Systems, 30
  (1997), pp.~231--247,
  \href{http://dx.doi.org/10.1007/BF02679461}{doi:\nolinkurl{10.1007/BF02679461}},
  \url{http://dx.doi.org/10.1007/BF02679461}.

\bibitem{Rif08}
{\sc L.~Rifford}, {\em On viscosity solutions of certain {Hamilton}--{Jacobi}
  equations: Regularity results and generalized sard's theorems},
  Communications in Partial Differential Equations, 33 (2008), pp.~517--559.

\bibitem{shu2007high}
{\sc C.-W. Shu}, {\em High order numerical methods for time dependent
  hamilton-jacobi equations}, in Mathematics and computation in imaging science
  and information processing, World Scientific, 2007, pp.~47--91.

\bibitem{SpaVet12}
{\sc K.~L. Spafford and J.~S. Vetter}, {\em Aspen: A domain specific language
  for performance modeling}, in Proceedings of the International Conference on
  High Performance Computing, Networking, Storage and Analysis, SC '12, Los
  Alamitos, CA, USA, 2012, IEEE Computer Society Press, pp.~84:1--84:11.

\bibitem{SteWhi09}
{\sc R.~Stevens and A.~White}, {\em Architectures and technology for extreme
  scale computing}, in ASCR Scientific Grand Challenges Workshop Series, 2009.

\bibitem{Tes12}
{\sc G.~Teschl}, {\em Ordinary Differential Equations and Dynamical Systems},
  vol.~140 of Graduate Studies in Mathematics, American Mathematical Society,
  2012.

\bibitem{VetMer15}
{\sc J.~S. Vetter and J.~S. Meredith}, {\em Synthetic program analysis with
  aspen}, in Proceedings of the 3rd International Conference on Exascale
  Applications and Software, Edinburgh, Scotland, UK, 04/2015 2015, University
  of Edinburgh, University of Edinburgh.

\bibitem{Wal98}
{\sc W.~Walter}, {\em Ordinary Differential Equations}, vol.~182 of Graduate
  Texts in Mathematics, Springer-Verlag New York, 1998.

\end{thebibliography}

\end{document}